\documentclass[acmsmall,screen,nonacm, anonymous=false, review=false]{acmart} % review=true and anoymous=true when submit for review
\usepackage{booktabs}
\usepackage{subcaption}
\usepackage{tikz}
\usepackage{amsfonts}
\usepackage{amsmath}

\usepackage{amssymb}
\usepackage{physics}
\usepackage{xspace}
\usepackage{appendix}
\usepackage{cleveref}
\usepackage[noend]{algpseudocode}
\usepackage{algorithm}
\usepackage{color}
\usepackage{graphicx}
\usepackage{multirow}
\usepackage{soul}
\usepackage{booktabs}
\usepackage{threeparttable}
\usepackage{balance}
\usepackage{xfrac}
\usepackage{rotating,tabularx}
\usepackage{rotating}
\usepackage{diagbox}
\usepackage{svg}
\usepackage[clean]{svg-extract}
\usepackage{soul}
\usepackage{amsthm}
\usepackage{balance}
\usepackage{wrapfig}
\newtheorem{theorem}{Theorem}
\usepackage{natbib}
\setcitestyle{authoryear,open={(},close={)}} %Citation-related commands

\newcommand{\asplossubmissionnumber}{299}
\newcommand{\sys}{Quarl\xspace}
\newcommand{\qiskit}{Qiskit\xspace}
\newcommand{\quilc}{Quilc\xspace}

\newcommand{\tket}{t$|$ket$\rangle$\xspace}
\newcommand{\voqc}{\textproc{voqc}\xspace}
\newcommand{\er}[1]{\mbox{\rm\em #1}}
\newcommand{\m}[1]{\mathcal{#1}}

\newcommand{\critic}{gate selector\xspace}
\newcommand{\actor}{transformation selector\xspace}
\newcommand{\tof}{{\tt barenco\_tof\_3}\xspace}
\newcommand{\local}{\textproc{IG}\xspace}
\newcommand{\mdp}{Markov decision process\xspace}

\newcommand{\ZJ}[1]{{\textcolor{red}{ZJ: #1}}}
\newcommand{\yw}[1]{{\color{orange} [Yi: #1]}}
\newcommand{\TODO}[1]{{\textcolor{red}{TODO: #1}}}
\newcommand{\oded}[1]{{\textcolor{blue}{OP: #1}}}
\newcommand{\ZL}[1]{{\textcolor{cyan}{Zikun: #1}}}
\newcommand{\JP}[1]{{\textcolor{violet}{Jinjun: #1}}}
\newcommand{\YM}[1]{{\textcolor{magenta}{Yixuan: #1}}}

\newcommand{\commentout}[1]{\ifdim\lastskip>0pt\ignorespaces\fi}
\iftrue
  \renewcommand{\ZJ}[1]{\ifdim\lastskip>0pt\ignorespaces\fi}
  \renewcommand{\yw}[1]{\ifdim\lastskip>0pt\ignorespaces\fi}
 \renewcommand{\TODO}[1]{\ifdim\lastskip>0pt\ignorespaces\fi}
 \renewcommand{\oded}[1]{\ifdim\lastskip>0pt\ignorespaces\fi}
 \renewcommand{\ZL}[1]{\ifdim\lastskip>0pt\ignorespaces\fi}
 \renewcommand{\JP}[1]{\ifdim\lastskip>0pt\ignorespaces\fi}
 \renewcommand{\YM}[1]{\ifdim\lastskip>0pt\ignorespaces\fi}
\fi

\newtheorem{definition}{Definition}

\begin{document}

\title{\sys: A Learning-Based Quantum Circuit Optimizer}

% \commentout{
\author{Zikun Li}
\affiliation{\institution{Carnegie Mellon University}
\city{Pittsburgh}
\state{PA}
\country{USA}}
\email{zikunl@andrew.cmu.edu}

\author{Jinjun Peng}
\affiliation{\institution{Tsinghua University}
\city{Beijing}
%\state{CA}
\country{China}}
\email{mail@co1in.me}

\author{Yixuan Mei}
\affiliation{\institution{Tsinghua University}
\city{Beijing}
%\state{CA}
\country{China}}
\email{meiyixuan2000@gmail.com}

\author{Sina Lin}
\affiliation{\institution{Microsoft}
\city{Mountain View}
\state{CA}
\country{USA}}
\email{silin@microsoft.com}

\author{Yi Wu}
\affiliation{\institution{Tsinghua University}
\city{Beijing}
%\state{CA}
\country{China}}
% \email{jxwuyi@mail.tsinghua.edu.cn}
\email{jxwuyi@gmail.com}

\author{Oded Padon}
\affiliation{\institution{VMware Research}
\city{Palo Alto}
\state{CA}
\country{USA}}
\email{oded.padon@gmail.com}

\author{Zhihao Jia}
\affiliation{\institution{Carnegie Mellon University}
\city{Pittsburgh}
\state{PA}
\country{USA}}
\email{zhihao@cmu.edu}
\date{}
% }

\thispagestyle{empty}

\begin{abstract}

Optimizing quantum circuits is challenging due to the very large search space of functionally equivalent circuits and the necessity of applying transformations that temporarily decrease performance to achieve a final performance improvement. This paper presents \sys, a learning-based quantum circuit optimizer. 
Applying reinforcement learning (RL) to quantum circuit optimization raises two main challenges: the large and varying action space and the non-uniform state representation.
\sys addresses these issues with a novel neural architecture and RL-training procedure. Our neural architecture decomposes the action space into two parts and leverages graph neural networks in its state representation, both of which are guided by the intuition that optimization decisions can be mostly guided by local reasoning while allowing global circuit-wide reasoning.
Our evaluation shows that \sys significantly outperforms existing circuit optimizers on almost all benchmark circuits. Surprisingly, \sys can learn to perform rotation merging---a complex, non-local circuit optimization implemented as a separate pass in existing optimizers.
%while existing optimizers rely on manual implementation of rotation merging.
%A key idea behind \sys is formalizing quantum circuit optimization as a Markov decision process and training a reinforcement learning (RL) agent to apply circuit transformations. 
%\oded{I don't think we need to mention MDP in the abstract.}
%\oded{How about adding something like:

%\oded{we need to add a bottom line, how much better are the results, maybe something about generalizability or success of our novel neural architecture.}

%\JP{As a result, \sys can effectively capture the local features of quantum circuits and efficiently optimize diverse categories of real-world quantum circuits, outperforming existing approaches by a large margin (maybe 21.9\% (\sys w/ r.m. / Quartz w/ r.m.)  or up to 129.26\% (\sys w/o r.m. / Quartz w/o r.m.) ?).}

\end{abstract}

\maketitle

\section{Introduction\label{sec:intro}}

%\oded{I think the opening here is too abrupt and gets technical too fast. We need a more gentle opening paragraph.}\ZJ{I have added an opening para.}

Quantum computing presents a novel paradigm that enables significant acceleration over classical counterparts in a wide range of applications, such as quantum simulation~\cite{cao2019quantum}, integer factorization~\cite{monz2016realization}, and machine learning~\cite{biamonte2017quantum}.
However, programming quantum computers is a challenging task due to the scarcity of qubits and the diverse forms of noise that affect the performance of near-term intermediate-scale quantum (NISQ) devices.

Quantum programs are commonly represented as {\em quantum circuits}, such as the one shown in \Cref{fig:long-sequence}, where each horizontal wire represents a qubit and boxes on these wires represent quantum gates.
To enhance the success rate of executing a circuit, a common form of optimization is applying {\em circuit transformations}, which replace a subcircuit matching a specific pattern with a functionally equivalent subcircuit that has better performance (e.g. fidelity, depth).

%, such as reduced execution cost or improved fidelity.
% TODO: maybe we should go deeper on why we should optimize quantum circuit at era of NISQ devices.

\begin{figure}[ht]
    \centering
    % \hspace{-5mm}
    \includegraphics[scale=0.27]{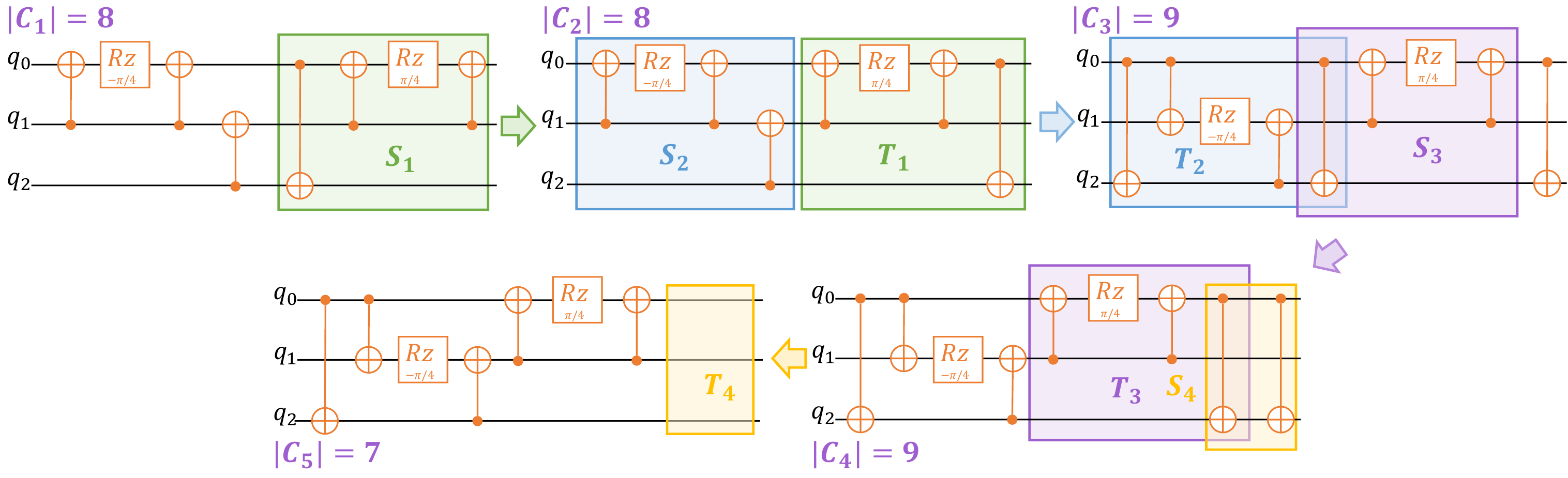}
    \vspace{-2mm}
    \caption{A quantum circuit optimization with four transformations, one of which increases the cost as an intermediate step. Each arrow indicates a transformation, whose source and target sub-circuits are highlighted by boxes in the same color.
    % \ZJ{Can we convert this to a two-column figure}
    }
    \label{fig:long-sequence}
\end{figure}

%%%
% Prior work has proposed two approaches to performing circuit transformations on an input circuit.
% First, many quantum compilers, such as \qiskit~\cite{qiskit}, \tket~\cite{tket}, and \quilc~\cite{quilc}, use {\em rule-based} strategies that greedily apply a set of manually designed circuit transformations.
% These circuit transformations are designed by quantum-computing experts to strictly optimize the performance of quantum circuits,
% such that greedy application results in improved performance.
% Second, recent work~\cite{pldi2022-quartz,quanto} introduced a {\em search-based} approach that optimizes a quantum circuit by exploring a search space of circuits functionally equivalent to the input one.
% Quartz~\cite{pldi2022-quartz} automatically generates and verifies circuit transformations for a given gate set. These transformations preserve equivalence but do not necessarily improve performance.
% To optimize an input circuit, Quartz uses a cost-based backtracking search algorithm to apply these transformations and discover an optimized circuit.
% Although existing approaches improve the performance of quantum circuits, they are limited by the following challenges in transformation-based quantum circuit optimization.

Prior research has proposed two approaches for performing circuit transformations on an input circuit. The first approach is the use of {\em rule-based} strategies, which are employed by many quantum compilers such as \qiskit~\cite{qiskit}, \tket~\cite{tket}, and \quilc~\cite{quilc}. 
These strategies involve the greedy application of a set of circuit transformations that are manually designed by quantum computing experts to improve the performance of quantum circuits. 
The second approach, as introduced in recent work~\cite{pldi2022-quartz,quanto, xu2022synthesizing}, is a {\em search-based} approach that explores a search space of circuits that are functionally equivalent to the input circuit. 
Quartz~\cite{pldi2022-quartz} automatically generates and verifies circuit transformations for a given gate set, which preserves equivalence but may not necessarily improve performance. 
To optimize an input circuit, Quartz employs a cost-based backtracking search algorithm to apply these transformations and discover an optimized circuit. 
Although existing approaches improve the performance of quantum circuits, they are limited by the following challenges in transformation-based quantum circuit optimization.

% \paragraph{Planar optimization landscape.}
% Given an input circuit and a set of (verified, equivalence-preserving) transformations, the search space consists of all circuits reachable from the input circuits by iteratively applying  transformations.
% Discovering an optimal circuit in the search space is challenging; the search space is too large to exhaustively explore all of it, and the cost function (derived from a selected performance metric) does not provide enough guidance for a greedy approach. We refer to this situation as a \emph{planar optimization landscape} as the path from a circuit to a circuit with lower cost typically may contain many steps in which the cost does not change.
% (These so called plateaus also exist in classic program optimization~\cite{DBLP:conf/pldi/KoenigPA21}.)

\paragraph{Planar optimization landscape.}
The set of circuits that can be reached from an input circuit by iteratively applying verified, equivalence-preserving transformations comprises the search space in quantum circuit optimization. 
However, finding the optimal circuit in this space is challenging due to the space's size, which makes exhaustive exploration infeasible, and the inability of the cost function (derived from a selected performance metric) to provide enough guidance for a greedy approach. 
This scenario is referred to as a \emph{planar optimization landscape} since the path from one circuit to another with lower cost often contains many steps in which the cost remains unchanged. (Plateaus of this sort are also present in classic program optimization~\cite{DBLP:conf/pldi/KoenigPA21}.)

\commentout{
\begin{figure}[t]
    \centering
    \subfloat[Cost-increasing transformation.]{
    \includegraphics[scale=0.22]{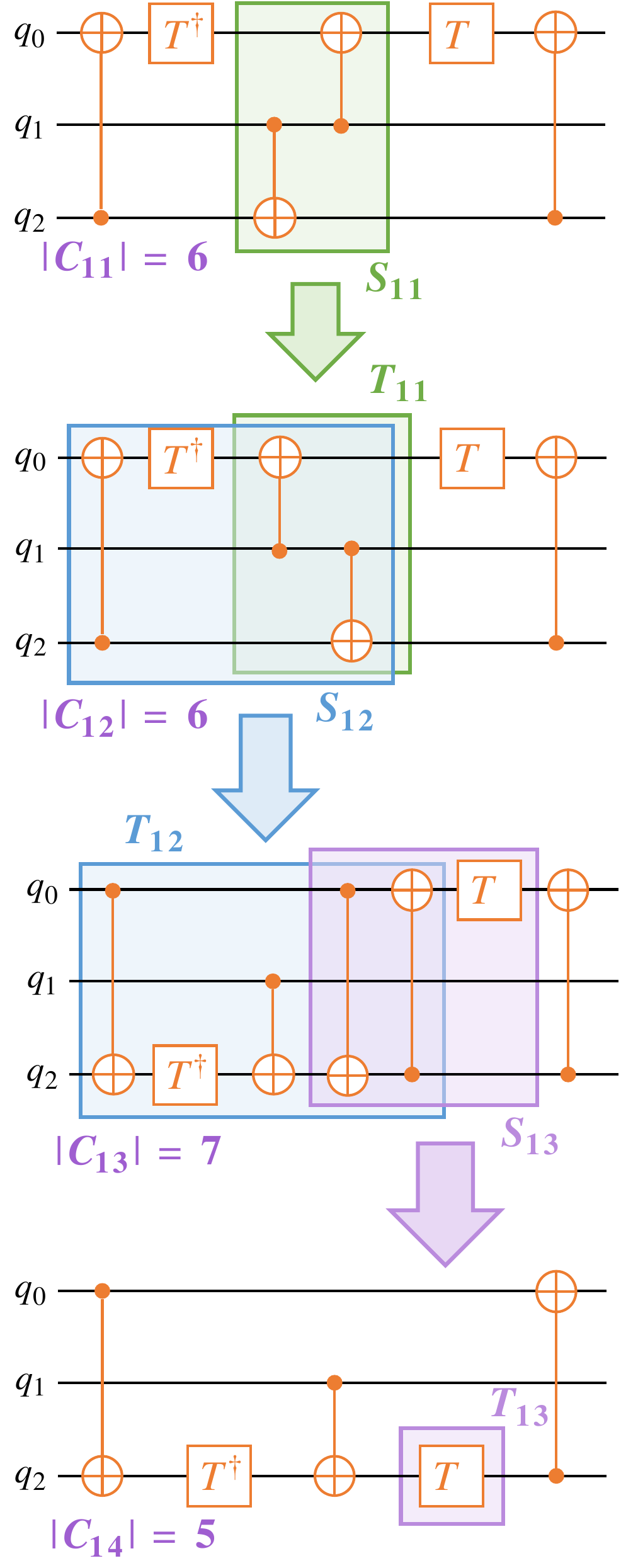}
    \label{fig:xfer_left}
    }
    \subfloat[Long sequence optimization.]{
    \includegraphics[scale=0.205]{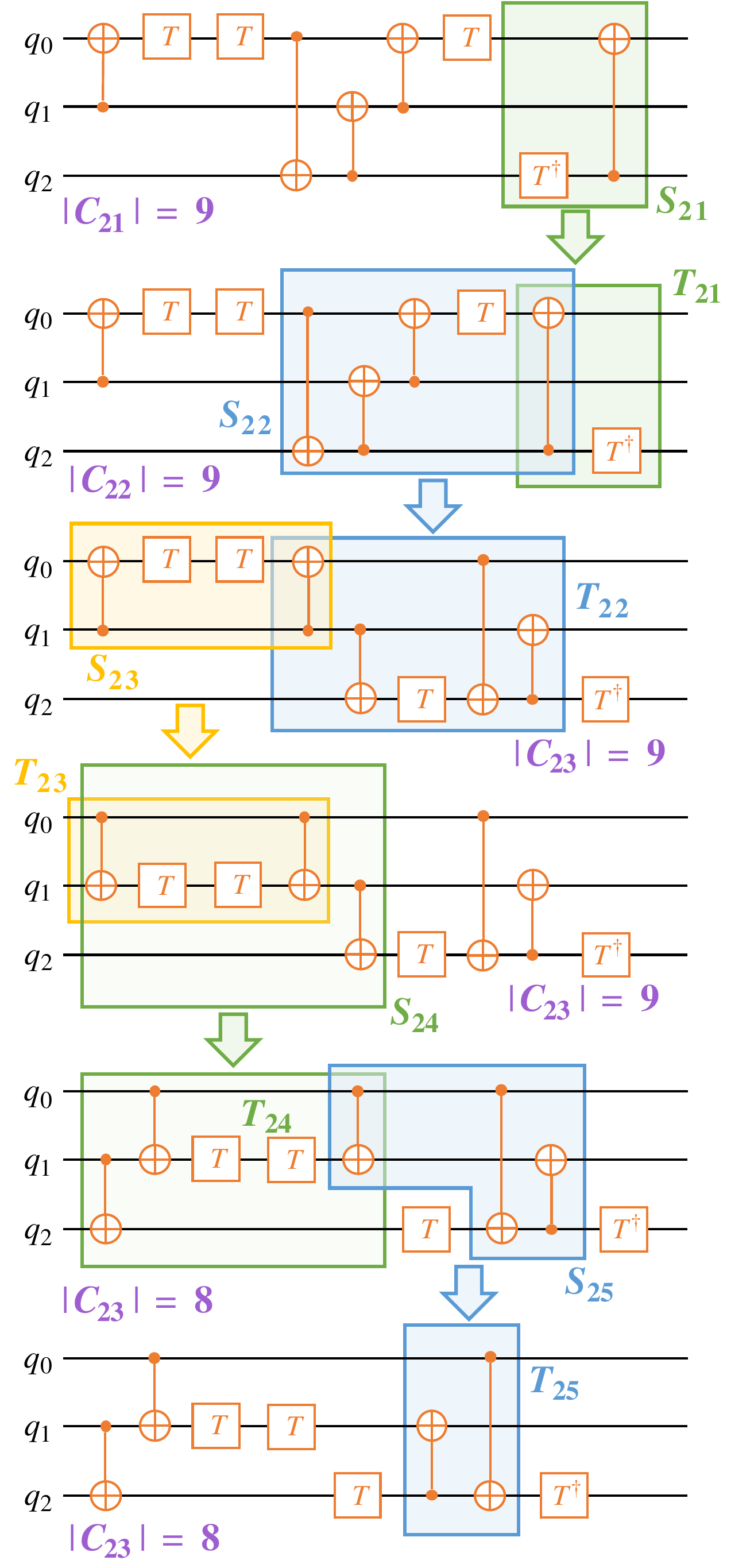}
    \label{fig:xfer_right}
    }
    \caption{captain}
\end{figure}
}

% \begin{figure}
%     \centering
%     \vspace{-5mm}
%     \includegraphics[scale=0.42]{figs/figure2_new_cdf_v3.pdf}
%     \vspace{-8mm}
%     \caption{CDF of the minimal number of transformations needed to reduce gate count. The result is calculated using 1000 sampled circuits in the search space of \tof.}
%     \label{fig:all_dist}
% \end{figure}

% \begin{figure}
%     \centering
%     \includegraphics[scale=0.36]{figs/seq_barenco_tof_3_2.pdf}
%     \caption{The optimization trace of the best discovered circuits when including and excluding cost-increasing transformations.}
%     \label{fig:sequence_analysis}
% \end{figure}

\begin{figure}[htbp]
\label{fig:ibm_depth_fidelity}
  \centering
  \subfloat[]{
    \centering
    \includegraphics[width=0.46\textwidth]{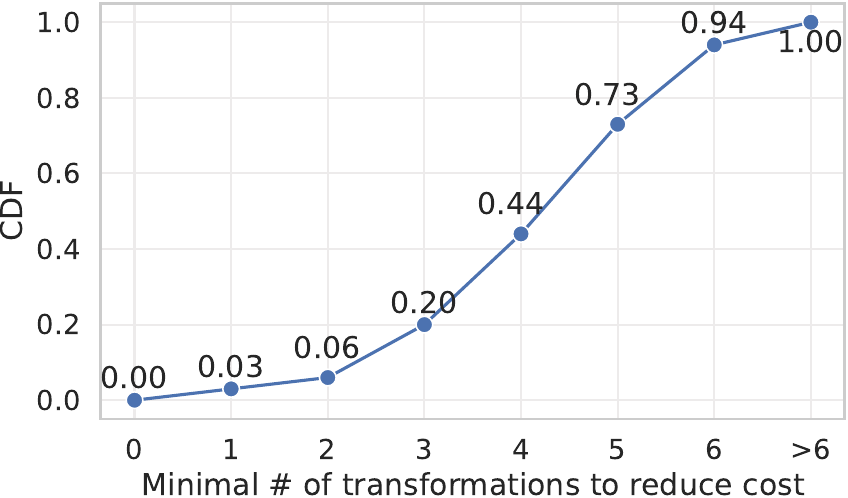}
    %\caption*{(a)
    }
    \hspace{0.02\textwidth}
  \subfloat[]{
    \centering
    \includegraphics[width=0.46\textwidth]{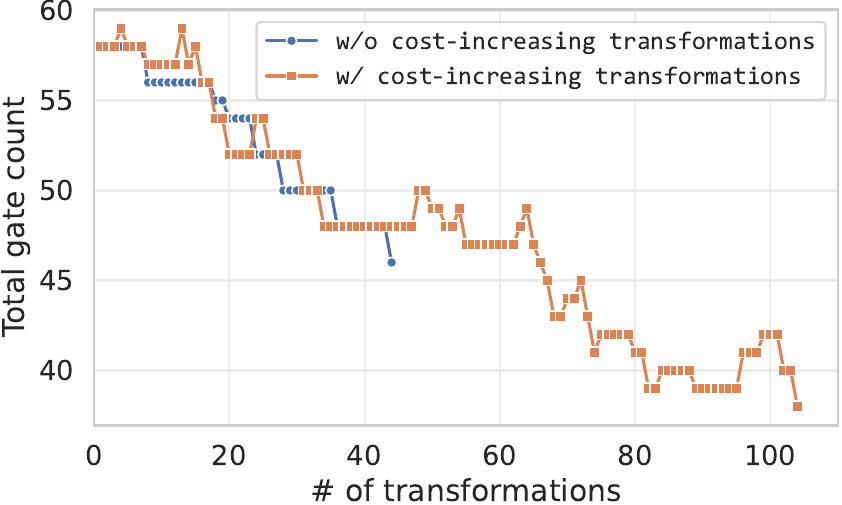}    
    \label{fig:search_space_b}
    }
    \vspace{-2mm}
  \caption{An analysis on the search space of \texttt{barenco\_tof\_3}. (a) CDF of the minimal number of transformations needed to reduce gate count. 
  %The result is calculated using 200 sampled circuits in the search space of \tof. 
  (b) The optimization trace of the best discovered circuits when including and excluding cost-increasing transformations.}
  \label{fig:search_space}
\end{figure}

To illustrate the challenge of discovering an optimal circuit in the search space, we analyze the search space of a relatively small circuit \tof~\cite{nam2018automated}.
The circuit includes three Toffoli gates, implemented with 58 gates in the Nam gate set ({$CNOT$, $X$, $H$, $Rz$}). 
We consider the 6,206 transformations discovered by Quartz~\cite{pldi2022-quartz} for the Nam gate set and exhaustively find all circuits reachable within seven transformation applications. 
Of these roughly 162000000 circuits, we randomly sample 200 circuits and analyze the optimization landscape around them. 
Specifically, for each sampled circuit $C$, we perform a breadth-first search (BFS) to determine the shortest path from $C$ to a circuit $C'$ with a lower cost (using the total number of gates as the cost function). 
That is, we determine the radius of the planar optimization landscape around $C$. 
The BFS is performed up to a radius of 6, so we either find the exact radius or conclude that it is greater than 6. 
Notably, it is guaranteed that all of the 200 circuits can be optimized because their cost is larger then the optimal cost.
The results are summarized in \Cref{fig:search_space}. 
For more than half of the sampled circuits, reducing the total gate count requires applying more than 4 transformations, and 6\% of the sampled circuits require more than 6 transformations to improve performance. 
While we could exhaustively explore transformation sequences of length seven for \tof, this would not be practical for larger circuits with hundreds of gates or more. 
Therefore, in the absence of guidance from the cost function, it is natural to consider a learning-based approach to learn a good heuristic that can guide the application of transformations towards a lower-cost circuit.

\paragraph{Cost-increasing transformations.}
One of the challenges in optimizing quantum circuits is the need to use transformations that may temporarily increase the cost. 
For example, the second transformation in \Cref{fig:long-sequence} increases the gate count from 8 to 9, which is necessary to enable the subsequent transformations that ultimately reduce the gate count to 7. 
Cost-increasing transformations are crucial to discovering highly optimized circuits. 
To illustrate this point, we compare the performance of the best discovered circuits with and without cost-increasing transformations in \Cref{fig:search_space_b}. 
Without using cost-increasing transformations for circuit \tof, the best circuit that can be discovered in our experiments has 46 gates, while involving cost-increasing transformations can further reduce the gate count to 36. 
However, determining when and where to apply cost-increasing transformations is non-trivial, as these transformations are only beneficial when combined with other transformations that eventually reduce the cost.

\paragraph{Our approach.}
This paper presents \sys, a learning-based quantum circuit optimizer. 
Specifically, \sys utilizes reinforcement learning (RL) to guide the application of quantum circuit transformations. In the RL task, each circuit is defined as a state, and each application of a transformation is considered an action. By adopting this formulation, \sys can learn to identify circuit optimization opportunities and perform long sequences of transformations to achieve them.

The first challenge we must address in applying RL to circuit optimization is the large action space. For example, when learning to apply the 6,206 transformations discovered by Quartz~\cite{pldi2022-quartz} on a thousand-gate circuit, there may be millions of possible actions (i.e., ways to match a transformation to a subcircuit of the current circuit). Such a large action space would degrade the training efficiency of most RL algorithms.
%%%
To deal with this issue, we propose a {\em hierarchical action space} that uses a gate-transformation pair to uniquely identify a potential transformation application, so that the RL agent first chooses a gate and then a transformation to apply. This decomposition allows \sys to use two much smaller sub-action spaces, enabling effective training.
To effectively train the RL agent to select both a gate and a transformation, we propose {\em hierarchical advantage estimation} (HAE), which allows \sys to train two policies with a single {\em actor-critic} architecture.
\sys combines HAE with {\em proximal policy optimization} (PPO)~\cite{ppo} to jointly train the gate- and transformation-selecting policies.

% To apply reinforcement learning (RL) to circuit optimization, the first challenge is the large action space. 
% In the case of a thousand-gate circuit with 6,206 potential transformations to choose from, there could be millions of possible actions, which would significantly reduce the training efficiency of most RL algorithms. 
% To address this issue, we propose a {\em hierarchical action space} that decomposes the action selection process into two sub-steps: gate selection and transformation selection. 
% By using a gate-transformation pair to identify a potential transformation application, our approach enables effective training with much smaller sub-action spaces. 
% To train the RL agent to select both a gate and a transformation, we propose {\em hierarchical advantage estimation} (HAE), which enables us to train two policies with a single {\em actor-critic} architecture. 
% To jointly train the gate- and transformation-selecting policies, we utilize {\em proximal policy optimization} (PPO)~\cite{ppo}. 
% Our proposed method, \sys, leverages HAE and PPO to enable effective RL-guided quantum circuit optimization.

\commentout{
In addition, \ZL{we propose {\em hierarchical advantage estimation} (HAE) which allows us to train the two policies jointly with a single actor-critic architecture}. we introduce {\em hierarchical proximal policy optimization} (HPPO), a variant of PPO~\cite{ppo}, to learn two policies that jointly a gate-transformation pair for a transformation.
\oded{Is Zikun's suggestion above meant to replace the following sentence or  is it in addition? Anyway, here's my attempt to combine both of them them:
To effectively train the RL agent to select both a gate and a transformation, we propose {\em hierarchical advantage estimation} (HAE), which allows us to train the two policies jointly with a single actor-critic architecture.
We combine HAE with proximal policy optimization (PPO)~\cite{ppo}, yielding {\em hierarchical proximal policy optimization} (HPPO), a variant of PPO that jointly trains policies for selecting a gate and a transformation.
}
}

The second challenge for RL-based quantum circuit optimization is state representation.
An RL application typically represents states as vectors in a fixed, high-dimensional space. 
Unlike most RL tasks whose states have relatively uniform structure, a state in our setting is a quantum circuit whose size and topology depend on the input circuit and change at each step during the optimization process, making it non-trivial to design a fixed, uniform state representation.
A key insight behind \sys's approach is to leverage the {\em locality} of quantum circuit transformations, that is, the decision of applying a transformation at a certain location is largely guided by the local environment at that location.
Based on this insight, \sys uses a {\em graph neural network} (GNN)-based approach to represent the local environment of each gate. 
Although each representation encodes the local environment of only one gate, combining all gate representations allows \sys to identify optimization opportunities across the entire circuit. 
A key advantage of our learned gate-level representations is that they are independent to the circuit size and generalize well to unseen circuits.
By combining this local decision making with global circuit-wide fine-tuning, we aim to achieve a balance of both local and global guidance for the optimization process.

% The second challenge of applying RL to quantum circuit optimization is the representation of the states. 
% In most RL tasks, states are represented as vectors in a fixed, high-dimensional space. 
% However, in our setting, a state is a quantum circuit with a size and topology that depend on the input circuit and change at each step during optimization. 
% Thus, it is non-trivial to design a fixed, uniform state representation. 
% To address this challenge, \sys leverages the locality of quantum circuit transformations, which suggests that the decision to apply a transformation at a particular location is largely guided by the local environment at that location. 
% To represent the local environment of each gate, we use a graph neural network (GNN)-based approach. 
% Although each representation encodes the local environment of only one gate, combining all gate representations allows \sys to identify optimization opportunities across the entire circuit. 
% Our learned gate-level representations are independent of circuit size and generalize well to unseen circuits. To balance both local and global guidance for the optimization process, we combine this local decision-making with global circuit-wide fine-tuning.

The evaluation results demonstrate the superior performance of \sys over existing circuit optimizers for almost all circuits. 
For example, on the Nam gate set, \sys achieves an average reduction of 35.8\% and 33.9\% in total gate count and {\tt CNOT} gate count, respectively (geometric mean), while the best existing optimizers achieve reductions of only 28.3\% and 20.6\%. On the IBM gate set, \sys achieves an average reduction of 36.6\% and 21.3\% in total gate count and {\tt CNOT} gate count, respectively (geometric mean), while the best existing optimizers achieve reductions of 20.1\% and 7.7\%. 
Furthermore, \sys improves circuit fidelity by up to 4.66$\times$ (with an average fidelity improvement of 1.37$\times$) on the IBM gate set, while the best existing optimizer only improves circuit fidelity by 1.07$\times$.
%optimized circuits exhibit fidelity up to 4.66 times that of the original circuits, with an average improvement of 1.37 times, while the best existing work improves fidelity by 1.07 times. 
Notably, the evaluation also reveals that \sys can automatically discover rotation merging, a non-local circuit optimization, from its own exploration, while existing optimizers require manual implementation of this technique.

%\oded{we say "such as rotation merging". do we mean anything else other than rotation merging? if not, I suggest to replace the sentence with: In addition, the evaluation shows that for some circuits, \sys can automatically discover rotation merging---a known, non-local circuit optimization---from its own exploration, while existing optimizers rely on a manual implementation of rotation merging.}
%%%
%Finally, we show that \sys can effectively generalize to previously unseen circuits, reducing the fine-tuning time by 10$\times$ on new circuits.
%\oded{The last sentence is unclear. I think we can remove it, or if it's important we should rewrite it.}

%\oded{Unless we're absolutely out of space, we need something like the following:}

The remainder of this paper is organized as follows. 
\Cref{sec:background} provides background information on transformation-based quantum circuit optimization and reinforcement learning. 
\Cref{sec:challenge} outlines the key challenges associated with applying reinforcement learning to quantum circuit optimization and describe our approach for addressing them. 
The main technical contributions of this paper are presented in \Cref{sec:design} that presents \sys's neural architecture, and \Cref{sec:training} that presents \sys's training methodology. 
Empirical evaluations are presented in \Cref{sec:eval}, and related work are discussed \Cref{sec:related}.

%The first challenge we must address in building \sys is designing a learnable representation for circuits that can cope with reinforcement learning and generalize to circuits with arbitrary number of qubits and gates.
%\sys uses a novel {\em graph neural network} (GNN)-based approach to representing the local subcircuits associated with each gate as a high-dimensional vector, which is used by \sys to identify long-sequence optimization opportunities in a circuit and guide \sys to perform the necessary transformations to enable these optimizations.

%The second challenge is that {\em explicitly} constructing an action space (i.e., all valid applications of transformations) for a given state (i.e., a quantum circuit) is computationally expensive, since a transformation can be applied to multiple places of a circuit. 
%%%
%To deal with this issue, we propose a {\em hierarchical action space} that uses a gate-transformation pair to uniquely identify an application of a transformation on a circuit. 

%\oded{above, we use both cost and performance, maybe we should make it uniform.}

\commentout{
\section*{Old Version}
% Need some articulation, quartz's limitation
\oded{Following our discussion, I think this paragraph needs to be removed or moved later, or alternatively we would also have here other paragraphs or more sentences in this paragraph about the other two challenges: size of the search space and need for temporary cost increase.}
A key challenge in optimizing quantum circuits is discovering {\em long-sequence} optimizations, which reduce a circuit's cost by applying a long sequence of transformations, while any subset of these transformations does not reduce the cost.
\Cref{fig:long-sequence} demonstrates such an example, where we optimize the input circuit $C_1$ using transformations discovered by Quartz.
%%%
The optimization can reduce both the gate count and depth of the input circuit, but realizing this optimization requires sequentially applying five transformations as depicted in the figure.
Note that the first three transformations temporarily increasing the circuit depth in intermediate steps, while the last two transformations reduces both the depth and gate count.

\oded{I think we should eliminate depth from \Cref{fig:long-sequence}. It would be ideal if we can use another example where the gate count needs to temporarily increase. Also, can Quartz discover the sequence currently depicted in the figure? If not, we should say so. If yes, then I think it would be better to use an example where Quartz cannot/does not discover it.}

Existing quantum circuit optimizers cannot efficiently discover long-sequence optimizations.
In particular, enabling the aforementioned long-sequence optimization in today's rule-based circuit optimizers requires designing transformation rules that can match all nine gates in the given pattern, while the number of rules needed to capture long-sequence optimizations grows exponentially with the number of gates.
%%%
Enabling long-sequence optimizations in the search-based approach is also challenging due to the very large search space.
%%%
\Cref{fig:long-sequence} shows the shortest sequence of transformations needed to perform an optimization from the initial circuit (for the given set of transformations used by Quartz); note that all intermediate transformations do not immediately reduce the cost and achieving improved performance requires simultaneous application of all transformations, which is challenging in a purely randomized search. A potential strategy is introducing heuristics into the cost function used to conduct the search; however, the heuristics largely depend on (1) the set of gates supported by the underlying quantum hardware and (2) the performance metrics we are optimizing for.
In addition, the heuristics also need to consider transformations that decrease performance as intermediate steps.
\oded{The end of the above paragraph is a little repetitive. If we keep it like this I can try to edit it later.}

The first contribution of this paper is identifying and systematically analyzing long-sequence optimizations, a key class of circuit optimizations missing in existing quantum compilers. 
%%%
For a quantum circuit, the transformations discovered by prior work constitutes a very large search space of equivalent circuits.
We show that long-sequence optimizations are critical to discover high-optimized circuits. This is because most circuits in the search space are each associated with a large, planar optimization landscape, requiring combinations of multiple transformations to reduce cost.
In addition, discovering highly optimized circuits also requires applying transformations that temporarily decrease performance in order to enable subsequent optimizations; however, existing optimizers do not consider cost-increasing transformations, thus miss these optimization opportunities.
%Optimizing a quantum circuit requires exploring a very large search space of equivalent circuits. In addition, most circuits in the search space is associated with a planar optimization landscape, requiring 

We introduce \sys, a learning-based quantum circuit optimizer to discover long-sequence optimizations.
A key idea behind \sys is formalizing quantum circuit optimization as a Markov decision process (MDP) and applying long-sequence optimizations using reinforcement learning (RL).
\sys learns to detect the appearances and positions of long-sequence optimization opportunities and perform the the right sequence of transformations to realize them.

The first challenge we must address in building \sys is designing a learnable representation for circuits that can cope with reinforcement learning and generalize to circuits with arbitrary number of qubits and gates.
\sys uses a novel {\em graph neural network} (GNN)-based approach to representing the local subcircuits associated with each gate as a high-dimensional vector, which is used by \sys to identify long-sequence optimization opportunities in a circuit and guide \sys to perform the necessary transformations to enable these optimizations.

The second challenge is that {\em explicitly} constructing an action space (i.e., all valid applications of transformations) for a given state (i.e., a quantum circuit) is computationally expensive, since a transformation can be applied to multiple places of a circuit. 
%%%
To deal with this issue, we propose a {\em hierarchical action space} that uses a gate-transformation pair to uniquely identify an application of a transformation on a circuit. 

allowing \sys to construct only a sub-action space to save computation while being concise.

%By adopting RL in the long-sequence optimization problem, our goal is to learn to detect the appearances and positions of optimization opportunities and perform the right sequence of transformations to realize them.
This goal is deeply reflected in the design of \sys, which has two key components, namely the \critic and the \actor.
The \critic learns to evaluate the potential to be optimized of different positions of a circuit.
It aims to propagate the reward signal from the end of a long optimizing sequence to the beginning, enabling the agent to detect the optimization opportunity multiple transformations away.
The \actor learns how to choose the right transformation in different cases to realize the optimization.
with the two components, \sys discovers, locates and seizes long-sequence optimization opportunities, which allows it to optimize circuits more thoroughly.

\commentout{
\subsection*{Old version}
To discover long-sequence optimizations, we introduce \sys, a learning-based quantum circuit optimizer.
A key idea behind \sys is formalizing quantum circuit optimization as a Markov decision process (MDP) and discovering long-sequence optimizations using reinforcement learning (RL).
By adopting RL in the long-sequence optimization problem, our goal is to learn to detect the appearances and positions of optimization opportunities, and to learn to perform the right sequence of transformations to realize them.
This goal is deeply reflected in the design of \sys, which has two key components, namely the \critic and the \actor.
The \critic learns to evaluate the potential to be optimized of different positions of a circuit.
It aims to propagate the reward signal from the end of a long optimizing sequence to the beginning, enabling the agent to detect the optimization opportunity multiple transformations away.
The \actor learns how to choose the right transformation in different cases to realize the optimization.
with the two components, \sys discovers, locates and seizes long-sequence optimization opportunities, which allows it to optimize circuits more thoroughly.

While the long-sequence optimization problems are difficult to solve, as will be stated in section \ref{sec:motivation}, designing a RL algorithm for this problem is also non-trivial.
One of the challenges is that constructing an action space (all valid applications of transformations) for a given state (a quantum circuit) is highly inefficient in computation.
To deal with this issue, we propose a {\em hierarchical action space} which uses a gate-transformation pair to uniquely locate an application of a transformation on a circuit, which allows to construct only a sub-action space in action selection to save computation while being concise.
This will be discussed more exhaustively in section \ref{sec:background}.
Also, designing an RL algorithm that matches our problem is not easy since deploying RL algorithms in literature naively works poorly.
In order to tackle the particular difficulties in our problem, we modified the PPO algorithm~\cite{ppo} and introduce a hierarchical policy.
We will discuss our design of the RL algorithm and the training procedure in section \ref{sec:design} and \ref{sec:training} respectively.
By focusing on long-sequence optimization, \sys is able to explore for and exploit more optimization opportunities, which leads to a great improvement in the optimizing results.
As shown in \ref{}, \sys outperforms previous works in all the benchmarks in total gate count.
*** also mention cx count, need experiment data***
Also, it is shown that \sys is able to discover optimizing rules from scratch without external knowledge, such as rotation merging.
This indicates that \sys has a potential to discover unknown rules which can help us better understand quantum circuit optimization problem.
%%%
% A key idea behind \sys is a learning-based method to predict the potential benefit of different circuit candidates. \ZL{Another is to train an RL agent to automatically select transformations to perform on selected subcircuits.}

\ZL{
\begin{itemize}
    % \item define long-sequence optimization in a easy to understand way
    % \item our approach is based on search based approach
    % \item  May need to mention how we form this problem into a MDP
    % \item the challenge in forming the action space
    % \item our advantages
    \item contributions
    % \item rationale behind our rl algorithm
\end{itemize}
}
}
\ZJ{What are the challenges we had addressed when designing a RL solution for circuit optimizations?}
\begin{itemize}
    \item May worth mentioning that \sys discovers a long-sequence optimization that optimizes barenco\_tof\_3 from 38 to 36 gates by applying 75 transformations, some of which even increases the gate count. 
\end{itemize}
}

\section{Background\label{sec:background}}

\if 0
\begin{figure}[t]
    \centering
    \subfloat[RL {\em without} action space decomposition.]{
    \includegraphics[scale=0.38]{figs/naive_rl.pdf}
    \label{fig:non_hierarchical_rl}
    }
    \\
    \subfloat[RL {\em with} action space decomposition.]{
    \includegraphics[scale=0.38]{figs/hierarchical_rl.pdf}
    \label{fig:hierarchical_rl}
    }
    \caption{
    Comparing the interactions between the RL agent and circuit transformation engine with and without action space decomposition.
    %\Cref{fig:non_hierarchical_rl} shows how the RL agent interact with the circuit transformation engine with a non-hierarchical action space. \Cref{fig:hierarchical_rl} shows how the RL agent interact with the circuit transformation engine with a hierarchical action space.
    }
    \label{fig:rl_interaction}
\end{figure}
\fi
%Specifically, we first will show how this problem forms an \mdp, and then we will introduce in detail how to organize the actions into an hierarchical action space.
%After that we will formally state the reinforcement learning.

\begin{figure}
    \centering
    %\hspace*{2mm}
    \subfloat[Circuit representation.]{
    \includegraphics[scale=0.28]{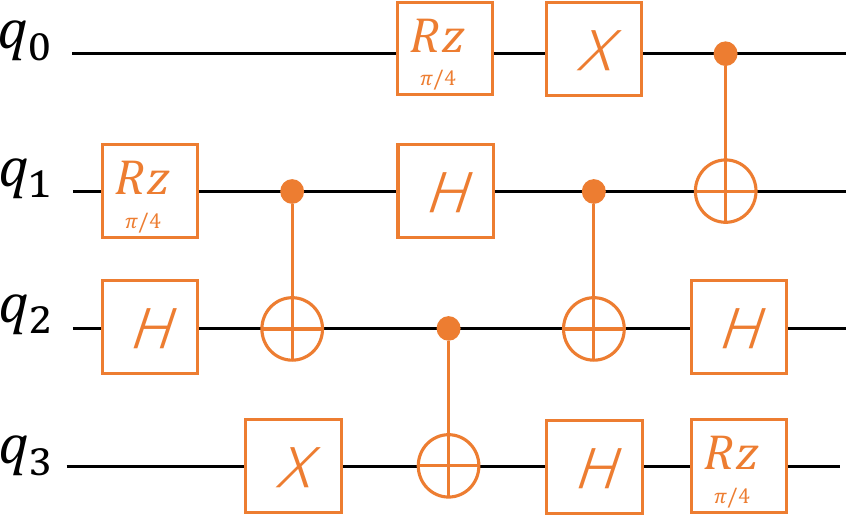}
    \label{fig:circuit_representation}
    }
    % \hspace*{3mm}
    % \vspace{2mm}
    \subfloat[Graph representation.]{
    \includegraphics[scale=0.28]{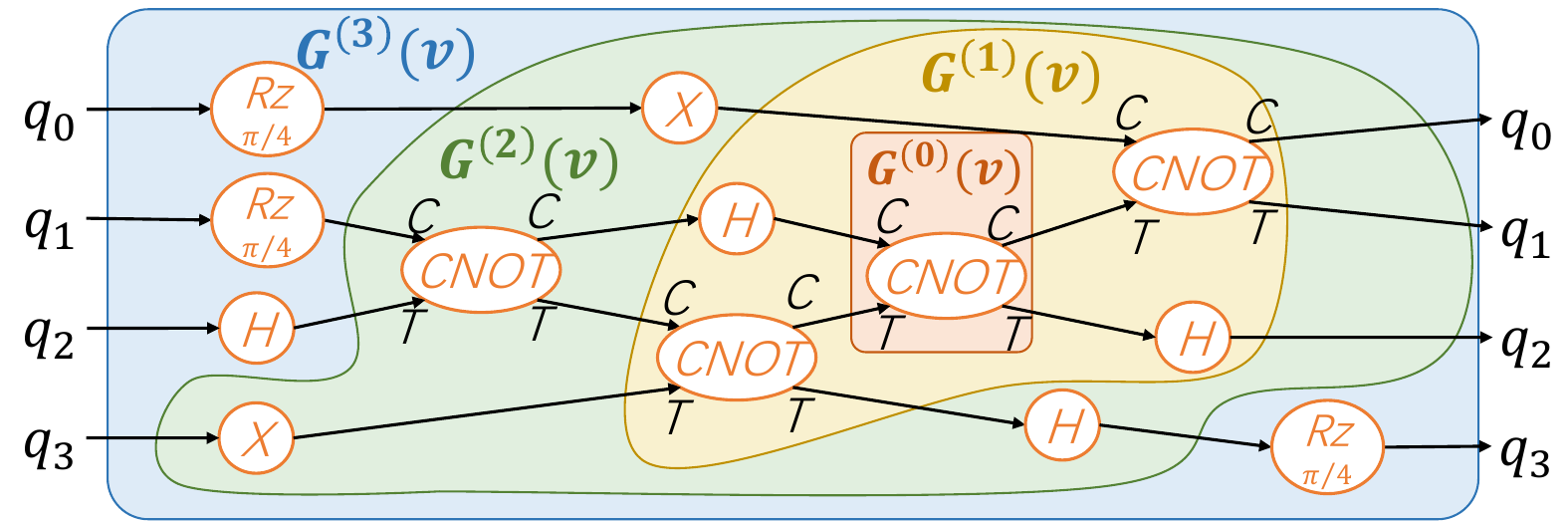}
    \label{fig:graph_representation}
    }
    \vspace{2mm}
    \\
    \subfloat[Circuit transformation and its graph representation.]{
    \includegraphics[scale=0.34]{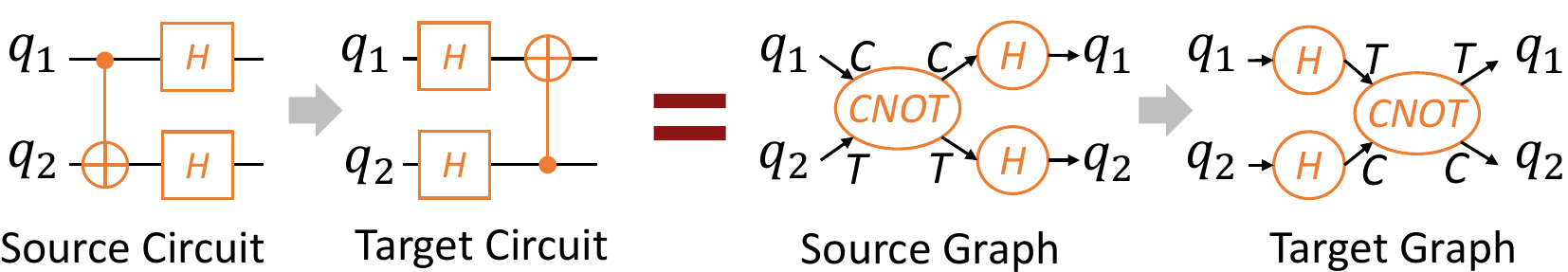}
    \label{fig:graph_subst}
    }
    \vspace{-2mm}
    \caption{A quantum circuit (a) and its graph representation (b). For each {\tt CNOT}, C and T correspond to the control and target qubits. $G^{(k)}(v)$ shows the $k$-hop neighborhood (see \Cref{subsec:design:gnn}) for the {\tt CNOT} gate identified by $G^{(0)}(v)$. (c) illustrates a circuit transformation and its graph representation.}
\end{figure}

\paragraph{Graph representation of quantum circuits.}\label{subsec:graph_repr}
We adopt the graph representation of quantum circuits from prior work~\cite{pldi2022-quartz}.
As illustrated in \Cref{fig:circuit_representation,fig:graph_representation},
a circuit $C$ is represented as a directed acyclic labeled graph $G = (V, E)$.
Each gate with $q$ qubits is represented as a node $v\in V$ with $q$ in- and out-edges, and each edge $e\in E$ represents a connection of two adjacent gates on a qubit.
Nodes are labeled with gate types, and edges are labeled to distinguish different qubits of a multi-qubit gate 
(e.g., the control and target qubits of a {\tt CNOT} gate).

\commentout{
\oded{We can explain a bit more about transformations: mention that a transformation involves changing a convex subgraph, say that it preserves the semantics of the circuit (refer to Quartz paper for details), and maybe also explain here that a an application of a transformation to a circuit is determined once a single gate is determined (and explain why). Maybe it makes sense to split this paragraph into two paragraphs: "Graph representation of quantum circuits." and "Circuit transformations."/"Optimizing quantum circuits with transformations."}
\ZJ{Done.}

\oded{We should say somewhere that for the rest of this paper, we fix a given gate set and a given set of transformations over this gate set. I think we can say it here. Later when we define quantum optimization as RL we currently don't explain that this all depends on a given set of transformations.}
}

%\paragraph{Optimizing quantum circuits with transformations.}
\paragraph{Transformation-based circuit optimization.}
To improve the performance of quantum circuits, a common form of optimization is {\em circuit transformations}, which replace a subcircuit matching a specific pattern with a distinct equivalent subcircuit.
In the graph representation, a subcircuit correspond to a convex subgraph~\cite[p.8]{pldi2022-quartz}; a transformation is represented as a pair of connected graphs $(G,G')$ representing equivalent circuits as illustrated in \Cref{fig:graph_subst}, and applying it to a circuit $C$ amounts to finding a convex subgraph of $C$ that is isomorphic to $G$, and then replacing it by $G'$ to yield a new circuit $C'$.
\commentout{
\oded{There isn't such a thing as a convex graph, only a convex subgraph of another graph. Maybe we can say instead:
In the graph representation, a subcircuit correspond to a convex subgraph~\cite[p.8]{pldi2022-quartz}; a transformation is represented as a pair of graphs $(G,G')$ representing equivalent circuits as illustrated in \Cref{fig:graph_subst}, and applying it to a circuit $C$ amounts to finding a convex subgraph of $C$ that is isomorphic to $G$, and then replacing it by $G'$ to yield a new circuit $C'$.
}
}
An application of a transformation $(G,G')$ to a circuit $C$ is uniquely determined once a single gate $g \in G$ is matched with a gate $g_C \in C$, due to the unique edge labeling and the fact that $G$ is always a connected graph.
\commentout{
\oded{I don't think we need the part about explaining what gate we actually use. If it's important we can explain it later when we explain more low-level implementation details, or maybe we can even drop it completely (does it have any significance which gate we choose for locating a transformation, as long as we have some consistent choice of one gate in each source graph?)
}
}
\commentout{
\sys uses a {\em single} gate to locate an application of the transformation on a circuit.
While in general graph transformations, a single node cannot deterministically identify how a transformation is applied, this approach is feasible and highly desirable under the quantum setting, since the source and target graphs of a transformation are connected, and
the in- and out-edges of multi-qubit gates are indexed\footnote{
For instance, the control and target qubits of a {\tt CNOT} are indexed 0 and 1.}.
This property guarantees that the orders of in- and out-edges are preserved in graph pattern matching.
% For multi-qubit gates the in-edges and out-edges are associated with different qubits and therefore we order them based on the qubit ids.
%%%
As a result, when a gate $g_x$ in transformation $x$ is mapped to a gate $g_C$ in circuit $C$, all other gates in the source graph of the transformation deterministically match gates in $C$ if the transformation is applicable.
In practice, we use the topologically first gate $g_x$ of a transformation $x$ to locate the transformation\footnote{When there is a tie, we pick the gate with the smallest qubit id.}.%, that is, when applying $x$ to circuit $C$, if $g_x$ is matched with gate $g_C$, we say that we apply $x$ to $g_C$.
For example, the first transformation in \Cref{fig:long-sequence} is applied to the $T^\dag$ gate in $C_1$, and the remaining four transformations in the figures are applied to the left-most $CNOT$ gate in $S_2$, $S_3$, $S_4$, and $S_5$, respectively.
%\ZJ{To be removed:A key benefit of locating a transformation using a single gate instead of a subcircuit is to simplify the position space $P$, since $P$ is exactly all the gates in a circuit. As a result, each action in the action space can be represented as $(g, x)$, where $g$ is a gate in circuit $C$ and $x \in X$.}
}

\paragraph{Reinforcement learning (RL)} 
RL is a class of machine learning algorithms that train an agent to take actions in an environment while optimizing a given objective.
\if 0
\ZL{ Old Version:
In an RL task, the agent perceives the {\em state} of the environment ---
for a given state, the agent decides to take a particular {\em action} using an internal mechanism called {\em policy}, denoted as $\pi$.
There are two outcomes of an action: (1) the environment is influenced state may change; and (2) the environment also returns a {\em reward} signal.
}
\fi
An RL problem is formalized as a \mdp defined by a tuple $(S, A_s, P, r, \gamma)$, where $S$ is the state space of the environment, $A_s$ is the action space for the agent at state $s\in S$, $P(s'|s, a)$ defines the probability of the state transiting from $s$ to $s'$ if the agent takes the action $a$, $r(s',s, a)$ 
%\oded{why is $r$ denoted like a probability? shouldn't it be $r(s',s,a)$ or maybe $r(s,a,s')$?} \ZL{done}
defines the immediate reward of the transition from state $s$ to state $s'$ by action $a$, and $\gamma\in(0,1)$ is the discounted factor used to give more immediate rewards greater weight than future rewards.
%, and $\rho_0(\cdot)$ is a probability distribution over the initial states.
Each state-action pair is a {\em step}, and a sequence of steps is a {\em trajectory}, denoted as $\tau = (s_0, a_0, s_1, a_1, \dots)$.
The {\em return} of a trajectory 
$\tau = (s_0, a_0, s_1, a_1, \dots)$
is the discounted cumulative reward along the trajectory, given by
$R_{\tau} = \sum_{t = 0}\gamma^t r_t$, where $r_t=r(s_{t+1},s_t,a_t)$ denotes the reward at step $t$.
We sometimes use $R_t$ for the discounted cumulative reward starting from step $t$ in the trajectory, which is $R_t = \sum_{t'=t}\gamma^{t' - t}r_t$.
%\ZL{I modified the suffix sentence because I think it's not clear clear in terms of discount}
%\oded{I don't like the notations $R_\tau$ and $R_t$. They are not consistent with each other. For example, we get that $R_\tau$ is the same as $R_0$. Do we really need $R_t$ for anything?}
%\oded{I edited above and commented out the old stuff below. Please review.}
%
%Cumulative rewards along a trajectory is called {\em return}.
%Some RL tasks have an upper limit on the number of steps in a trajectory, denoted as $T$.
% while the agent iteratively interacts with the environment to maximize cumulative rewards, called the {\em return}.
%and the agent interacts with the environment continuously to maximize the cumulative reward, namely {\em return}.
%In this case, the return of a trajectory $\tau$ is defined as $R_{\tau} = \sum_{t = 0}\gamma^t r_t$, where $r_t=r(s_{t+1}|s_t,a_t)$ denotes the reward at step $t$, and the return starting from a specific step $t$ in the trajectory is defined as $R_t = \sum_{t'= t}\gamma^{t'-t}r_{t'}$.

RL agents aim to learn a {\em policy} to maximize the return through trial and error, by first collecting trajectories and then optimizing the policy based on the actions taken and the returns observed.
Formally, a {\em policy} is a function $\pi$ that maps each state to a  probability distribution over valid actions.
A policy induces a probability distribution over trajectories starting at a given state (i.e., trajectories {\em generated} by the policy), where $a_t$ is sampled according to $\pi(s_t)$ and $s_{t+1}$ is sampled according to $P(s_{t+1} | s_t, a_t)$.
A policy parameterized by $\theta$ is denoted as $\pi(\cdot|s; \theta)$ or  $\pi_\theta$. 
The learning objective in RL is to maximize  the expected return of trajectories generated by the policy $\pi_\theta$, i.e., $J(\theta) =  E_{\tau \sim \pi_\theta}[R_\tau]$. %, w.r.t. $\theta$ via gradient ascent.

Policy gradient methods~\cite{sutton1999policy} optimize this objective by gradient accent over $J(\theta)$ w.r.t. $\theta$. 
By the policy gradient theorem~\cite{sutton1999policy}, an equivalent formulation in the form of a loss function can be expressed as
%The gradient of this objective of policy gradient is
$L^{PG}(\theta) = E_{\tau \sim \pi_\theta}[\sum_t log\pi(a_t|s_t;\theta)A(s_t, a_t)]$, whose gradient is equivalent to $\nabla J(\theta)$ and can be estimated by sampling $\tau$. %The loss $L^{PG}$ is often called a surrogate loss. 
%\oded{we used to have a reference to this equivalent formulation. I think we need it because it's not trivial or something the reader should immediately understand.}
%\ZJ{@Zikun: can you take a look at this issue?}
%\yw{Remark: actually the two losses are not mathematically \emph{equivalent}. Instead, their \emph{gradients} are equivalent, i.e., $\nabla J=\nabla L^{PG}$, which is called the policy gradient theorem. The loss $L^{PG}$ itself makes no sense, and the only use of it is to compute the gradient $\nabla J$. That's why the PG loss $L^{PG}$ is often called a \emph{surrogate} loss in the literature. The original explanation of $L^{PG}$ by Zikun is problematic, so I removed them all. Putting a citation here should be sufficient. }
%, whose gradient is equal to that of $J(\theta)$ but has significantly lower variance if estimated with sampling.
% The gradient of $L^{PG}(\theta) = E_{\tau \sim \pi_\theta}[log\pi(a_t|s_t;\theta)A(s_t, a_t)]$, is equal to that of $J(\theta)$, as shown in the literature~\cite{}, but has significantly lower variance if estimated with sampling, so $L^{PG}(\theta)$ is adopted as the objective of policy gradient.
$A(s_t, a_t)$ in $L^{PG}(\theta)$ is the {\em advantage}, defined as $A(s, a) = Q^\pi(s, a) - V^\pi(s)$, where $Q^\pi(s, a) = E[R_\tau|s_0 = s, a_0 = a]$ is the Q-function (i.e., the expected return starting from state $s$ and taking action $a$), and $V^\pi(s) = E[R_\tau|s_0 = s]$ is the expected return starting from state $s$.
Advantage measures how much better the agent can get by taking a specific action at a state than average.
By gradient ascent, actions are {\em reinforced} based on their advantages (i.e. the probability of an action increases if its advantage is positive, and drops otherwise).
The empirical advantage $\hat{A}_t$ is often estimated as $\hat{A}(s_t, a_t) = R_t - V^\pi(s_t)$ over a sampled trajectory, where $R_t$ is the return, which is a sampled estimation of $Q^\pi(s_t, a_t)$, and %$V^{\pi_x}(C'(g'))^\pi(s_t)$ 
$V^\pi(s_t)$ is typically approximated by training another neural network. %estimates $V^\pi(s_t)$, which is typically computed by training another neural network. %~\cite{a3c}.
%A common neural architecture to implement policy gradient is the {\em actor-critic} architecture, where the actor is a neural network that instantiate the policy, and the critic is a neural estimator that computes $\hat{V}^\pi(\cdot)$.
Such a framework is also referred to as the \emph{actor-critic} method, where actor denotes the policy and critic denotes the value network~\cite{a3c}.

\paragraph{Proximal policy optimization (PPO)} PPO is a popular variant of policy gradient methods that achieve state-of-the-art performances in a wide range of applications~\cite{ppo}. 
%is a  policy gradient algorithm that aims to maximize policy improvement by exploiting different trajectories while not taking too large steps that can destruct training~\cite{ppo}.
%In vanilla policy gradient, collected data is only used in one step of policy update.
%The data is not fully utilized if the step is small, while the policy may step into unsafe region if the step is too large.
%To address this, 
PPO propose an clipped surrogate objective:
\begin{equation}
\label{eqn:ppo_objective}
L^{CLIP}(\theta) = \mathop{E}\limits_{\tau \sim \pi_{\theta_{old}}}\left[\sum_t min(\rho_\theta( s_t, a_t)\hat{A}_t, \er{clip}(\rho_\theta(s_t, a_t), \epsilon)\hat{A}_t)\right]
\end{equation}
where $\rho_\theta(s, a) = \frac{\pi_\theta(a|s)}{\pi_{\theta_{old}}(a|s)}$, $\hat{A}_t$ denotes the estimated advantage $\hat{A}(s_t, a_t)$ at step $t$ over trajectory $\tau$ generated by $\pi_{\theta_{old}}$, and $\er{clip}(\rho, \epsilon) = \min(1 + \epsilon, \max(1 - \epsilon, \rho))$ is the clip function used by PPO~\cite[Eq. 7]{ppo}.

\section{Challenges and High-Level Approach}\label{sec:challenge}
Applying reinforcement learning to optimize quantum circuits presents several challenges unique to this problem setting. 
This section presents these challenges and the key ideas \sys uses to overcome them.
%
%To optimize quantum circuits using reinforcement learning, \sys must overcome two key challenges. This section discusses the challenges and \sys's solutions to them.
%There are two key challenges we must address in optimizing quantum circuits using reinforcement learning.
%
%\oded{Maybe we can move the RL formulation of quantum circuit optimization to here.}
%
%\paragraph{Problem formulation.}\label{subsec:formulation}
%We first formulate quantum circuit optimization as an RL task by defining the task as an \mdp.
In the sequel, we assume a fixed gate set, set of equivalence-preserving transformations, and cost function.
We formulate quantum circuit optimization as a \mdp $(S, A_s, P, r, \gamma)$ as follows.
$S$ is the set of circuits over the given gate set.
For a circuit $C \in S$, $A_C$ includes all valid applications of transformations on $C$.
%\oded{old version: 
%For a given input circuit $C_0$, the states space $S$ includes all circuits functionally equivalent to $C_0$.
%For each circuit $C$ in the state space $S$, its action space $A(C)$ includes all valid applications of transformations on $C$.
%}
Applying a transformation (i.e., an action) $a$ to a circuit $C$ deterministically defines a new circuit $C'$; therefore we let $P(C'|C, a) = 1$ and $P(C''|C,a)=0$ for $C''\neq C'$.
%Third, taking a specific action (i.e., applying a transformation to a specific sub-circuit) on a certain circuit always yields the same new circuit, which defines the state transition probability function.
Finally, the reward function is given by the cost difference between the circuits before and after a transformation: $r(C, a)$ = \Call{Cost}{$C$} - \Call{Cost}{$C'$}.

%\oded{I think the rest of this paragraph should be moved to \Cref{sec:training}, or maybe it's already explained there and can just be removed from here}
%\ZJ{We can remove this paragraph. When collecting a rollout trajectory in training, we randomly sample an initial circuit $C\in S$ as the starting state for the RL policy, because our goal is to train a policy that can generalize well to unseen circuits.}
%Moreover, each trajectory is limited to at most $T$ steps.
% \ZL{each circuit functionally equivalent to it is a state, and all the equivalent circuits form the state space; (2) all valid applications of transformations on a circuit form an action space; (3) a state transition function is defined by the transformations; and (4) the reward of an action (i.e., applying a transformation) is the difference of cost between the circuits before and after the transformation.
% Therefore, we can formulate quantum circuit optimization as a reinforcement learning task where an RL agent optimizes a quantum circuit by iteratively applying transformations.}
%%%
% \oded{we should make the notation here consistent with Section 3 ($A(C)$ vs $A_C$ and similarly for $P$).}
% \ZL{Done}
% \ZL{Initial distribution, training episodes, stop condition:}

\commentout{
%\paragraph{RL formulation.}
Therefore, we use reinforcement learning to train an agent that optimize circuits by applying transformations over a number of discrete steps.
% Our goal is to train an RL agent that optimizes an input circuit $C_0$ in terms of cost by applying transformations over a number of discrete steps.
% For circuit $C$, the transformation set $X$ does not naturally form $C$'s action space $A_C$, since a transformation can be applied to multiple subcircuits in $C$.
At step $t$, the agent receives a circuit $C_{t}$  and selects a transformation $x_t$ from the transformation set $X$ with its policy.
Note that the agent also needs to select a subcircuit in $C_{t}$ with its policy to apply the transformation because a transformation can be potentially applied to multiple subcircuits in a circuit. 
%which is composed of a {\em gate selecting} policy $\pi_g$ and a {\em transformation selecting} policy $\pi_x$.
%At each time step $t$, the agent receives a circuit $C_t$ and selects a gate $g_t \in C_t$ following $\pi_g$ and a transformation $x_t \in X$ following $\pi_x$.
% The gate $g_t$ and transformation $x_t$ are chosen according to a {\em gate-selecting} policy $\pi_g(g_t | C_t)$ and a {\em transformation-selecting} policy $\pi_x(x_t | C_t, g_t)$ respectively. 
%%%
Each time the agent applies a transformation, it receives (1) a new circuit $C_{t+1}$ and (2) a scalar reward $r_t = \textproc{Cost}(C_t) - \textproc{Cost}(C_{t+1})$, where $\textproc{Cost}(.)$ is a user-provided cost metric for evaluating the performance of a circuit.
Example cost metrics include total gate count, {\tt CNOT} gate count, and circuit depth.
% A trajectory is an optimization procedure that starts from circuit $C_0$ and continues until some user defined stop conditions are reached.
% the agent reaches a final circuit and chooses not to apply any further transformations. \ZL{need to state more accurately}
% We define the return of a trajectory $\tau$ as the accumulated reward from every time steps in the sequence (i.e., $R_\tau = \sum_{t=0}^{T - 1}{r_t}$).
Thus, the objective of the quantum circuit optimization problem is $J(\theta) = E_{\tau\sim\pi_\theta}[R_\tau]$ where $\theta$ denotes the parameters of policy $\pi$.
}
\subsection{Challenge 1: Action Space}

Directly applying existing policy gradient methods to our setting requires the RL agent to learn a policy that can simultaneously select a transformation and a subcircuit to apply the transformation for a given circuit.
However, learning such a policy is challenging due to the very large action space.
For example, when learning to apply the 8,664 transformations discovered by Quartz~\cite{pldi2022-quartz} on a thousand-gate circuit, there can be up to millions of actions (i.e., possible applications of the transformations) for a given state (i.e., circuit).
The large action space degrades the training efficiency of the RL agent, since a training sample only directly updates the probability of a single action, and exploring a large action space requires a huge amount of training samples.
Moreover, the action space of a state depends on its graph structure, which changes at each step.

%The large action space significantly degrades the PPO training efficiency, since a training sample only directly updates the probability of a single action, and comprehending a larger action space requires more training samples.
%%%

\commentout{
\oded{I think we need to remove this paragraph or rewrite it. Currently I don't see how it's related to the challenge of the large action space. The whole reason we're using RL is to find long sequences that sometimes include cost increases. This subsection discusses a specific technical challenge in applying RL to our setting, not the challenge that makes us want to use RL.}
In addition, as introduced in \Cref{sec:intro}, discovering highly optimized circuits requires applying transformations that decrease performance as intermediate steps; however, learning when and where to apply performance-decreasing transformations is non-trivial, since they are only beneficial when combined with other transformations that further improve performance.
}
% \oded{I moved this up }Finally, the action space of a state highly depends on its graph structure and changes at each step.
%
%most transformations do not immediately improve performance, thus the agent has to learn the dependencies between transformations in long-sequence optimizations.
%\oded{The last sentence doesn't seem very clear here. Maybe we should elaborate a bit more.}
%\ZJ{Let me know if the rewrite makes it clearer.}

%\oded{Above, can we give some numeric examples? For example, how big would the action space be for optimizing a typical circuit from our evaluation? Also, is the fact that the action space dynamic and depends on the circuit also part of the challenge?}
%\ZJ{Added a paragraph to elaborate.}

\paragraph{Solution.}

For a circuit $C$, \sys decomposes its action space $A_C$ into two ``subspaces'': a position space $P_C$ and a transformation space $X$.
$P_C$ includes all gates in $C$, and $X$ contains all transformations. (Recall that a subcircuit matching a transformation can be determined by matching a single gate, see \Cref{sec:background}.)
Under this decomposition,
$A_C$ is a subset of $P_C \times X$,
and we can train separate policies for $P_C$ and $X$.
\commentout{
\oded{I shortened here and left what was here commented out below. Zhihao, please review and accept if you like it.}
}
\commentout{
\ZJ{May need to move this paragraph to a later section.}
\oded{I think the following can be cut significantly and moved to the previous section. Especially if we have a paragraph there about transformations. If we're tight for space, this really doesn't require more than a few sentences, and we could also cut the example.}
\ZJ{Yes, this has been moved to the background section.}
\ZJ{To be removed:
\sys uses a single gate to locate each application of the transformation on a circuit.
While in general graph transformations a single node cannot deterministically identify how a transformation is applied, this approach is feasible and highly desirable under the quantum setting for two reasons.
First, the source and target graphs of a transformation are connected, since a transformation with a disconnected source (or target) graph can be decomposed into multiple smaller transformations, each of which has a connected source (or target) graph.
Second, the in- and out-edges of multi-qubit gates are indexed.
For instance, the control and target qubits of a {\tt CNOT} are indexed 0 and 1, respectively.
This property guarantees that the orders of in- and out-edges are preserved in graph pattern matching.
% For multi-qubit gates the in-edges and out-edges are associated with different qubits and therefore we order them based on the qubit ids.
%%%
As a result, when a gate $g_x$ in transformation $x$ is mapped to a gate $g_C$ in circuit $C$, all other gates in the source graph of the transformation deterministically match gates in $C$ if the transformation is applicable.
In practice, we use the topologically first gate $g_x$ of a transformation $x$ to locate the transformation~\footnote{When there is a tie, we pick the gate with the smallest qubit id.}, that is, when applying $x$ to circuit $C$, if $g_x$ is matched with gate $g_C$, we say that we apply $x$ to $g_C$.
For example, the first transformation in \Cref{fig:long-sequence} is applied to the $T^\dag$ gate in $C_1$, and the remaining four transformations in the figures are applied to the left-most $CNOT$ gate in $S_2$, $S_3$, $S_4$, and $S_5$, respectively.
A key benefit of locating a transformation using a single gate instead of a subcircuit is to simplify the position space $P$, since $P$ is exactly all the gates in a circuit.
As a result, each action in the action space can be represented as $(g, x)$, where $g$ is a gate in circuit $C$ and $x \in X$.
%, and $(g, x)$ denotes a valid action if transformation $x$ can be applied to gate $g$. 
}
}

\subsection{Challenge 2: State Representation}
At the core of most RL algorithms is a representation of states as high-dimensional vectors.
%%%
However, unlike most RL tasks whose states have a uniform structure, a state in our setting is a quantum circuit whose size and topology depend on the input circuit and may change at each step during the optimization process.
Therefore, designing a uniform state representation for quantum circuits is challenging.
%%%
A straightforward approach would be to directly represent each quantum circuit as a high-dimensional vector.
%%%
However, due to the diversity of quantum circuits, this approach leads to learned representations that are highly tailored to the circuits used in training and do not generalize well to  unseen circuits.

\paragraph{Solution.}
A key insight for addressing this challenge is leveraging the {\em locality} of circuit transformations, that is, while the overall optimization strategy depends on the entire circuit, we hypothesize that the decision of applying a transformation at a gate can be largely guided by the local environment of the gate.
%%%
Based on this hypothesis, we design a neural architecture that relies on local decision making when selecting a gate to apply a transformation. 
%%%
In particular, \sys uses a $K$-layer graph neural network (GNN) to represent the $K$-hop neighborhood of each gate (see \Cref{def:k_hop_neighborhood}).
%%%
While each representation only encodes a local subcircuit, combining all gates' representations allows \sys to collectively represent an entire circuit. 
A key advantage of our approach is that the representations generated by the GNN is independent of the circuit size and thus can generalize to circuits at different time steps in the optimization process.
While our approach localizes the decision making, it still allows global circuit-wide guidance, since we fine-tune the RL agent (including the weights of the GNN) when optimizing a circuit (see \Cref{sec:training}). With this design, we aim to achieve a good balance between local and global decision making in the optimization process.
% The following sections  describes our neural architecture and RL training methodology in more detail.

\commentout{
\oded{I think we need to say something here about how our design allows some global view (unless you think it's not correct). How about the following:
While our approach localizes the decision making, it still allows some global circuit-wide guidance, since we fine-tune RL agent (including the the weights of the GNN) for each input circuit. With this desing, we aim to achieve a good balance between local and global decision making in the circuit optimization process.
}
\ZJ{I like this suggestion. Added this sentence to the end of the previous paragraph.}
}

\commentout{
Challenge 1: Action space
The action space is large and also not fixed (varies with the circuit structure).
Solution: split the choice into first choosing a gate and then choosing a transformation. Then, the second part is fixed.
%%%
Challenge 2: State representation
The classic RL and PPO would be to try to train some high-dimensional representation of an entire circuit.
The state in our setting is a quantum circuit that we represent as a graph. The graph size and structure is different for each input circuit, and at each step in the optimization process of each input circuit. Finding a uniform state representation is therefore not straightforward, and it might miss some generalization opportunities. Moreover, an important part of the action space, namely selecting a gate, is closely tied to the state: the state is a graph, and selecting a gate means selecting a node in that graph. Furthermore, while the overall optimization sequence depends on the entire circuit, we expect/hypothesize that often the decision is largely guided by the local environment of the affected node.
For these reasons, we desing a neural architecture that relies on \emph{local} decision making when selecting a gate to apply a transformation. The key idea is to train a graph neural network (GNN) that collects information from a $K$-hop neigborhood of a gate ($K=6$ in our experiments) and then assigns a score to that gate, that is meant to represent the potential accumelated reward that would result from selecting that gate for optimization. We note that while this localizes the decision making, our design still allows some global circuit-wide information to ultimately affect the decisions, since during the training process we fine-tune the weights in the GNN for each circuit, so there is some feedback from distant parts of the circuit that affects the GNN weights.
The next section explains how we combine our ideas for splitting the actions space, and for embracing a local view of the circuit graph into our neural architecture.
5.
Combining our choice to split the actions space to choosing a gate and a transformation, and our idea to localize the gate selection, we arrive at our novel neural design: hierarchical PPO.
....
Use the gate-selection network as the critic network for the transformation selection process.
5.1 ..
5.2 ..
}

\commentout{
\ZJ{Old version: TO BE REMOVED::
Directly applying PPO to learn a representation for an entire circuit (i.e., state) results in suboptimal performance and limited generalizability.
%%%
\sys is designed as a learning-based optimizer that can generalize to circuits with arbitrary number of gates and qubits.
%%%
However, if we directly use PPO to learn to represent a quantum circuit as a high-dimensional vector, this learned representation will be highly tailored to the circuits used in training and cannot generalize to previously unseen circuits.
\paragraph{Solution:} A key insight for addressing this challenge is leveraging the {\em locality} of quantum circuit transformations --- each transformation only affects a fixed number of locally connected gates.
%%%
Therefore, instead of directly learning a representation for an entire circuit, \sys learns to represent a local subcircuit for each gate. Specifically, for each gate $g$ in the circuit, \sys uses a $K$-layer graph neural network (GNN) to represent $g$'s $K$-hop neighborhood, which includes all gates within $K$ hops of $g$ and all edges between these gates. These per-gate representations are used by \sys to identify optimization opportunities and perform necessary transformations to realizes the optimizations.
Compared to directly representing an entire circuit, representing a local subcircuit for each gate has two advantages.
First, our approach significantly reduces the complexity of the RL task, since the learned representation only needs to capture optimization opportunities around a single gate instead of capturing circuit-wide optimizations. These gate-level representations can collectively identify potential optimizations of a circuit.
Second, our approach allows the learned RL-agent to generalize to previously unseen circuits.
}
}

\iffalse
\yw{I think there is a missing challenge of initial state distribution? or episode termination? I think we should also describe how we construct each PPO episode before moving on to the application of PPO.}

\yw{In addition, instead of stating ``hierarhical PPO'', we can simply state another challenge, i.e., ``policy representation'' after the state representation section. We have never describe the generation process for actions (sec 5 describes the details. You may need to state the high-level ideas in this section before moving into implementation/architecture details). After you clearly define ``action'', ``state (input encoding)'', ``output (how to generate an action)'', and ``episode'', applying PPO becomes trivial and you can summarize everything as your overall algorithm.}, 
\fi

\subsection{\sys's Approach}\label{subsec:approach}
Combining our solutions to the two challenges discussed above, we propose a hierarchical approach to optimizing quantum circuits using RL.
\sys's neural architecture is outlined in \Cref{fig:overview}.
%\sys applies transformations iteratively to optimize a input circuit.
%For each step, \sys takes in a circuit $C_t$.
The first stage in processing the current circuit $C_t$
is the {\em gate representation generator}, which is a graph neural network (GNN) that computes a learned vector representation for the $K$-hop neighborhood of each gate in $C_t$ (see \Cref{subsec:design:gnn}).
%%%
Next, based on these learned representations, \sys's {\em gate selector} chooses a gate $g_t$ using a learned {\em gate-selecting policy}, denoted  $\pi_g(\cdot|C_t;\theta_g)$, which is a probability distribution over all gates in $C_t$ (see \Cref{subsec:design:critic}).
%%%
Finally, the learned representation of $g_t$ is fed into \sys's {\em transformation selector}, which selects a transformation using a learned  {\em transformation-selecting policy}, denoted  $\pi_x(\cdot|C_t, g_t;\theta_x)$, which is a probability distribution over all valid transformations at $g_t$ (see \Cref{subsec:design:actor}).

The gate- and transformation-selecting policies are trained jointly in \sys with a combined actor-critic architecture.
Specifically, the actor network learns the transformation-selecting policy, and the critic network acts both as a value estimator for the transformation-selecting policy and as a predictor for the gate-selecting policy.
In other words, the transformation-selecting policy is the {\em only} policy whose parameters are updated by gradient ascent and whose actions (i.e., transformation selections) are reinforced.
The critic network learns to estimate the value of the state derived %\oded{not so clear to me what ``taken'' means here} 
by applying the transformation-selecting policy, which is the $K$-hop neighborhood of a gate.
The gate-selecting policy evaluates all gates with the value estimator of the transformation-selecting policy and selects (with high probability) a high-value gate.
Intuitively, 
the value estimator is suitable to form the gate-selecting policy, because high value indicates high optimization opportunity.

Compared to a straightforward PPO approach, where the actor learns both $\pi_g$ and $\pi_x$, and the critic learns to estimate the value of an entire circuit, our method has fix-sized state (the $K$-hop neighborhood) for both the actor and critic, and a fixed and relatively small action space for the transformation-selecting policy.
These advantages ultimately make our policies easier to train.
However, our specialized architecture requires a different advantage estimator from standard PPO, which we develop in \Cref{sec:training}.
%\ZJ{@Oded: I rewrote the first sentence of this paragraph, please check.}
%Next, we describe \sys's neural architecture in \Cref{sec:design} and the hierarchical PPO training methodology in \Cref{sec:training}.

\commentout{
\ZJ{To be removed:
%\paragraph{Hierarchical PPO.}
%%% Hierarchical PPO
We introduce {\em hierarchical PPO}, a variant of the proximal policy optimization (PPO) algorithm~\cite{ppo}, to decompose policy $\pi\big((g,s) | C ; \theta\big)$ into a gate-selecting policy $\pi_g(g | C; \theta_g)$ and a transformation-selecting policy $\pi_x(x | C, g; \theta_x)$, each of which involves a much smaller action space and thus easier to learn.
%The hierarchical PPO deals with a type of decision making model where selecting an action can be decomposed into making decisions multiple in different sub-action spaces that are interdependent.
%Not only is this model suitable for quantum circuit optimization, it is also descriptive in many combinatorial optimizations \cite{local_rewrite}.
%%% Key ideas of hierarchical PPO
% the key idea of hierarchical PPO is to solve a hard problem with a simple policy
%To address the aforementioned challenge, we propose {\em hierarchical PPO}, which decomposes policy $\pi\big((g,s) | C ; \theta\big)$ into a gate-selecting policy $\pi_g(g | C; \theta_g)$ and a transformation-selecting policy $\pi_x(x | C, g; \theta_x)$, each of which involves a much smaller action space and thus easier to learn.
%breaks down the hard-to-learn policy into two policies that are rather easy to learn, namely the gate selecting policy and the transformation selecting policy.
%These policies are easy to learn because they both focus on a single task.
In particular, the transformation-selecting policy only learns to choose the best transformation to apply for a given circuit and gate, while the gate-selecting policy only learns to choose a gate where optimization is likely to take place under the current transformation-selecting policy.
Another key insight of our hierarchical PPO algorithm is to use a single actor-critic architecture to train both policies.
Specifically, the actor network learns the transformation-selecting policy, and the critic network acts both as a value estimator for the transformation-selecting policy and as a predictor for the gate-selecting policy.
}
}
\section{\sys's Neural Architecture\label{sec:design}}

\begin{figure*}
    \centering
    \includegraphics[scale=0.28]{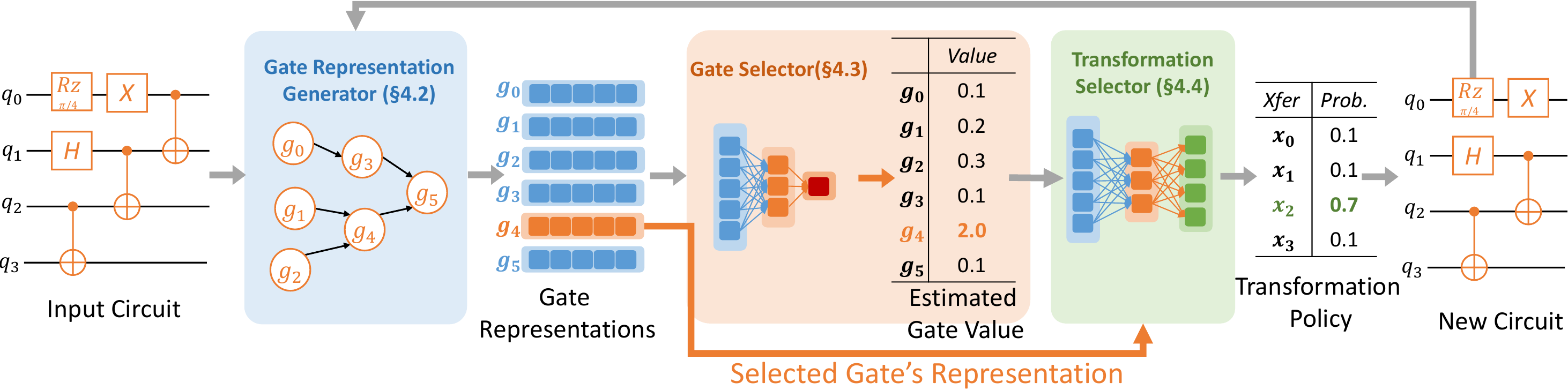}
    \vspace{-2mm}
    \caption{The neural architecture of \sys's RL agent. The arrows indicate the control flow.}
    \label{fig:overview}
\end{figure*}

We introduce \sys's gate representation generator in \Cref{subsec:design:gnn}, the gate selector in \Cref{subsec:design:critic}, and the transformation selector in \Cref{subsec:design:actor}.

% \subsection{Overview}

%%%

% \ZL{Generally, \sys is consisted of three components, a GNN-based network for node representation generation, a multiple layer perceptron(MLP) network for node value evaluation, and another MLP network for transformation selection. These three networks collectively form a RL agent which selects a node and a transformation to apply based on an input circuit.}
% The input to \sys is a quantum circuit to be optimized, which can be the initial input circuit or any circuits generated in the previous iterations.
%%%
\commentout{
\ZJ{To be removed:
\Cref{fig:overview} shows the neural architecture of \sys's RL agent.
To optimize a circuit, \sys first constructs its graph representation which is then fed into the {\em gate representation generator}.
The gate representation generator is a {\em graph neural network} (GNN) trained to compute a vector representation for each gate $g$.
The representation captures a subgraph that includes $g$ and its $K$-hop neighbors, where $K$ is the number of layers in the GNN.
%%%
Based on the learned representation, \sys's {\em \critic} predicts a value $V_g$ for each gate $g$, which is used to chose a gate to apply a transformation.
%used to form the gate selecting policy $\pi_g$ to select a gate and also acts as a baseline value for the training of the transformation selecting policy $\pi_x$.
$V_g$ is introduced in \Cref{subsec:design:critic}.
%expected return by selecting gate $g$ and following the learned transformation selecting policy $\pi_x$.
% A gate is selected by \sys based on the predicted values for all gates.
After a gate is chosen by the \critic, its learned representation is fed into \sys's {\em \actor}, which predicts a probability distribution across all transformations from which a transformation is sampled to apply.
%\sys selects a gate-transformation pair $(g, x)$ as an action to take in the next step. 
%\sys selects gates and transformations with two different networks, namely, the \critic network and the \actor network.
%To be specific, the \critic network first predicts the on-policy value $V_{\pi}$ for each nodes in the circuit, then \sys samples one node based on the predicted value.
%After that, \sys selects the representation of the specific node and feeds it to the \actor network which outputs a distribution over all transformations from which \sys samples a transformation to apply.
We introduce the graph neural network, the \critic and the \actor in \Cref{subsec:design:gnn}, \Cref{subsec:design:critic}, and \Cref{subsec:design:actor}, respectively.
}
}

\subsection{Gate Representation Generator\label{subsec:design:gnn}}

\if 0
\begin{figure}
    \centering
    \includegraphics[scale=0.4]{figs/gnn.pdf}
    \caption{\sys's GNN architecture.}
    \label{fig:gnn}
\end{figure}
\fi

%%%
%For each node $v\in G$, its {\em $k$-hop neighbor graph}, denoted as $G^{(k)}(v)$, is a subgraph of $G$ that includes all nodes within $k$ hops of $v$ (i.e., the distance between $u$ and $v$ in $G$ is at most $k$) as well as all edges between these nodes.

For an input circuit $C$, \sys uses a graph neural network (GNN) with $K$ layers to learn to represent the \emph{$K$-hop neighborhood} of each gate in $C$ as a high-dimensional vector. Our GNN architecture follows GraphSAGE~\cite{graphsage}.

\begin{definition}[$k$-hop neighborhood]
\label{def:k_hop_neighborhood}
For a node $v$ in graph $G$, its {\em $k$-hop neighborhood}, denoted $G^{(k)}(v)$, is the subgraph of $G$ that includes all nodes within $k$ (undirected) hops from $v$.
\commentout{
\oded{I think we can remove this( but if we keep it we should fix it to clarify that the distance is in the undirected graph---a node u can be in the 3-hop neighbourhood of v even if there isn't a path from u to v or from v to u):
and all edges between these nodes:
\begin{itemize}
    \item A node $u \in G^{(k)}(v)$ iff $u \in G$ and the distance between $u$ and $v$ is at most $k$.
    \item An edge $(u_1, u_2) \in G^{(k)}(v)$ iff $u_1 \in G^{(k)}(v)$ and $u_2 \in G^{(k)}(v)$.
\end{itemize}
}
}
\end{definition}

\oded{Explain that actually the representation we learn doesn't have access to all edges in the $k$-hop neighborhood.} \ZJ{I don't think I understand what Oded said here...}

%$k$-layer {\em graph neural network} (GNN) to learn a vector representation for a $k$-hop neighbor graph for every node of an input graph.
%\sys uses {\em graph neural networks} (GNN) for learning representations for different subgraphs of an input circuit.
%%%

\Cref{fig:graph_representation} shows the $k$-hop neighborhood for a selected gate (identified by $G^{(0)}(v)$) and different values of $k \in \{1, 2, 3\}$. 

The GNN architecture takes as inputs (1) a circuit represented as a graph, (2) gate-level features (i.e., the type of each gate), and (3) edge-level features (i.e., the direction of each edge and the qubits it connects to).
%a quantum circuit, represented as a graph, and its node and edge-level features as inputs and learns a high-dimensional representation for the $k$-hop neighbor graph for each node of the graph, where $k$ is the number of GNN layers in the GNN model.
%%%
%The learned representation is used for the downstream gate and transformation selection networks.
% \ZL{also worth mention here that for each edge in the circuit, we add an inverse }
%
Let $h_g^{(k)}, k \in \{1, \dots, K\}$ denote the representation of the $k$-hop neighborhood of gate $g$ outputted by the $k$th GNN layer, let $h_g^{(0)}$ denote the input features of the first GNN layer for $g$ (i.e., an embedding of $g$'s gate type), and let $e_{uv}$ denote the edge features of the edge between node $u$ and $v$.
%   $h_g^{(k)}$ denote the representation for the $k$-hop neighbor graph of gate $g$,
% The GNN architecture contains $K$ identical GNN layers. 
$h_g^{(k)}$ is computed by taking $h_u^{(k-1)}$ and $e_{ug}$ as inputs, where $u$ is a neighbor gate of $g$.
Each GNN layer includes an {\em aggregation} phase, which first gathers the representations of each gate's neighbors from the previous GNN layer, and an {\em update} phase, which computes a new representation for each gate by combining its previous representation and the aggregated neighborhood representations.
%\oded{I moved it above} The neural architectures of the two phases are adopted from GraphSAGE~\cite{graphsage}, a graph neural network that learns inductive representation on graphs.
%\Cref{fig:gnn} illustrates the GNN architecture.

%\ZJ{To be removed: The GNN model includes $k$ identical GNN layers, each of which gathers the representations of the neighbor nodes from the previous GNN layer and updates the representations of the node using DNN operators such as convolution and matrix multiplications. Specifically, let $h_v^{(k)}$ denote the learned representation of node $v$ at the $k$-th layer, and $h_v^{(0)}$ is the input node features of $v$.  Each GNN layer includes an {\em aggregation} phase that gathers the representations of the neighbor nodes from the previous GNN layer and an {\em update} phase that computes a new representation for each node. }

\paragraph{Aggregation phase.}
Let $N(g)$ denote gate $g$'s neighbors in the graph (i.e. a set of gates that share an in- or out-edge with $g$).
For each gate $g$, the aggregation phase at layer $k$ takes as inputs $h_u^{(k-1)}$ and $e_{ug}$ ($u\in N(g)$) and computes an aggregated representation of $g$'s neighbors $a_{g}^{(k)}$ with a multi-layer perceptron (MLP)~\cite{gardner1998artificial}.
The neural architecture of the aggregate phase is formalized as follows:
% The edge features are concatenated to the gate features during aggregation.
% Note that the edge features are also taken into consideration in the aggregation phase.
% Since gates in the graph may have different number of neighbors (i.e., a $k$-qubit gate has $2k$ neighbors), \sys uses average pooling~\cite{NIPS2017_5dd9db5e} as the aggregating function.
% Our implementation uses one linear and one ReLU layer for the MLPs in the aggregation phase, whose computation can be formalized as follows.
% The aggregation phase used in \sys can be formalized as follow:
\begin{equation}
    a_g^{(k)} = \sum_{u \in N(g)} \sigma\big(W_{a}^{(k)} \cdot \text{concat}(h_u^{(k-1)} , e_{ug}) + b_a^{(k)}\big)
\end{equation}
where $W_a^{(k)}$ and $b_a^{(k)}$ denote the weights of the MLP in the aggregation phase of the $k$-th layer, and $\sigma(\cdot)$ is the ReLU function~\cite{ReLU}.
% , and $e_{ug}$ represents features of the edge between $u$ and $g$.

%Overall, the aggregation phase can be formalized as follows.

%\ZL{The aggregation fomula: $a_{v}^{(k)} = average(\sigma(W^{(k)}concat(h_{u}^{(k-1)}, e_{uv})}), \forall u \in N(v)$}
%\ZL{here $e_{uv}$ represents the feature on the edge $u \rightarrow v$, we may use another letter other than $e$ to represent it..}

\paragraph{Update phase.}
For each gate $g$, the update phase computes a new representation $h_g^{(k)}$ by combining $g$'s representation from the previous GNN layer (i.e., $h_{g}^{(k-1)}$) and the aggregated neighbor representation (i.e., $a_{g}^{(k)}$). 
The neural architecture of the update phase is formalized as follows:
%, where we also employ multi-layer perceptron~\cite{gardner1998artificial} for updating the gates' representations. 
% Our implementation uses one linear layer with a ReLU activation for the MLP. The update computation can be formalized as follows:
%%%
\begin{equation}
    h_g^{(k)} = \sigma\big(W_u^{(k)}\cdot \text{concat}(h_g^{(k-1)} , a_g^{(k)})\big)
\end{equation}
where $W_u^{(k)}$ denotes the weight matrix of the MLP in the update phase.

Note that using more GNN layers allows \sys to represent a larger neighborhood of each gate but introduces more trainable parameters, which requires more time and resource to train (see \Cref{sec:training}).
\Cref{sec:eval} analyzes the choice of $K$.
% \ZL{Need more description of the trade-off: what is our inductive bias? larger observation size: richer information, slower forward time, larger noise}

%\ZJ{To be removed::}
%takes concatenates the aggregated representation of $v$'s neighbor (i.e., $a_{v}^{(k)}$) and $v$'s representation from the previous GNN layer (i.e., $h_{v}^{(k-1)}$) as inputs and updates the representation of the node using multi-layer perceptron~\cite{gardner1998artificial}.

% \ZL{The formula for update: $h_u^{(k)}=MLP(concat(h_u^{(k-1)}, a_u^{(k)}))$. The MLP we used is contains a single linear layer with bias, and a Relu activation.}
%%%
%Note that the output of the $k$-th GNN layer $h_{v}^{(k)}$ represents a subcircuit including all nodes within $k$-hop from $v$ as well as all edges between these nodes.
%Figure \ref{fig:graph_representation} shows the subcircuits represented by $h_{v}^{(1)}$, $h_{v}^{(2)}$, and $h_{v}^{(3)}$ for the example graph.
%%%

%\sys uses $K$ identical GNN layers for representing a $K$-hop subcircuit centered at $v$ for each node $v$, where $K$ is a hyper-parameter.
%Using more GNN layers allows \sys to represent a larger subcircuit for each node but requires more computation and becomes harder to generalize.
%We observe that $K=6$ achieves good performance for all experiments in our evaluation.

\subsection{Gate Selector\label{subsec:design:critic}}

The \critic is composed of two parts: a {\em gate value predictor} and a {\em gate sampler}.
The gate value predictor predicts the on-policy value $V^{\pi_x}(C, g)$ (see \Cref{subsec:update}) of gate $g$ on circuit $C$ for the transformation selecting policy $\pi_x$.
%\ZL{define $V^{\pi_x}(C, g)$ and explain its meaning.}
% $V^{\pi_x}(C, g)$ is defined as the expected return an RL agent can get by applying transformation
% in the local area of $g$ if it chooses to first apply a transformation at gate $g$ and continues applying transformations following the optimal gate selecting policy $\pi^{opt}_g$ and its transformation selecting policy $\pi_x$.
% \ZL{explain "local area"}
% By $\pi^{opt}_g$, we are referring to a gate selecting policy that always choose the gate with the maximal gate value.
% To be specific, the transformation selecting policy $\pi_x$ 
% $V^{\pi_x}(C, g)$ is defined as follows:
% which estimates the expected accumulated reward that can be achieved by selecting gate $g$ and then following the policy $\pi_x$ for selecting future transformations.
%%%
% A high $V^\pi(C, g)$ indicates that selecting gate $g$ to apply transformations could be desirable. 
%%%
The gate value predictor takes as an input $h^{(K)}_g$, which represents the $K$-hop neighborhood of $g$ on circuit $C$, and outputs $V^{\pi_x}(C, g;\theta_g)$, which approximates $V^{\pi_x}(C, g)$, and $\theta_g$ denotes the trainable parameters of the predictor.
The predictor uses a multi-layer perceptron (MLP)~\cite{hinton1987learning} in our current implementation.
% The predictor uses a multi-layer perceptron to compute an estimated $V^{\pi}(C, g)$, as depicted in \Cref{fig:overview}.
%Section \ref{subsec:update} describes the learning procedure for the gate value predictor.

%%%
Our gate selecting policy $\pi_g$ is formed by applying a temperature softmax~\cite{he2018determining} to the outputs of the gate value predictor, which is parameterized as follows:
\begin{equation}
    \pi_g(g | C; \theta_g) = \frac{exp\big(V^{\pi_x}(C, g; \theta_g)/t\big)}{\sum_{g'\in C}{exp\big(V^{\pi_x}(C, g'; \theta_g)/t\big)}}
\end{equation}
% We parameterize the gate selecting policy $\pi_g(g | C; \theta)$ as a temperature softmax~\cite{he2018determining} of the underlying $V^{\pi}(C, g; \theta)$:
% \begin{equation}
%     \pi_g(g | C; \theta) = \frac{exp\big(V^{\pi}(C, g; \theta)/t\big)}{\sum_{g\in C}{exp\big(V^{\pi}(C, g; \theta)/t\big)}}
% \end{equation}
where $\pi_g(g|C;\theta_g)$ denotes the probability of choosing $g$ in circuit $C$, and $t$ is a temperature parameter for the softmax function.
The temperature $t \in (0, +\infty)$ balances exploration and exploitation.
Specifically, when $t$ is larger, $\pi_g$ selects gates with increased randomness and becomes more explorative.
On the other hand, when $t$ becomes smaller, $\pi_g$ becomes more exploitative and tends to select the gate with the highest estimated value. 
Balancing exploration and exploitation across different circuits requires circuit-specific temperatures. 
For a specific circuit $C$, we set the temperature $t$ as
\begin{equation}
\label{eq:temperature}
    t = 1/ln\frac{\lambda(|C| - 1)}{1-\lambda}, \lambda \in (0, 1)
\end{equation}
where $|C|$ is the number of gates in circuit $C$ and $\lambda$ is a measure of exploitation.
Specifically, we set $t$ such that even if there is only one gate with a value closes to 1 (representing an optimization opportunity to reduce cost by 1), and the values of all other gates are close to 0 (representing no optimization opportunity), the \critic samples the high-value gate with probability $\sim \lambda$.

\oded{I think \Cref{eq:temperature} needs to be explained more.}

%\yw{do we want a remark here to briefly mention that the value function $V$ is \emph{local} due to the GNN representation? A remark here can make it more natural to explain the hierarchical advantage estimation later. }

% \ZL{The gate selection takes as input the representation of subcircuits and detects optimization opportunities.
% To be specific, the gate selection network act as an estimator of the on-policy value function for subcircuits.
% The on-policy value function for a subcircuit is defined as follows. 
% \begin{equation}
%     V^{\pi}(C) = \mathop{E}\limits_{\tau\sim\pi}[R(\tau)|C_0 = C]
% \end{equation}
% The on-policy value function gives the expected reward the agent can get in the future if it follows policy $\pi$ and starts at subcircuit $C$.
% }
% \ZJ{@zikun, do you think the above equation is accurate if $V^\pi(C)$ only captures local rewards?}
% %%%
% Recall that $h_{v}^{(K)}$ represents a subcircuit that includes all gates within $K$-hop of $v$ and their edges.
% \ZL{(We should mention policy $\pi$ earlier, maybe in the overview.)}
% The {\em gate selection network} takes the learned representation as an input and estimates the expected value of the subcircuit under policy $\pi$: $V_{\pi}(C)=\mathbb{E}[R_t | C_t = C]$, where $C$ represents the $K$-hop subcircuit for $v$.
% Intuitively, $V_{\pi}(C)$ indicates the expected improvement that can be achieved by following the policy $\pi$ to apply transformations on circuit $C$.

\subsection{Transformation Selector\label{subsec:design:actor}}
Given a circuit $C$ and a gate $g$, \sys's {\em \actor} chooses a transformation to apply at $g$.
%with the {\em transformation-selecting policy} $\pi_x$.
The \actor is an MLP, which takes as an input the representation of the selected gate $h^{(K)}_g$, and outputs a probability distribution $\pi_x(\cdot | C, g; \theta_x)$ over the entire set of transformations $X$, where $\theta_x$ denotes the trainable parameters of the transformation selector.
%The final activation of the MLP is a {\em masked softmax}.
The final output of the MLP is followed by a masked softmax layer.
The mask is generated by the circuit transformation engine by checking every transformation in $X$ to figure out which of them can be applied to gate $g$.
The masked softmax function filters out invalid transformations.

\section{Training and Inference Methodology \label{sec:training}}

%%% pseudo-code
\begin{algorithm}[t]
    \caption{\label{alg:hier_ppo}\sys's RL-based circuit optimization algorithm.
    %$$V_g$ and $\pi_x$ on line \ref{line:rollout} refer to the output of the gate value predictor $V_g(C_{t-1}^{(j)}, g_t^{(j)}; \theta_g)$ and the probability of choosing $x_t^{(j)}$ under the current transformation-selecting policy $\pi_x(x_t^{(j)}|C_{t-1}^{(j)}, g_t^{(j)}; \theta_x)$.
    $B$ and $T$ are hyper-parameters that specify the number of trajectories to collect in each training iteration and the maximum number of transformations in each trajectory.
    {\tt NOP} is a special transformation that stops the current trajectory when selected.
    \sys uses $\alpha$ to control data collection (see \Cref{subsec:data}).
    %\oded{We don't initialize $C_{best}$. We should initialize it or eliminate it entirely and just return $\arg\!\min_{C \in \m{C}} COST(C)$}    
    %\oded{We should add a check that the transformation is not NOP, and that the cost is not too high, as explained in the stop condition paragraph in 5.2}
    %\ZJ{Fixed.}
    }
    {
        \small
        \begin{algorithmic}[1]
            \State {\bf Inputs:} A circuit $C_{\er{input}}$
            \State {\bf Output:} An optimized circuit
            \State $\m{C} = \{C_{\er{input}}\}$ \Comment{\em $\m{C}$ is the initial circuit buffer}
            \For{iteration = 1, \dots}
            \State \Comment{Training data collection}
            \label{line:data_begin}
            \State $\m{R} = \emptyset$ \Comment{Clear the rollout buffer in each iteration}
            \For{j = 1, \dots, B}
            \State{Sample an initial circuit $C_0^{(j)}$ from $\m{C}$ (see \Cref{subsec:data})}
            \For{t = 1, \dots, T
            }
            \State{Compute gate representations $h^{(K)}_g$ for $g\in C_{t-1}^{(j)}$ }
            \State{Selects gate $g_t^{(j)}$ using $\pi_g(\cdot|C_{t-1}^{(j)};\theta_g)$}\label{line:gate_select}
            \State{Selects transformation $x_t^{(j)}$ using $\pi_x(\cdot|C_{t-1}^{(j)}, g_t^{(j)};\theta_x)$}\label{line:xfer_select}
%            \If{$x_t^{(j)} = {\tt NOP}$} \label{line:stop_nop}
            %\State {\bf break} \Comment{End the trajectory when selecting {\tt NOP}}
            %\EndIf
            \State{Generate new circuit $C_{t}^{(j)}$ by applying $x_t^{(j)}$ at $g_t^{(j)}$}
            \State{Compute reward $r_t^{(j)} = \Call{Cost}{C_{t - 1}^{(j)}} - \Call{Cost}{C_{t}^{(j)}}$}
            %\State{Estimate advantage with \cref{compute_adv} and collect data}
            \State{$\m{R} = \m{R} \cup \big\{
            \big(C_{t-1}^{(j)}, C_{t}^{(j)}, g_t^{(j)}, x_t^{(j)}, r_t^{(j)},
           V^{\pi_x}(C_{t-1}^{(j)}, g_t^{(j)}; \theta_g), \,
           \pi_x(x_t^{(j)}|C_{t-1}^{(j)}, g_t^{(j)}; \theta_x)\big)\big\}$}\label{line:rollout}
           % \State{\hfill
           % $V^{\pi_x}(C_{t-1}^{(j)}, g_t^{(j)}; \theta_g), \,
           % \pi_x(x_t^{(j)}|C_{t-1}^{(j)}, g_t^{(j)}; \theta_x)\big)\big\}$} \label{line:rollout}
            \If{$\Call{Cost}{C_{t}^{(j)}} \le \Call{Cost}{C_{0}^{(j)}}$} \label{line:initial_circuit_buffer_start}
                \State{$\m{C} = \m{C} \cup \{C_{t}^{(j)}\}$}
            \EndIf \label{line:initial_circuit_buffer_end}
            \If{$x_t^{(j)} = {\tt NOP} \vee \Call{Cost}{C_t^{(j)}} > \alpha \cdot \Call{Cost}{C_{\er{input}}}$}\label{line:stop_cost}\label{line:stop_nop}
                \State {\bf break} \Comment{End the trajectory}
            \EndIf\label{line:data_end}
            \EndFor
            \EndFor
            \State \Comment{Agent update}\label{line:update_begin}
            \State{Update $\theta_g, \theta_x$ using SGD 
            (loss given by \cref{loss}) for $M$ epochs}\label{line:update_end}
            \EndFor
        \State {\bf return} $\arg\!\min_{C\in\m{C}}{\Call{Cost}{C}}$
        \end{algorithmic}}
\end{algorithm}

To train \sys's neural architecture with PPO, \Cref{subsec:hae} introduces {\em hierarchical advantage estimator}, a novel approach to estimating the actions' advantages in our problem setting.
\Cref{alg:hier_ppo} lists \sys's RL-based optimization algorithm, which optimizes a circuit by training the RL agent.
A training iteration of \sys consists of two phases: {\em data collection}, which uses the current RL agent to generate trajectories for RL training (line \ref{line:data_begin}-\ref{line:data_end}), and {\em agent update}, which uses gradient descent to update the agent's trainable parameters (line \ref{line:update_begin}-\ref{line:update_end}).
The two phases are introduced in \Cref{subsec:data,subsec:update}, respectively.
% We introduce the hierarchical PPO algorithm in \Cref{subsec:hppo}, training data collection in \Cref{subsec:data}, and agent update in \Cref{subsec:update}.
%\sys also includes several optimizations to scale to large quantum circuits, which are described in \Cref{subsec:scalable}.
Finally, \Cref{subsec:pretrain} discusses \sys's combination of pre-training and fine-tuning to optimize an input circuit.

\subsection{Hierarchical Advantage Estimator}
\label{subsec:hae}
\if 0
\ZL{
mention our actor-critic architecture \\
say that we need to estimate the advantage of $\pi_x$\\
challenge is that this is local but not global, so we cannot use methods in original ppo\\
so we use our HAE
}
\fi

As described in \Cref{subsec:approach,sec:design}, \sys uses a combined actor-critic architecture to jointly train two policies: the gate-selecting policy $\pi_g$, which is directly approximated using the values $V^{\pi_x}(C,g)$; and the transformation-selecting policy $\pi_x$, on which we apply RL training with an adaptation of PPO.
%where the transformation-selecting policy $\pi_x$ is trained as the actor and the gate-selecting policy $\pi_g$ is trained as the critic.
%Estimating the policy gradient of $\pi_x$ requires an advantage estimator for it.
For the transformation-selecting policy $\pi_x$, its input is the embedding of the $K$-hop neighborhood of a gate $g$ on circuit $C$ (i.e., $C^{(K)}(g)$), and its output action is an applicable transformation $x$ at gate $g$.
To update $\pi_x$ using policy gradient, a key challenge is to estimate the advantage of applying transformation $x$.
In the canonical framework of PPO, the advantage $\hat{A}_t$ of applying $x$ at step $t$ over a sampled trajectory $\tau$ is estimated by the difference between the return $R_t$ and the value function, namely $\hat{A}_t=R_t-V^{\pi_x}(C,g)$.
However, such a straightforward approach can be problematic in our setting. The value function $V^{\pi_x}(C,g)$ is computed using the GNN embedding over the $K$-hop neighborhood of gate $g$, so it represents the \emph{local} value for the neighborhood of $g$ rather than the \emph{global} value of the entire circuit $C$. 
Note that when generating a trajectory $\tau$ during RL training, our hierarchical policy may choose an arbitrary gate $g$ according to $\pi_g$ to apply a transformation, so the following steps in $\tau$ after $x$ is applied to $g$ can involve gates that are arbitrarily far away from $g$. Accordingly, the trajectory return $R_t$ is in fact estimating the \emph{global} return over the circuit $C$ rather than the \emph{local} return over the $K$-hop neighborhood. Therefore, the \emph{multi-step} trajectory return $R_t$ may not be the most appropriate choice for advantage estimation.

In order to obtain an accurate \emph{local} advantage, we propose a \emph{hierarchical advantage estimator}, which is based on a \emph{1-step} return estimation given by
\begin{equation}
\label{eqn:hae}
%\begin{aligned}
\hat{A}(C, g, x) =  r(C, g, x) + \gamma\left( \mathop{max}\limits_{g'\in \Call{IG}{\ell, C, g, x}}V^{\pi_x}(C',g')\right) - V^{\pi_x}(C,g),
%\end{aligned}
\end{equation}
%
% Original Q-estimator
\if 0
\begin{equation}
\label{eqn:hae}
    \hat{Q}^{\pi_x}(C(g), x) = r(C, g, x) + \gamma\max_{g' \in \Call{IG}{\ell, C, g, x}}{V^{\pi_x}(C'(g'))},
\end{equation}
\fi
%\yw{I changed all the $\hat{V}$ to $V$ by removing the hat.}
%\yw{Also, the notation of $r$ seems inconsistent across the paper...}
where $r(C,g,x)$ is the reward of applying transformation $x$ to gate $g$ of circuit $C$, $C'$ denotes the new circuit obtained by applying the transformation, and $\Call{IG}{\ell, C, g, x}$ is the $\ell$-hop influenced gates of this transformation (defined blow) to ensure locality.
%(see \Cref{def:influenced_gates}).

\begin{figure}
    \centering
    \includegraphics[scale=0.29]{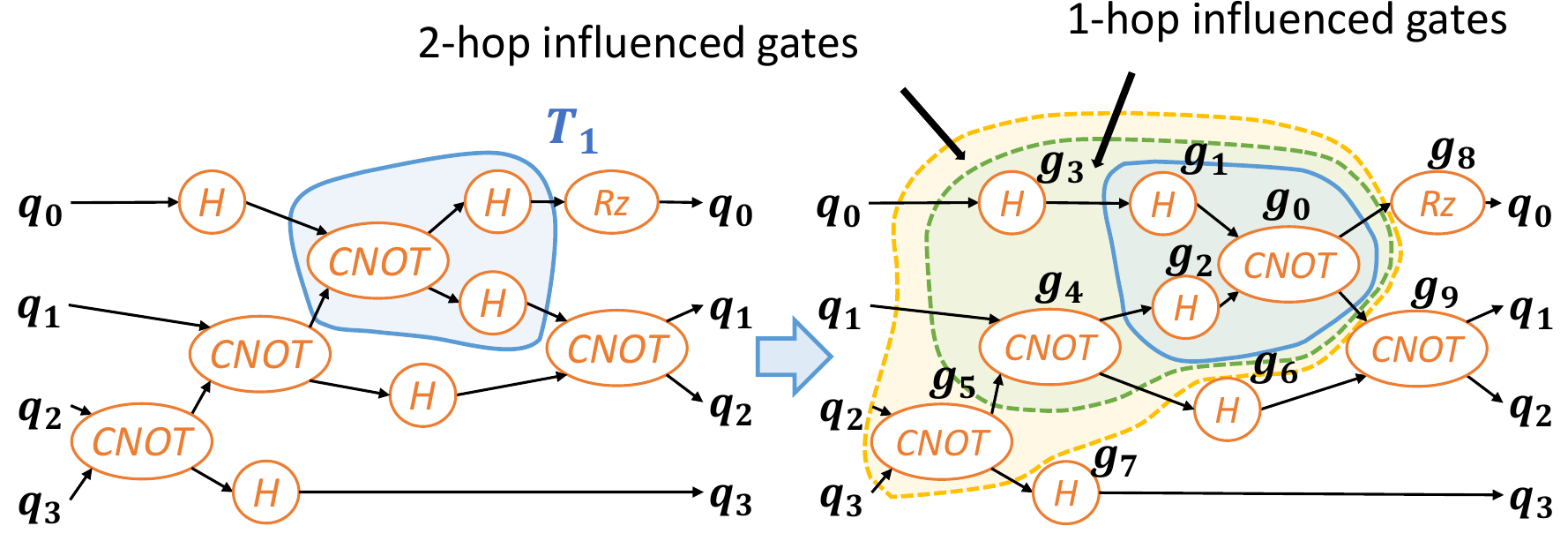}
    \vspace{-2mm}
    \caption{The $\ell$-hop influenced gates of applying the graph transformation in \Cref{fig:graph_subst}.}
    \label{fig:local}
\end{figure}

%\oded{we should change $l$ to $\ell$}

%(We note that while the original PPO algorithm~\cite{ppo, a3c} uses a multi-step advantage estimator, our estimator considers a single step.)

%\oded{The text from the opening of 5.1 until here  can be improved, Yi will make a pass.}

\begin{definition} [$\ell$-hop influenced gates]
\label{def:influenced_gates}
For a transformation $x$ applied at gate $g$ of circuit $C$, its \emph{$\ell$-hop influenced gates}, denoted $\Call{IG}{\ell, C, g, x}$, is gate set of the new circuit $C'$ that includes (1) all the new gates introduced by the transformation, and (2) all $\ell$-hop (directed) predecessors of these gates.
\end{definition}
%%%
$\Call{IG}{\ell, C, g, x}$ includes all gates in the new circuit $C'$ influenced by transformation $x$, which can fall two categories. First, all gates in the target graph of the transformation are in $\Call{IG}{\ell, C, g, x}$, since these gates are newly introduced by the transformation. For transformation $T_1$ in \Cref{fig:local}, $g_0$, $g_1$, and $g_2$ are added to the new circuit by applying $T_1$. Second, $\Call{IG}{\ell, C, g, x}$ also includes gates whose applicable transformations may change due to the transformation. For this category, we consider all $\ell$-hop predecessors of the new gates. \Cref{fig:local} shows the 1- and 2-hop influenced gates. 
\sys locates a transformation based on a topologically minimal gate in the source graph of the transformation.
Therefore for any gate not influenced by a transformation $x$, its applicable transformations remain the same after applying transformation $x$, as shown in \Cref{thm:influenced_gates}.

%\oded{do we explain anywhere that we locate gates based on a topologically minimal gate in their source graph? This is why we only look at predecessors and not the whole neighborhood, so we should explain that here.}

%whose applicable transformations may be influenced by transformation $x$. On the other hand, for any other gates in $C'$, their applicable transformations are not affected by $x$.
\begin{theorem}[Property of influenced gates]\label{thm:influenced_gates}
Let $C'$ be the new circuit obtained by applying transformation $x$ to gate $g$ on circuit $C$, and assume that for any transformation $(G,G')$ the depth of $G$ is at most $d$ (the depth of a directed acyclic graph is the maximal length of a path in the graph). For any gate $g'. g'\in C' \wedge g'\notin \Call{IG}{d, C, g, x}$, its set of applicable transformations is identical in $C$ and $C'$.
\end{theorem}
\begin{proof}
See supplementary material. \ZL{No appendix now} \ZJ{I think we can still submit a supplemental material.}
\end{proof}
%\oded{replace $Q$ below by advantage}\ZJ{done}

\sys uses influenced gates to capture dependencies between transformations when estimating their advantages. As per \cref{eqn:hae}, \sys estimates the local advantage of a transformation based on the maximum $V^{\pi_x}$ (calculated by the gate value predictor) across all influenced gates of the transformation.
%\yw{*****}
%The design is based on an observation that the application of optimizing transformations (with positive rewards) are often enabled by previous transformations applied to the same area. 
%\yw{I don't know if this sentence still makes a lot of sense now. As long as we state that the advantage is a local estimation for the neighborhood, it becomes natural to me to only propagate values from nearby gates. So this sentence doesn't fit well here. It better justifies the fact that we use $K=4$ (or the design choice of using local values rather than global), imo. Anyway, the next sentence looks good.}
The rationale behind the $max$ aggregator is to propagate the reward signal through the dependency path of transformations, so our hierarchical advantage estimator can capture the rewards that truly depend on the current transformation.
In our experiments, the maximal depth of circuits used in transformations is 4, but in practice we found that using 1-hop influenced gates leads to better results than using 4-hop influenced gates suggested by \Cref{thm:influenced_gates}. We note that while {\em some} transformations have a depth of 4, most have a lower depth; including more predecessors introduces more variance into the training process as the value of the maximal-value node in the influenced gates may not always be related to the applied transformation.

\oded{TODO: add experiments that check sensitivity to $\ell$ between 0 and 4 (or even $>4$ and also $\infty$, i.e., taking max over the whole graph).}

\commentout{
\sys uses the influenced gates to estimate the Q-value of a transformation. Specifically, \sys uses the maximum value of gates in $\Call{IG}{\ell, C, g, x}$ given by the critic network (i.e., the gate value predictor) as the value of the next state.
This formula uses the maximum value of gates in \local($C, g, x$) given by the critic network (the gate value predictor) as the value of the next state.
\ZJ{This sentence is broken: This design is based on the observation that, in a long-sequence optimization, transformations are dependent to their predecessors (i.e. some transformations before a particular transformation are its enablers).}
\ZL{This design is based on the observation that in the optimization process, the application of the cost-decreasing transformation is often enabled by previous transformations in the same local area.
% To be specific, a transformation $x_0$ enables another transformation $x_1$ if and only if $x_1$ is located on a gate in the local area influenced by $x_0$.
For example, in \Cref{fig:long-sequence}, cost-reducing transformation $T_4$ is enabled by $T_1$, $T_2$ and $T_3$.
In this case, the reward generated by the cost-decreasing transformation depends on these enablers, and thus it should be counted into the Q-value of the enablers.
}
The rationale behind the $max$ aggregator in the equation is that we want the reward signal to propagate through the dependency path of transformations and be reflected, so that this Q-value estimator can capture the rewards that truly depend on the current transformation $x$.
}

\commentout{
Note that since we concentrate on the transformation-selecting policy here, the definitions of states and actions are distinct from what is stated in \Cref{sec:challenge}.
Instead of the entire circuit, the states taken by the transformation-selecting policy is the $K$-hop neighborhood of a gate $g$ on circuit $C$.
The action space of the transformation-selecting policy is the transformation set, a sub-space of the original action space.
To state it formally, the advantage to be estimated is $A(C, g, x)$ where circuit $C$ and gate $g$ is used to locate a sub-graph that forms a state, and transformation $x$ is the action.
Here comes the problem: how to estimate the advantage of the transformation-selecting policy given trajectories collected with two policies hierarchically?
To address this problem, instead of using the multi-step advantage estimators in the original PPO~\cite{ppo, a3c}, we propose the hierarchical advantage estimation, which estimates the advantage by estimating the Q-value, that is $\hat{A}(C, g, x) = \hat{Q}^{\pi_x}(C, g, x) - \hat{V}^{\pi_x}(C, g)$.
The Q-value is estimated as follows:
% The surrogate value $V'$ is defined as
\begin{equation}
    \hat{Q}^{\pi_x}(C, g, x) = r(C, g, x) + \gamma\max_{g' \in \Call{IG}{C, g, x}}{V^{\pi_x}(C', g'; \theta_g)}
\end{equation}
where \Call{IG}{\cdot} is a abbreviation of  \local(\cdot), denoting the set of gates in the local area that is influenced by the transformation, which is formally defined below.
This formula uses the maximum value of gates in \local($C, g, x$) given by the critic network (the gate value predictor) as the value of the next state.
This design is based on the observation that, in a long-sequence optimization, transformations are dependent to their predecessors (i.e. some transformations before a particular transformation are its enablers).
% To be specific, a transformation $x_0$ enables another transformation $x_1$ if and only if $x_1$ is located on a gate in the local area influenced by $x_0$.
\ZL{Example mighe be changed due to change fo figure.For example, in \Cref{fig:long-sequence}, transformation $T_2$ depends on $T_1$, because $T_2$ cannot be applied if $T_1$ have not changed that area.}
The rationale behind the $max$ aggregator in the equation is that we want the reward signal to propagate through the dependency path of transformations and be reflected, so that this Q-value estimator can capture the rewards that truly depend on the transformation.
% The key idea is that the estimated value $V'(C, g, x)$ reflects the influence and optimization opportunities that the transformation $x$ brings about to the local area.

\begin{figure}
    \centering
    \includegraphics[scale=0.3]{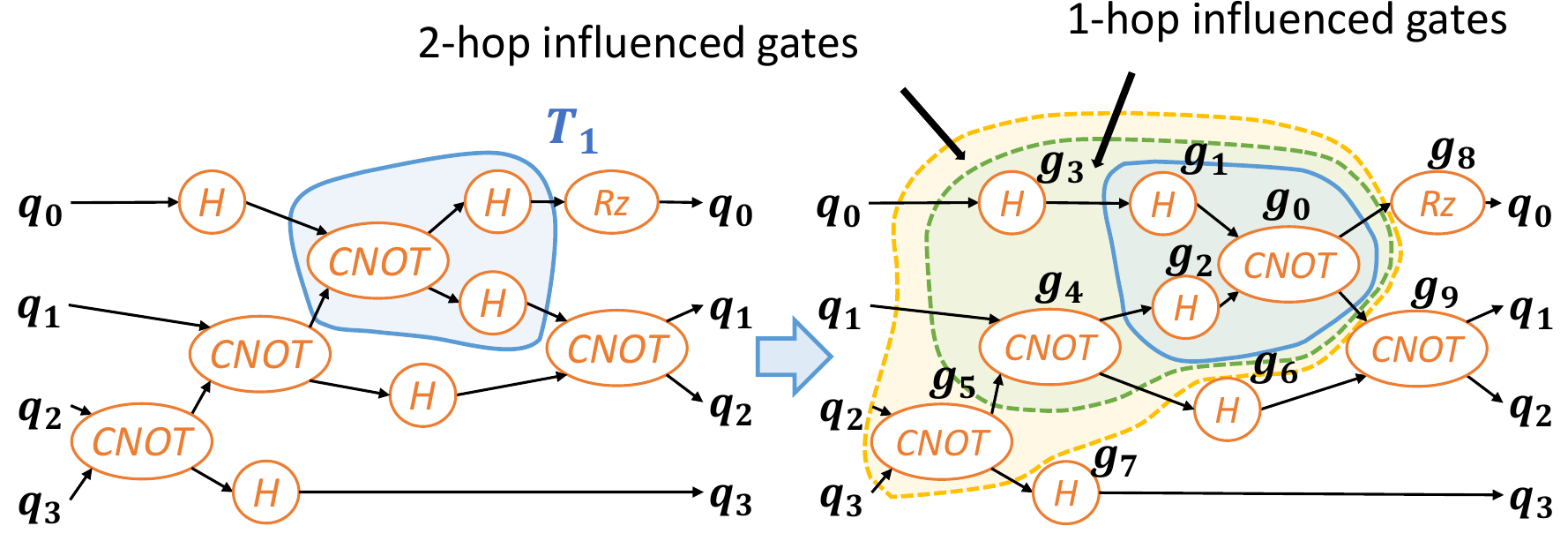}
    \caption{The $l$-hop influenced gates of applying the graph transformation in \Cref{fig:graph_subst}.}
    \label{fig:local}
\end{figure}
Now we define \Call{IG}{\cdot} formally.
Two types of gates are influenced by a transformation.
The first are the new gates that appear in the destination graph of the transformation (i.e. the substitutes).
There is no doubt that this kind of gates should be counted into \Call{IG}{\cdot}.
As shown in \Cref{fig:local}, for transformation $T_1$, gates $g_0$, $g_1$, $g_2$ belongs to the first kind.
The second are the gates which themselves are not changed but there are modifications occur in their neighborhood that leads to changes of validity of transformations on them.
Some of these changes open up new optimization opportunities.
For example, in \Cref{fig:local}, after transformation $T_1$, an H-gate cancellation can be performed on gate $g_1$ and $g_3$.
Nevertheless, we do not have a quantitative measurement of the influence on these gates. 
To decide which gates should be included into \Call{IG}{\cdot}, we follow the idea that the closer they are to the new gates, the greater they are influenced.
% Thus, for these gates we follow the idea that the closer they are to the new gates, the greater the influences are.
Specifically, we view gates in the $l$-hop neighbor of the new gates as "close", where $l$ as a hyperparameter, and exclude gates outside this area from \Call{\local}{\cdot}.
As an example, gate $g_3$ and $g_4$ in \Cref{fig:local} are 1-hop neighbor of the new gates, and $g_5$ is their 2-hop neighbour.
Also note that the first topology-order gate is used to locate a transformation, so successors of the new gates are never influenced by a transformation.
For instance, the validity of transformations on gate $g_8$ and $g_9$ in \Cref{fig:local} is not influenced by $T_1$.
Thus, successors of new gates are excluded from \Call{\local}{\cdot}.
Overall, the formal definition of \Call{\local}{\cdot} is as follows:
\begin{definition} [$l$-hop influenced gates]
For a transformation $x$ located on gate $g$ of circuit $C$, the $l$-hop influenced gates, denoted as \Call{IG$_l$}{$C, g, x$}, is a set of gates that includes the new gate introduced by $x$ and the $l$-hop predecessors of these gates, where $l$ is a hyperparameter.
\end{definition}
For $T_1$ in \Cref{fig:local}, the 1 hop dependent gates include $g_0$, $g_1$, $g_2$, $g_3$ and $g_4$, labeled as ``1-hop influenced gates'', and the 2 hop dependent gates include gates in the 1 hop dependent gates and gate $g_5$, labeled as ``2-hop influenced gates'' in the figure.
\begin{theorem}[Property of influenced gates]\label{thm:influenced_gates}
Assuming all transformations considered by \sys have at most $L+1(L\geq0)$ gates. 
For any gate g' in C' that is not in the ($L$)-hop influenced gates of transformation $x$ applied on gate $g$ of circuit $C$ (\Call{IG$_L$}{C, g, x}), its set of applicable transformation is identical to what it was in $C$.
\end{theorem}
Please see \Cref{prf:1} for the proof of \Cref{thm:influenced_gates}.
}
\commentout{
Overall, the {\em advantage estimator} is given by:
\begin{small}
\begin{equation}\label{compute_adv}
\begin{aligned}
\hat{A}(C(g), x) =  r(C, g, x) + \gamma \mathop{max}\limits_{g'\in \Call{IG}{\ell, C, g, x}}V^{\pi_x}(C'(g')) - V^{\pi_x}(C(g))
\end{aligned}
\end{equation}
\end{small}
}
% \ZL{
%     justification for max
%     % Argue against drawbacks: 1) max aggregation is not able to reflect multiple optimization opportunities, an justification: we form a lower bound 2) some nodes are not influenced by xfer but included in local(), a justification: trade-off
% }

% \ZJ{Not sure if this is the right place to mention this}
% Note that both the node reward predictor and transformation selector take $h^{(K)}(v)$ as inputs and make prediction based on the {\em local} $K$-hop neighbor graph of each node $v$, ignoring any nodes more than $K$ hops away from $v$. We discuss alternative design choices in Section xxx.

\subsection{Training Data Collection\label{subsec:data}}
\paragraph{Generating Trajectories.}
% Data collected along the way
As shown in \Cref{alg:hier_ppo} line \ref{line:data_begin}-\ref{line:data_end}, in each training iteration, \sys generates a total of $B$ trajectories according to the current policy, and each trajectory is limited to at most $T$ steps.
%generates $J$ trajectories 
These generated trajectories are taken as the training dataset to update the gate- and transformation-selecting policies.
To generate the $j$-th trajectory, \sys randomly selects an initial circuit $C_0^{(j)}$ from the {\em initial circuit buffer} $\m{C}$ (introduced later in this section), and iteratively applies transformations on the selected circuit.
Specifically, at the $t$-th time step, \sys takes the current circuit $C_{t-1}^{(j)}$ as an input and chooses a gate $g_t^{(j)}$ and a transformation $x_t^{(j)}$ using the gate- and transformation-selecting policies. After that, $x_t^{(j)}$ is applied to $g_t^{(j)}$ on $C_{t-1}^{(j)}$ to generate a new circuit $C_t^{(j)}$.
At each time step, \sys collects the following data:
%For the $j$-th trajectory, at time step $t$, \sys takes in a circuit $C_t^{(k)}$, chooses a gate $g_t$ with the \critic and chooses a transformation $x_t$ with the \actor.
%After that, $x_t$ is applied to $g_t$ on $C_t$, resulting in a new circuit $C_{t+1}$.
%After taking a step, \sys collects the following data:
(1) the current and new circuits $C_{t-1}^{(j)}$ and $C_{t}^{(j)}$,
(2) the selected action $(g_t^{(j)}, x_t^{(j)})$,
(3) the immediate reward $r_t$,
(4) the value of $g_t^{(j)}$ given by the gate value predictor $V^{\pi_x}(C_{t-1}^{(j)}, g_t^{(j)}; \theta_g)$ (see \Cref{subsec:design:critic}), and
(5) the probability of choosing $x_t^{(j)}$ under the current transformation-selecting policy (i.e., $\pi_x(x_t^{(j)}|C_{t-1}^{(j)}, g_t^{(j)}; \theta_x)$.
% (7) gates in the new circuit that are influenced by the transformation, given by $\Call{Local}{(C_t, g_t, x_t}$; (8) the value of $g_t$ given by the node reward predictor $V_t = V_{\theta}(g_t)$; (9) value used to estimate "next state" value given by $max({V}_{\theta}(C_{t+1}, g') \forall g' \in \Call{Local}{C_t, g_t, x_t}$.
The collected data is saved in a {\em rollout buffer} $\m{R}$, which is initialized to empty at the start of each iteration.

\paragraph{Stop conditions.}
After a trajectory begins, \sys keeps applying transformations on the current circuit until one of the following stop conditions is satisfied.
%There are several stop conditions for a trajectory.
First, each trajectory can have at most $T$ time steps, where $T$ is a hyper-parameter.
Second, \sys introduces a special transformation named {\tt NOP} into the transformation set.
Selecting {\tt NOP} as the transformation (\Cref{alg:hier_ppo} line \ref{line:stop_nop}) indicates that \sys chooses to end the current trajectory, which provides \sys the ability to stop when it finds that moving forward cannot bring benefits or even leads to negative rewards.
% Third, the data collection is also controlled by a hyper-parameter $\alpha$. 
Third, \sys stops the current trajectory when the cost of the current circuit is $\alpha$ times greater than the cost of the input circuit (\Cref{alg:hier_ppo} line \ref{line:stop_cost}), which prevents \sys from moving toward a wrong direction too far. We set $\alpha$ to be 1.2 in our evaluation.
%By doing this it prevent the agent by moving toward the false direction too far, wasting computation and time.

%%% initial circuit buffer
\paragraph{Initial circuit buffer.} Instead of always starting from the input circuit $C_{input}$ in a trajectory, \sys samples a circuit from an {\em initial circuit buffer} $\m{C}$ (\Cref{alg:hier_ppo} line~\ref{line:initial_circuit_buffer_start}-\ref{line:initial_circuit_buffer_end}). $\m{C}$ includes all circuits discovered in previous trajectories whose cost is lower than the trajectory's initial circuit. 
$\m{C}$ is a hash map with keys being the cost and values being the set of circuits with that cost.
%\sys randomly selects an initial circuit from $\m{C}$ in each trajectory.
%%%
To select a circuit from $\m{C}$, \sys first samples cost and then uniformly selects a circuit from the set of circuits with the sampled cost.
Users can define customized probability distributions over costs depending on how they expect their agent to behave.
Specifically, a greedier distribution where probabilities for sampling lower costs are larger makes the agent more progressive in doing optimization, preferring to extend lower-cost states.

%\oded{we should explain what distribution is used. I think we should explain this here rather than in \Cref{sec:eval}, and then I don't think we really need the ``Initial circuit buffer'' paragraph there. The details about the hash map don't seem so interesting.}
Compared to always starting a trajectory from $C_{input}$, sampling circuits from the initial circuit buffer enhances the exploration of the search space.
In particular, starting from $C_{input}$ restricts \sys to only explore circuits at most $T$ steps away from $C_{input}$, where $T$ is the maximum number of steps in a trajectory. In contrast, starting from a randomly selected circuit allows \sys to explore previously unknown areas, enabling \sys to discover more optimized circuits.

\subsection{Agent Update\label{subsec:update}}
To update \sys's neural architecture, we first traverse the rollout buffer, and compute the advantage value for each time step according to \cref{eqn:hae}.
Next, we train the graph embedding network, the node reward predictor network and the transformation selector network jointly with stochastic gradient descent for $M$ epochs.
% According to PPO~\cite{ppo}, in each epoch, \sys re-evaluate the probability for selecting the same transformation at the same circuit and same gate for data on every timestep in the rollout buffer.
% To be specific, let the policy \sys uses in data collection be $\pi_{\theta_{old}}$, and policy \sys uses now be $\pi_{\theta}$.
% For each timestep $t$, \sys computes $\pi_{\theta}(x_t|C_t, g_t)$.
Following the PPO algorithm in ~\cref{eqn:ppo_objective}, the loss function for the transformation selector can be rewritten as 
\begin{equation}\label{loss_ts}
    L^{TS}(\theta_x) = E_\tau\left[ \sum_t \min(\rho_t(\theta_x)\hat{A}_t, \er{clip}(\rho_t(\theta_x), \epsilon)\hat{A}_t)\right]
\end{equation}
where $\rho_t(\theta_x) = \frac{\pi_x(x_t|C_t, g_t; \theta_x)}{\pi_x(x_t|C_t, g_t; \theta_{x(old)})}$. %, $clip$ is a function that clips $\rho_t(\theta_x)$ into the interval $[1-\epsilon, 1+\epsilon]$.
%\oded{I still don't understand what the clip function is. Its current explanation uses concrete arguments, which is strange. Do we mean to say that $clip(x,y) = \min(1 + y, \max(1 - y, x))$? If so, we should say it. But anyway I find the above equation very strange. Specifically, is sometimes uses $\rho_t(\theta_x)\hat{A}_t^2$ and sometimes $\rho_t(\theta_x)\hat{A}_t$, which really doesn't make sense to me. Is there a typo?}
%\yw{@Zikun, To be frank, we don't need to dive into the algorithmic detail of PPO. It isn't important imo.}
%%%
To update the value estimator network, which is also the gate selecting policy, \sys minimizes $L^{VE}(\theta_g)$ given by
\begin{equation}
    L^{VE}(\theta_g) = E_{\tau}\left[\sum_t \hat{A}(C_t, g_t, x_t)^2\right]
\end{equation}

To train the networks jointly, we combine these loss functions into
\begin{equation}\label{loss}
    L(\theta) = L^{TS}(\theta_x) + c_1 L^{VE}(\theta_g) + c_2 H(\pi_x;\theta_x) %\hat{E}_t[S[\pi_x](C_t, g_t; \theta_x)]
\end{equation}
where $\theta$ denotes all the trainable parameters combining $\theta_g$ and $\theta_x$, $c_1 \geq 0$ and $c_2 \geq 0$ are two coefficients, and $H(\pi_x)$ denotes the entropy of policy $\pi_x$, which serves as a regularization term to promote exploration. 

\subsection{Pre-training, Fine-tuning, and Policy-guided Search \label{subsec:pretrain}}

% \ZJ{Consider updating the subsection header to ``Pre-training, fine-tuning, and policy-guided search''}
%\oded{explain the changes to \Cref{alg:hier_ppo} in pre-training (initialization, sampling of trajectory starting points, and no output)}

Circuits implemented in the same gate set share common local topologies, which offers opportunities to transfer the learned optimizations from one circuit to another in the same gate set.
This observation motivates our pre-training and fine-tuning approach through which we can avoid training from scratch on previously unseen circuits and accelerate training.
As introduced in \Cref{sec:design}, \sys's neural architecture can automatically adapt to circuits with different sizes and/or topologies, which enables the pre-training and fine-tuning pipeline.
%\Cref{subsec:ablation} evaluates \sys's pre-training and fine-tuning mechanism.

\paragraph{Pre-training}
% A key advantage of pre-training and fine-tuning is avoiding training a new agent from scratch for every circuit and saving optimization time for large circuits.
%In pre-training, to improve the generalization, a set of circuits that are functionally different are used, providing the agent with diverse states and preventing over-fitting.
The goal of the pre-training phase is to help the agent learn how to optimize different local structures.
\sys uses a diversity of circuits for pre-training, which allows the agent to explore various local structures and prevents over-fitting to the structure of a single circuit.
%In \sys, we approach this by the including multiple circuits to let the RL agent explore diverse states and prevent over-fitting to the structure of a single circuit.
%In the initial buffer, circuits are clustered into several equivalent groups based on their equivalence.
Minor changes are needed to support training \sys on multiple circuits.
In the initial buffer, circuits are clustered into equivalent groups.
At the beginning of each trajectory, \sys chooses the initial circuit of a trajectory by first uniformly sampling an equivalent group and then randomly selecting a circuit from the group.
All trajectories are used in gradient estimations to update the agent.

%At the beginning of each trajectory, \sys first uniformly samples a equivalent group and then selects an initial circuit from the sampled group.

\paragraph{Fine-tuning}
\sys optimizes a new circuit by fine-tuning the pre-trained model on the circuit, which allows \sys to discover optimizations specific to the new circuit.
During the fine-tuning, there is a single equivalent group in the initial buffer, which corresponds to the new circuit \sys fine-tunes on.
Optimized circuits may be discover during the fine-tuning, however, the primary goal of fine-tuning is to make the agent fit to the new circuit, specifically, the model should learn the gate values and policy distribution specific the new circuit.
Once the model has fitted to the new circuit, we should stop the fine-tuning to avoid wasting computation and use the fine-tuned model to apply a more efficient {\em policy-guided search} to finish the rest of the optimization.
% The best circuit discovered during fine-tuning is returned as the optimized circuit by \sys.
%During fine-tuning, we take pre-trained model, and fine-tune it on a single circuit.
%In this case, all circuits in the initial buffer are equivalent. 
%Each trajectory samples an initial circuit from it and starts running until stopped.

% \ZL{Describe our search algorithm: its difference with trajectory collection}
\paragraph{Policy-guided search}
In this stage, we employ a circuit optimization model that has been fitted to the circuit at hand. 
We use this model to guide the further optimization of the circuit using a technique we call {\em policy-guided search}, which shares some similarities with the trajectory collection stage during training. 
However, there are a few key differences between policy-guided search and trajectory collection.
%%%
First, policy-guided search only maintains circuits with lowest cost in the initial circuit buffer, ensuring that \sys uses one of the best discovered circuits to start a trajectory.
%%%
Second, after selecting a gate using the gate-selecting policy, instead of sampling a transformation from the transformation-selecting policy, \sys selects the transformation with highest probability.
During the search, if the transformation-selecting policy triggers a stop condition (described in \Cref{subsec:data}) when select a transformation for a gate $g$ on circuit $C$, \sys applies a {\em hard mask} to $g$, which will prevent \sys from revisiting $g$.
%This allows \sys to apply {\em hard masks} to certain gates during the search, preventing \sys from revisiting them.
%These hard masks are applied when a specific gate $g$ on a circuit $C$ satisfies the stop conditions described in \Cref{subsec:data} for a given transformation $x$. 
Furthermore, to prevent \sys from exploiting the policy without exploration, we also apply {\em soft masks} to gates that have been visited by \sys for at least once.
The difference between hard and soft masks is that once all gates in a circuit have been masked out, either by hard or soft masks, we remove the soft masks and do not reapply them.
The purpose of soft masks is to ensure that every gate in a circuit $C$ is visited at least once if the circuit is visited at least $|C|$ times.
Whenever a circuit with lower cost is discovered, \sys clears the initial circuit buffer, adds the circuit to the initial circuit buffer, and restarts the search.

\paragraph{Optimizing a circuit}
To optimize an input circuit, \sys runs a fine-tuning process and a policy-guided search process simultaneously.
\Cref{fig:ft_search} shows the interactions between the two processes, which exchange information whenever one process discovers a circuit with new lowest cost.
%The two processes exchange information about the lowest cost circuit.
Specifically, if a circuit $C_1$ with new lowest cost is found during fine-tuning, \sys restarts the search using $C_1$, as shown in \Cref{fig:ft_search_a}.
Conversely, if a circuit $C_2$ with new lowest cost is found during the search, the fine-tuning process will include $C_2$ in its initial circuit buffer and continue fine-tuning, as shown in \Cref{fig:ft_search_b}.
Moreover, a timeout is set for the search process in case that its model is obsolete.
If no new lowest cost circuit is discovered before timeout, \sys restarts the policy-guided search using the most recently model from fine-tuning, as shown in \Cref{fig:ft_search_c}.

\begin{figure}
  % \centering
  \subfloat[Fine-tuning discovers a\\circuit with new lowest cost.] {
    \label{fig:ft_search_a}
    \includegraphics[width=0.3\textwidth]{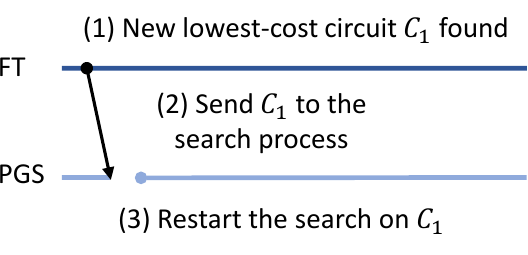}
  }
  \hspace{0.01\textwidth}
  \subfloat[Policy-guided search finds a \\circuit with new lowest cost.] {
    \label{fig:ft_search_b}
    \includegraphics[width=0.3\textwidth]{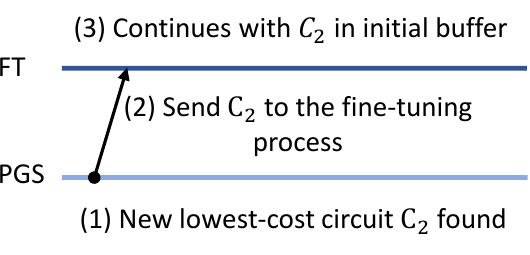}
  }
  \hspace{0.01\textwidth}
  \subfloat[No processes find a better circuit before search times out.] {
    \label{fig:ft_search_c}
    \includegraphics[width=0.3\textwidth]{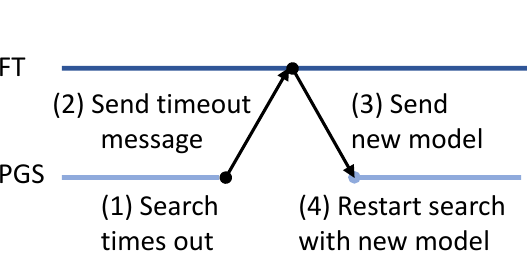}
  }
\if 0
  \begin{minipage}[t]{0.32\textwidth}
    \centering
    \includegraphics[width=\linewidth]{figs/ft_search_a.pdf}
    \label{fig:ft_search_a}
    \caption*{(a) Fine-tuning discovers a\\circuit with new lowest cost.}
  \end{minipage}
  \hfill
  \begin{minipage}[t]{0.32\textwidth}
    \centering
    \includegraphics[width=\linewidth]{figs/ft_search_b.pdf}
    \label{fig:ft_search_a}
    \caption*{(b) Policy-guided search finds a \\circuit with new lowest cost.}
  \end{minipage}
  \hfill
  \begin{minipage}[t]{0.32\textwidth}
    \centering
    \includegraphics[width=\linewidth]{figs/ft_search_c.pdf}
    \caption*{(c) No processes find a better circuit before search times out.}
  \end{minipage}
\fi 
\vspace{-2mm}
  \caption{Interactions between \sys's fine-tuning (FT) and policy-guided search (PGS) processes. (a) and (b) show how the two processes exchange information when a circuit with new lowest cost is discovered. (c) shows the case where the policy-guided search times out.}
  \label{fig:ft_search}
\vspace{-6mm}
\end{figure}

\oded{we should explain here how this achieves a balance of local and global reasoning. we should also mention here (if not somewhere else) that to optimize a circuit we do fine-tuning and just take the best circuit. We never use the RL agent not in training mode.}

\subsection{Scaling to Large Circuits\label{subsec:scale}}
This section analyzes the scalability of \sys during the training data collection and agent update phases, and describe how we improve its scalability for larger circuits.

\paragraph{Scalability of training data collection.}
To optimize an input circuit $C$, \sys's gate value predictor computes a value for every gate in $C$ at each step in a trajectory.
This process consists of a GNN inference to generate the representation of all gates in $C$, and an MLP inference that calculates the values of $C$'s gates.
For each gate, both the GNN and MLP inferences take constant time, therefore the overall time complexity of the gate selector is $O(|C|)$.
After a gate $g$ is selected, selecting a transformation and identifying the $l$-hop influenced gates takes constant time.
Overall, the time complexity of \sys's data collection is $O(|C|)$.
% Given the finite horizon $T$, the time complexity of the data collection phase is $O(T|C|)$.

% Intuitively, in a single step, \sys computes representation and value for all gates in the circuit, but all of them but one representation are discarded after the gate selection.
% The larger the circuit is, the more computation discarded, and hence the inefficiency.
% To tackle this, we can partition the circuit into sub-circuits whose number of gates is less than a constant number, and different sub-circuits are optimized separately.
% Thus, the time cost of taking a single step in training data collection become a constant.
% We evaluate the performance improvement of circuit partitioning in~\Cref{subsec:scale_evaluation}.
% However, partitioning circuits may lost optimization opportunities across partitions.
% This issue can be mitigated by combining the partitions and re-partition the circuit in a different way every iteration of training, which we would like to explore in future works.

However, as the size of the circuit grows, a significant amount of computation is discarded since all but one gate's representation are not used after gate selection. 
To address this inefficiency, \sys partitions an input circuit into sub-circuits by limiting the number of gates in each sub-circuit, and optimize these sub0circuits individually.
This approach reduces the cost of training data collection to a constant. We evaluate the performance improvement of circuit partitioning in \Cref{subsec:scale_evaluation}. 
However, partitioning circuits may lead to missed optimization opportunities that span across partitions, which we plan to explore in future work.

\paragraph{Scalability of agent update.}
% The network update phase involves re-computation of the probability of choosing the same transformation in the data collection phase.
% Specifically, to calculate the importance sampling ratio $r_t(\theta_x)$ in the loss function, we need the probability $\pi_x(x_t | C_t, g_t; \theta_x)$ under the updated $\theta_x$.
% Note that because $\pi_x$ is based on the gate representation which is also changed due to the update of the parameters in the GNN, \sys has to re-calculate the representation for $g_t$.
% Intuitively, this need \sys to run the GNN inference on circuit $C_t$, which takes time $O(|C_t|)$.
% However, the network update phase involves gradient computation and backward-propagation, which takes $O(|C_t|)$ peak GPU memory usage.
% This means that on limited hardware, we have to hand-tune the training batch-size to prevent the out of memory issue.

The agent update phase involves recomputing the probability of choosing the same transformation in the data collection phase. 
Specifically, to calculate the importance sampling ratio $\rho_t(\theta_x)$ in the loss function (\Cref{loss_ts}) \ZJ{add a reference to the loss function} \ZL{fixed}, we need the probability $\pi_x(x_t | C_t, g_t; \theta_x)$ under the updated $\theta_x$. Since $\pi_x$ is based on the gate representation, which is also changed due to the update to the parameters of the GNN. As a result, \sys needs to re-calculate the representation for $g_t$, which takes $O(|C_t|)$. Also, the network update phase involves gradient computation and backward-propagation, whose peak GPU memory usage is $O(|C_t|)$. 
On GPUs with limited device memory, we hand-tune the training batch-size to prevent out-of-memory issues.

Again, only a small number of gates are involved in the gradient generation, while the representation of other gates are compute-to-discard.
We deal with this issue based on the fact that since the gate representation only contains information of its $k$-hop neighborhood, we can obtain exactly the same representation by running GNN only on its $k$-hop neighborhood.
The number of gates in the $k$-hop neighborhood of a gate has a constant upper bound due to the sparse structure of quantum circuits.
With this optimization, the complexity of computing gradient on a single data point becomes a constant.

Note that due to the design of \sys's neural architecture, the time and space complexity of \sys only depend the number of gates in a circuit, rather than the number of qubits.

\section{Evaluation\label{sec:eval}}
%\begin{itemize}
%    \item Comparing \sys, Quartz, and others both with and without rotation merging.
%    \item Evaluating different cost metrics, such as total number of CNOT and/or T gates and circuit depth
%    \item RL v.s. search + RL v.s. RL + search
%\end{itemize}

\commentout{
\ZJ{
We describe our implementation of \sys and evaluate its performance in the following aspects:
\begin{itemize}
    \item Can \sys outperform existing rule-based and search-based quantum circuit optimizers?
    \item Can \sys's RL agent learn to adapt to different cost metrics?
    \item Can \sys automatically discover novel transformation paths from its own exploration?
\end{itemize}
}
}

We include the implementation details and evaluate \sys in this section.
\Cref{subsec:exp_setup} introduces the experimental setups and~\Cref{subsec:imple} describes the implementation of \sys and hyperparameters.
\Cref{subsec:nam_exp} and~\Cref{subsec:ibm_exp} shows the comparison of \sys with existing optimizers on the Nam gate set and the IBM gate set, respectively.
\Cref{subsec:scale_evaluation} focuses on evaluating \sys's scalability, and \Cref{subsec:ablation} performs a comprehensive ablation study.

\subsection{Experimental Setup\label{subsec:exp_setup}}

\paragraph{Benchmarks}
To evaluate the effectiveness of our framework, we employ a widely-adopted quantum circuit benchmark suite, which has been previously used by several works~\cite{nam2018automated, amy2014, extended, VOQC, pyzx} for logical circuit optimization. 
The benchmark suite primarily comprises arithmetic and reversible circuits, and we evaluate them in the Nam gate set (${CX, Rz, H, X}$), following existing studies. 
Additionally, we transpile these benchmark circuits to the IBM gate set ($\{CX, Rz, X, SX\}$) to demonstrate \sys's compatibility with different gate sets.
Since this benchmark suite contains a limited number of circuit types, we extend our evaluation to include circuits from the MQTBench library~\cite{quetschlich2022mqtbench}, which encompasses circuits from various categories, such as QAOA~\cite{qaoa} and VQE~\cite{vqe}.
%As these circuits are implemented in the IBM gate set, we evaluate them on the IBM gate set.

\paragraph{Metrics}
We use four metrics in our evaluation: total gate count, CNOT count, circuit depth, and fidelity.
Operations on NISQ devices are affected by noise, and the error rates of different gates vary.
Specifically, the error rate of two-qubit gates (e.g. CNOT) are typically an order of magnitude larger than single-qubit gates (e.g., $3\times 10^{-2}$ and $4.43\times 10^{-3}$, respectively, on IBM Q20~\cite{sabre}).
In light of this, we evaluate \sys with not only total gate count but also CNOT count.
Moreover, since qubits in NISQ devices have limited coherence time, circuit depth should be within a certain range for successful execution.
The optimizing results ultimately translate to the fidelity of executions.
However, the end-to-end fidelity involves many factors, such as mapping method, routing method, device coupling map and device error rate, which are out of the scope of logical circuit optimization.
To rule out the effect of these factors, we follow prior work~\cite{xu2022synthesizing} and report the absolute circuit fidelity, which is defined as the product of the success rate of all gates in the circuit.
For a circuit $C$, its fidelity can be expressed as $\prod_{g\in C}(1 - e(g))$ where $e(g)$ denotes the empirical device error rate for gate $g$.
During evaluation, we use the calibration data of IBM Washington device~\cite{ibm_washington} where CNOT error rate is $1.214\times 10^{-2}$, Rz error rate is {0}, and the error rate of other single-qubit gates (i.e. X gate and SX gate) are $2.77\times 10^{-4}$.

% Finally, although \sys optimizes circuits at the logical layer and does not handle physical mapping or routing, we report the fidelity of the \sys-optimized circuits mapped to noisy devices with third-party compilation tools to demonstrate how \sys's optimization benefits end-to-end compilation results.
% We compute the fidelity as the product of success rate of all gates in the compiled circuits, including all the SWAP gates added. 
% To rule out the effect of the mapping algorithm, we exclude the readout error rate of the qubits in our fidelity expression.
\paragraph{Server specification}
Our experiments are conducted on the Perlmutter supercomputer~\cite{perlmutter}. Pre-training is performed on a node equipped with an AMD EPYC 7763 64-core 128-thread processor, 256GB DRAM, and four NVIDIA A100-SXM4-40GB GPUs. 
For circuit optimization, we use a node with the same hardware specification, but with only one NVIDIA A100-SXM4-40GB GPU.

% \ZL{May modify}We compare \sys on the above metrics with existing rule-based\cite{VOQC, nam2018automated, amy2014, pyzx, qiskit, tket} and search-based~\cite{pldi2022-quartz} quantum circuit optimizers.
% Quartz~\cite{pldi2022-quartz} is used to generate the transformation rules.
% There are xxxx rules in the Nam gate set and xxxx rules in the IBM gate set.
% The benchmark contains arithmetic circuits, multiple controlled gates (e.g., CCX and CCZ), quantum Fourier transformations, and Galois field multipliers.
% We use \sys to optimize these circuits on the Nam gate set, which includes $\{H, X, R_z, CNOT\}$, and use the 8,668 transformations discovered by Quartz~\cite{pldi2022-quartz} for this gate set.

% \ZL{This paragraph should be moved to specific end-to-end comparison part} We pre-train \sys's neural architecture on six small circuits (i.e., {\tt barenco\_tof\_3}, {\tt gf2\^{}4\_mult}, {\tt mod5\_4}, {\tt mod\_mult\_55},  {\tt tof\_5}, and {\tt vbe\_adder\_3}), and fine-tune the pre-trained model to optimize an arbitrary new circuit. \sys's gate representation generator includes six identical GNN layers, each containing two linear layers with activation following the design of GraphSAGE~\cite{graphsage} with hidden dimension size 128. 
%\ZJ{@Zikun, @Jinjun: explain \sys' neural architecture here.} 
% We use 24 hours as the fine-tuning timeout.

\commentout{
\JP{
Since GNN in our model captures local features of a circuit and different circuits may share similar local structures, it is possible for the model to generalize to other circuits after training on some specific circuits. 
So we first pre-train the model on 6 small circuits ({\tt barenco\_tof\_3}, {\tt gf2\^{}4\_mult}, {\tt mod5\_4}, {\tt mod\_mult\_55}, {\tt tof\_5}, {\tt vbe\_adder\_3}) simultaneously (we input multiple circuits at once and the agent will collect data containing different circuits for training the network) to obtain a pre-trained model that has already learned some optimizations, and then fine-tune the pre-trained model on each single circuit to get more circuit-specific optimizations. This pre-training and fine-tuning mechanism saves us lots of time to get the best results, because we do not have to start from scratch for every circuit.}
}

\subsection{Implementation Details\label{subsec:imple}}

We build the \sys training pipeline on top of PyTorch~\cite{pytorch} and DGL~\cite{dgl}. We utilize the APIs provided by Quartz~\cite{pldi2022-quartz} to generate graph representations of circuits and transformation rules and perform graph pattern matching. These APIs are encapsulated with Cython~\cite{cython} to facilitate their invocations by Python processes.

\paragraph{Data-parallel pre-training.} 
% During the pre-training stage, \sys uses a distributed data parallel pipeline to train its neural architecture across multiple GPUs. 
% Specifically, each process occupies a single GPU and contains an independent agent that take actions to collects trajectories, and compute gradients with the collected data.
% \sys synchronizes the gradients across all GPUs at the end of each training iteration.
% Each process maintains an individual initial circuit buffer.
% Instead of synchronizing the whole buffer, the processes only exchange information about the best circuits,
% which greatly reduces the synchronization overhead across GPUs.
% In the experiments, we use 4 processes to for pre-training.

During pre-training, \sys employs a distributed data-parallel pipeline to train its neural architecture across multiple GPUs.
Specifically, each process occupies a single GPU and contains an independent agent that takes actions to collect trajectories and compute gradients using the collected data. 
At the end of each training iteration, \sys synchronizes the gradients across all GPUs to update the model parameters. 
Each process maintains an individual initial circuit buffer and exchanges information about the lowest-cost circuits with each other. 
In the experiments, for each gate set, we pre-train \sys's neural architecture on six circuits (i.e. {\tt barenco\_tof\_3}, {\tt gf2\^{}4\_mult}, {\tt mod5\_4}, {\tt mod\_mult\_55},  {\tt tof\_5}, and {\tt vbe\_adder\_3}) for 8 hours with 4 GPUs.
We choose these circuits for pre-training because they are relatively small and have diverse circuit topologies.
Since the pre-trained models are reused in circuit optimization, we do not count the pre-training time when reporting \sys's search time.

\begin{wraptable}{l}{0.46\textwidth}
\vspace{-5mm}
% \begin{table*}[th]
% \small
\caption{Hyperparameters used in \sys's training.}
\setlength{\tabcolsep}{1.1mm}{\fontsize{9}{10}\selectfont
\label{tab:hyp_2}
\vspace{-2mm}
% \begin{center}
\begin{tabular}{ l|l|l } 
\hline
 Hyperparameter & Pre-training & Fine-tuning \\ 
 \hline
 Horizon (T) & 600 & 600 \\ 
 Actor learning rate & 3e-4 & 3e-4  \\ 
 Critic learning rate & 5e-4 & 5e-4 \\
 GNN learning rate & 3e-4 & 3e-4 \\
 Num. actors & 128 & 64 \\
 Num. epochs & 20 & 5 \\
 Minibatch size & 4800 & 4800 \\
 Discount ($\gamma$) & 0.95 & 0.95 \\
 Clipping parameter $\epsilon$ & 0.2 & 0.2 \\
 Entropy coeff. & 0.02 & 0.02 \\
\hline
 
\end{tabular}
% \end{center}
% \end{table*}
\vspace{-6mm}
}
\end{wraptable}

\paragraph{Fine-tuning and policy-guided search.}
In the experiment, the fine-tuning process starts 5 minutes earlier than the search process, allowing \sys's neural architectures to learn to discover circuit-specific optimizations before starting the search.
The search timeout is set to 20 minutes.
We run \sys for 6 hours on each input circuit.

\paragraph{Rotation merging.} 
Aside from circuit transformations, many quantum circuit optimizers, including Nam~\cite{nam2018automated}, VOQC~\cite{VOQC} and Quartz~\cite{pldi2022-quartz}, incorporate {\em rotation merging}, a technique that merges $R_z$ gates with identical phase polynomial expressions~\cite{nam2018automated}, into their optimization pipeline.
However, the $R_z$ gates being merged may be arbitrarily far apart, making it difficult to represent rotation merging as combinations of local circuit transformations.
To address this challenge, we integrate rotation merging into \sys's circuit transformation process, following prior work.
\Cref{fig:rotation_merging_transfers} demonstrates that \sys can autonomously identify optimizations similar to rotation merging through its own exploration of certain circuits.
% By contrast, the search-based method employed by uartz cannot automatically detect rotation-merging optimizations on those circuits.

\paragraph{Cost function} 
In all our experiments, we use total gate count as the cost function for \sys. 
While 2-qubit gate count (e.g., CNOT count) can also be used as the cost function, we find that both of them yield equally good results in 2-qubit gate count. 
However, in some scenarios, minimizing the count of a specific type of 1-qubit gate is crucial. 
For example, minimizing the count of $Rz$ gates is essential in fault-tolerant quantum computing settings, as $Rz$ is non-Clifford and is more expensive to implement than Clifford gates in error correction. 
Thus, we select total gate count as the cost function. 
We do not adopt fidelity as the cost function because, at the logical optimization layer, we lack hardware specification, and the circuit is not physically mapped or routed.
% However, the goal of logical optimization is to benefit the end-to-end compilation result, and we assume that logical circuits with fewer gates, especially 2-qubit gates, generally have higher fidelity after mapping and routing. 
% We evaluate the correlation between total gate count and fidelity in~\Cref{subsec:ibm_exp}.

\commentout{
\begin{table*}[th]
\small
\caption{Hyperparameters of \sys's neural architecture used in evaluation.}
\label{tab:hyp_1}
\begin{center}
\begin{tabular}{ l|l|l|l|l|l|l } 
\hline
\multirow{2}{7em}{Hyperparameter} & \multicolumn{2}{c|}{GNN} & \multicolumn{2}{c|}{Actor} & \multicolumn{2}{c}{Critic}\\
 \cline{2-7}
 & $\#$ layers &  hidden size & $\#$ layers & hidden size & $\#$ layers & hidden size\\ 
 \hline
 Value & 6 & 128 & 2 & 256 & 2 & 128 \\
 \hline
\end{tabular}
\end{center}
\end{table*}
}

\paragraph{Hyperparameters} 
In our evaluation, \sys uses a 6-layer graph neural network (GNN) as the gate representation generator, with a hidden dimension size of 128.
The gate selector and the transformation selector are both two-layer multiple layer perceptron (MLP) with a hidden dimension size of 128 and 256, respectively.
\Cref{tab:hyp_2} lists the hyperparameters used in \sys's learning process. 
We choose $\lambda=0.9$ in \sys's gate selector and use 1-hop influenced gates in training.
% \Cref{subsec:ablation} evaluates the influence of important hyperparameters, including the number of GNN layers and the hop number of influenced gates.

\subsection{Comparison on the Nam Gate Set\label{subsec:nam_exp}}

\begin{table*}[t]
\centering
\caption{Comparing \sys and existing circuit optimizers on reducing total gate count of the benchmark circuits for the Nam gate set. The best result for each circuit is in bold.}
\vspace{-2mm}
\label{tab:nam_total_gate}
\setlength{\tabcolsep}{1.1mm}{\fontsize{7}{8}\selectfont
% \resizebox{\textwidth}{!}
% {%
\begin{tabular}{l||r|r|r|r|r|r|r||r|r}
\hline
\multicolumn{10}{c}{\bf Total Gate Count (Nam gate set)}\\
\hline
{\bf Circuit} & 
{\bf Orig.} & 
{\bf Nam} & 
{\bf VOQC} & 
{\bf Qiskit} & 
{\bf Tket} & 
% \rotatebox{90}{\parbox{1.5cm}{\textbf{Quartz \\ w/ R.M.}}} & 
% \rotatebox{90}{\parbox{1.5cm}{\textbf{Quartz \\ w/o R.M.}}} & 
{\bf \begin{tabular}{@{}c@{}}Quartz \\ w/ R.M.\end{tabular}} &
{\bf \begin{tabular}{@{}c@{}}Quartz \\ w/o R.M.\end{tabular}} &
{\bf \begin{tabular}{@{}c@{}}\sys \\ w/ R.M.\end{tabular}} &
{\bf \begin{tabular}{@{}c@{}}\sys \\ w/o R.M.\end{tabular}} \\ 
\hline
\texttt{tof\_3} & 45 & 35 & 35 & 43 & 39 & 35 & 35 & \textbf{33} & \textbf{33}\\
\texttt{barenco\_tof\_3} & 58 & 40 & 40 & 54 & 49 & 40 & 46 & \textbf{35} & 36\\
\texttt{mod5\_4} & 63 & 51 & 51 & 61 & 57 & 37 & 39 & \textbf{24} & \textbf{24}\\
\texttt{tof\_4} & 75 & 55 & 55 & 71 & 64 & 55 & 55 & \textbf{51} & \textbf{51}\\
\texttt{barenco\_tof\_4} & 114 & 72 & 72 & 106 & 97 & 72 & 90 & \textbf{62} & \textbf{62}\\
\texttt{tof\_5} & 105 & 75 & 75 & 99 & 89 & 75 & 75 & \textbf{69} & \textbf{69}\\
\texttt{mod\_mult\_55} & 119 & 91 & 92 & 115 & 103 & 90 & 94 & \textbf{84} & \textbf{84}\\
\texttt{vbe\_adder\_3} & 150 & 89 & 89 & 189 & 130 & 89 & 112 & \textbf{71} & \textbf{71}\\
\texttt{barenco\_tof\_5} & 170 & 104 & 104 & 158 & 145 & 104 & 134 & \textbf{90} & 96\\
\texttt{csla\_mux\_3} & 170 & 155 & 158 & 249 & 149 & 148 & 146 & \textbf{141} & \textbf{141}\\
\texttt{rc\_adder\_6} & 200 & 140 & 152 & 226 & 172 & 152 & 170 & \textbf{134} & 148\\
\texttt{gf2\^{}4\_mult} & 225 & 187 & 186 & 212 & 205 & 176 & 215 & \textbf{160} & 163\\
\texttt{tof\_10} & 255 & 175 & 175 & 239 & 214 & 175 & 175 & \textbf{169} & \textbf{169}\\
\texttt{mod\_red\_21} & 278 & 180 & 184 & 256 & 223 & 198 & 268 & 179 & \textbf{177}\\
\texttt{hwb6} & 259 & - & 209 & 258 & 217 & 214 & 250 & \textbf{193} & 196\\
\texttt{gf2\^{}5\_mult} & 347 & 296 & 287 & 327 & 319 & 274 & 334 & \textbf{257} & 266\\
\texttt{csum\_mux\_9} & 420 & 266 & 280 & 420 & 365 & 256 & 420 & \textbf{224} & 282\\
\texttt{barenco\_tof\_10} & 450 & 264 & 264 & 418 & 385 & 264 & 450 & \textbf{236} & 274\\
\texttt{qcla\_com\_7} & 443 & 284 & 269 & 498 & 377 & 288 & 431 & \textbf{258} & 277\\
\texttt{ham15-low} & 443 & - & 348 & 421 & 392 & 360 & 435 & 323 & \textbf{322}\\
\texttt{gf2\^{}6\_mult} & 495 & 403 & 401 & 464 & 453 & 391 & 485 & \textbf{352} & 386\\
\texttt{qcla\_adder\_10} & 521 & 399 & 416 & 630 & 444 & 404 & 518 & \textbf{380} & 391\\
\texttt{gf2\^{}7\_mult} & 669 & 555 & 543 & 627 & 614 & 531 & 657 & \textbf{484} & 530\\
\texttt{gf2\^{}8\_mult} & 883 & 712 & 706 & 819 & 806 & 703 & 883 & \textbf{648} & 755\\
\texttt{qcla\_mod\_7} & 884 & \textbf{624} & 678 & 933 & 759 & 652 & 884 & 629 & 664\\
\texttt{adder\_8} & 900 & 606 & 596 & 1001 & 802 & 624 & 874 & \textbf{560} & 669\\
\hline
{\bf \begin{tabular}{@{}c@{}}Geo. Mean\\Reduction\end{tabular}} & - & 28.0\% & 27.0\% & -0.6\% & 13.1\% & 28.3\% & 12.9\% & \textbf{35.8\%} & 32.6\% \\
\hline
\end{tabular}
}
\end{table*}

\begin{table*}[t]
\centering
\caption{Comparing \sys and existing circuit optimizers on reducing CNOT count of the benchmark circuits for the Nam gate set. The best result for each circuit is in bold.}
\label{tab:nam_cnot}
\vspace{-2mm}
\setlength{\tabcolsep}{1.1mm}{\fontsize{7}{8}\selectfont
% \resizebox{\textwidth}{!}
% {%
\begin{tabular}{l||r|r|r|r|r|r|r||r|r}
\hline
\multicolumn{10}{c}{\bf CNOT Count (Nam gate set)}\\
\hline
{\bf Circuit} & 
{\bf Orig.} & 
{\bf Nam} & 
{\bf VOQC} & 
{\bf Qiskit} & 
{\bf Tket} & 
{\bf \begin{tabular}{@{}c@{}}Quartz \\ w/ R.M.\end{tabular}} &
{\bf \begin{tabular}{@{}c@{}}Quartz \\ w/o R.M.\end{tabular}} &
{\bf \begin{tabular}{@{}c@{}}\sys \\ w/ R.M.\end{tabular}} &
{\bf \begin{tabular}{@{}c@{}}\sys \\ w/o R.M.\end{tabular}} \\ 
\hline
\texttt{tof\_3} & 18 & 14 & 14 & 18 & 18 & 14 & 14 & \textbf{12} & \textbf{12}\\
\texttt{barenco\_tof\_3} & 24 & 18 & 18 & 24 & 24 & 18 & 20 & \textbf{13} & 14\\
\texttt{mod5\_4} & 28 & 28 & 28 & 28 & 28 & 20 & 22 & \textbf{13} & \textbf{13}\\
\texttt{tof\_4} & 30 & 22 & 22 & 30 & 30 & 22 & 22 & \textbf{18} & \textbf{18}\\
\texttt{barenco\_tof\_4} & 48 & 48 & 34 & 48 & 48 & 34 & 40 & \textbf{24} & \textbf{24}\\
\texttt{tof\_5} & 42 & 30 & 30 & 42 & 42 & 30 & 30 & \textbf{24} & \textbf{24}\\
\texttt{mod\_mult\_55} & 48 & 40 & 42 & 48 & 48 & 39 & 41 & 37 & \textbf{36}\\
\texttt{vbe\_adder\_3} & 70 & 50 & 50 & 62 & 62 & 44 & 50 & \textbf{32} & \textbf{32}\\
\texttt{barenco\_tof\_5} & 72 & 50 & 50 & 72 & 72 & 50 & 60 & \textbf{36} & 42\\
\texttt{csla\_mux\_3} & 80 & 70 & 74 & 71 & 71 & 70 & 68 & \textbf{63} & \textbf{63}\\
\texttt{rc\_adder\_6} & 93 & 71 & 71 & 81 & 81 & 71 & 77 & \textbf{61} & 67\\
\texttt{gf2\^{}4\_mult} & 99 & 99 & 99 & 99 & 99 & 95 & 99 & \textbf{79} & 81\\
\texttt{tof\_10} & 102 & 70 & 70 & 102 & 102 & 70 & 70 & \textbf{64} & \textbf{64}\\
\texttt{mod\_red\_21} & 105 & 81 & 81 & 105 & 104 & 81 & 105 & \textbf{74} & \textbf{74}\\
\texttt{hwb6} & 116 & -  & 104 & 115 & 111 & 99 & 115 & \textbf{86} & 88\\
\texttt{gf2\^{}5\_mult} & 154 & 154 & 154 & 154 & 154 & 149 & 154 & \textbf{132} & 133\\
\texttt{csum\_mux\_9} & 168 & 140 & 140 & 168 & 168 & 140 & 168 & \textbf{112} & 140\\
\texttt{barenco\_tof\_10} & 192 & 130 & 130 & 192 & 192 & 130 & 192 & \textbf{102} & 122\\
\texttt{qcla\_com\_7} & 186 & 132 & 132 & 174 & 174 & 127 & 178 & \textbf{111} & 120\\
\texttt{ham15-low} & 236 & - & 210 & 236 & 225 & 209 & 236 & 185 & \textbf{180}\\
\texttt{gf2\^{}6\_mult} & 221 & 221 & 221 & 221 & 221 & 221 & 221 & \textbf{182} & 197\\
\texttt{qcla\_adder\_10} & 233 & 183 & 199 & 213 & 205 & 187 & 230 & \textbf{165} & 169\\
\texttt{gf2\^{}7\_mult} & 300 & 300 & 300 & 300 & 300 & 300 & 300 & \textbf{253} & 271\\
\texttt{gf2\^{}8\_mult} & 405 & 405 & 405 & 405 & 402 & 405 & 405 & \textbf{364} & 393\\
\texttt{qcla\_mod\_7} & 382 & 292 & 328 & 366 & 366 & 307 & 382 & \textbf{285} & 296\\
\texttt{adder\_8} & 409 & 291 & 301 & 385 & 383 & 305 & 395 & \textbf{265} & 313\\
\hline
{\bf \begin{tabular}{@{}c@{}}Geo. Mean\\Reduction\end{tabular}} & - & 18.4\% & 17.8\% & 2.5\% & 3.0\% & 20.6\% & 10.9\% & \textbf{33.9\%} & 30.7\% \\
\hline
\end{tabular}
}
\end{table*}

% \ZL{Explain why we are not comparing with pyzx}
We compare \sys with existing rule-based~\footnote{PyZX~\cite{pyzx} is another rule-based optimizer. However, it only minimizes the T gate count and does not explicitly optimize our chosen metrics, namely total gate count, CNOT count, circuit fidelity, and depth. We observe that PyZX achieves worse performance on these metrics and therefore exclude it in this comparison.} (i.e. Nam~\cite{nam2018automated}, VOQC~\cite{VOQC}, Qiskit~\cite{qiskit} and \tket~\cite{tket}) and search-based~\footnote{QUESO~\cite{xu2022synthesizing} is another search-based optimizer. It is not publicly available yet, which prevents us from conducting a comparison. We plan to include it in the final paper if it is open source later.} (i.e. Quartz~\cite{pldi2022-quartz}) optimizers on total gate count, CNOT count, and circuit depth on the Nam gate set.
We did evaluate on circuit fidelity since the Nam gate set is not hardware native.
Both \sys and Quartz~\cite{pldi2022-quartz} use a set of 6206 transformation rules generated by Quartz.
This section (and \Cref{subsec:ibm_exp}) focus on circuits with less than 1,500 gates \ZJ{check if this is true}, and \Cref{subsec:scale_evaluation} evaluates the performance and scalability of \sys on larger circuits.
\citet{nam2018automated} and VOQC~\cite{VOQC} have both incorporated rotation merging into their framework while Quartz adopts it as a preprocessing procedure.
Similarly, \sys adopts rotation merging as a preprocessing step.
However, to demonstrate the effectiveness of \sys without rotation merging, we report results of both \sys and Quartz without rotation merging.

\begin{figure}
\label{fig:rotation_merging}
    \centering
    \subfloat[Rotation merging] {
    \includegraphics[scale=0.26]{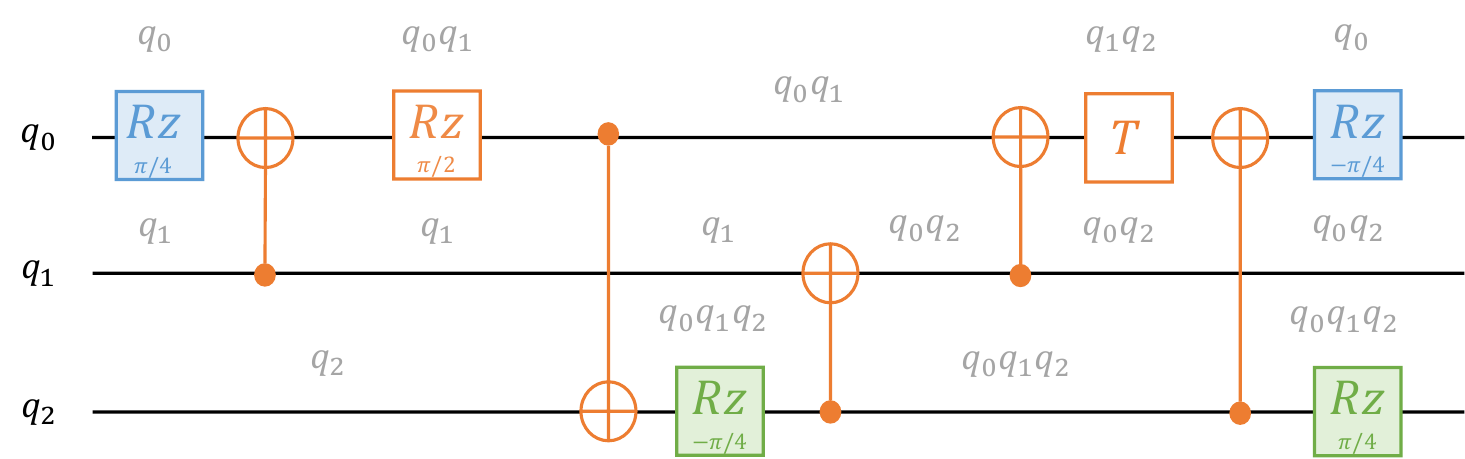}
    \label{fig:rotation_merging_phases}
    }
    \\
    \subfloat[Rotation merging as transformations]{
    \includegraphics[scale=0.24]{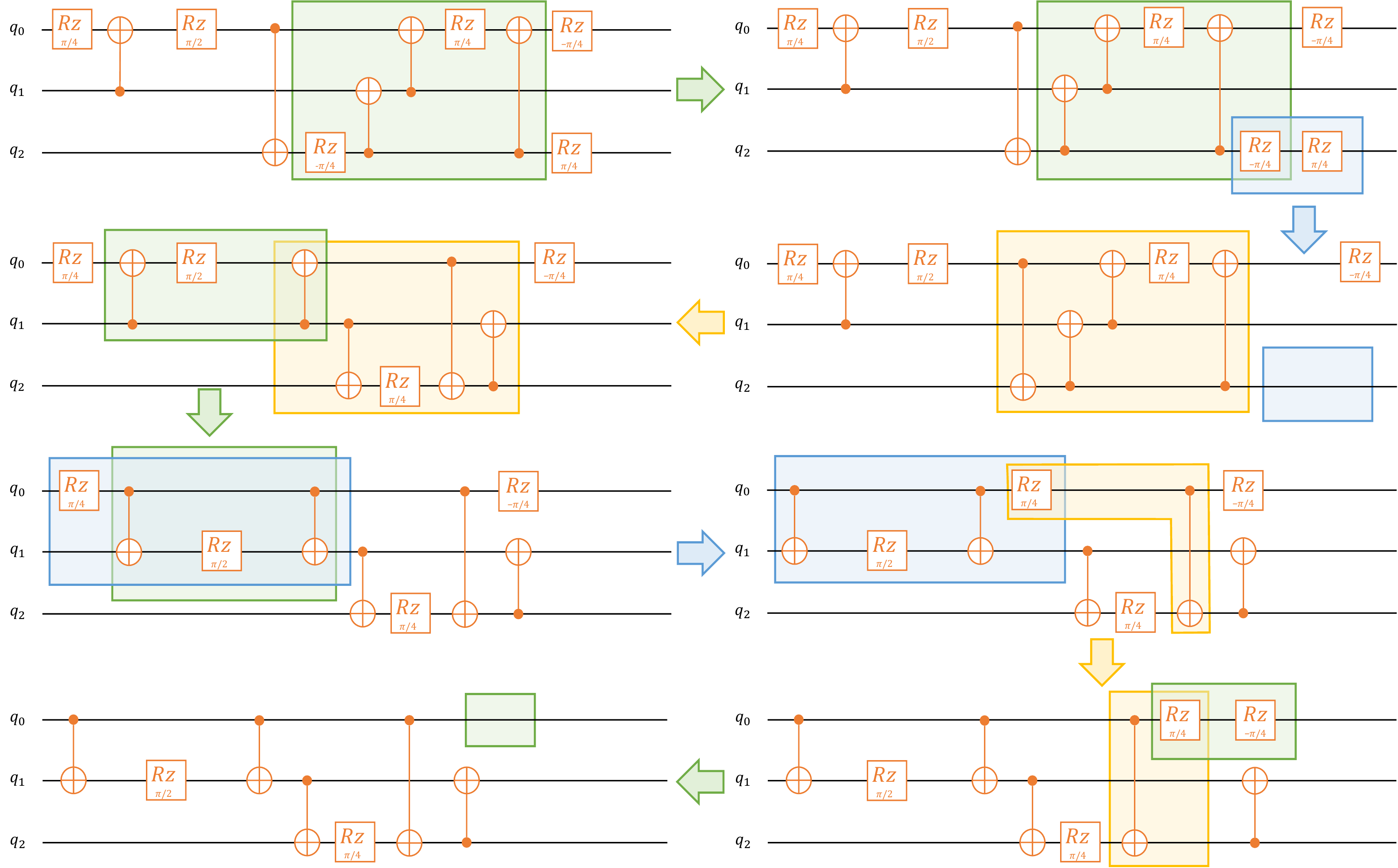}
    \label{fig:rotation_merging_transfers}
    }
\vspace{-2mm}
    \caption{%Rotation merging as a sequence of transformations.
    (a) shows an example circuit with $CNOT$ and $Rz$ gates. The phase of each $Rz$ gate is shown above the gate. Rotations with identical phase can be combined by rotation merging.
    (b) represents the rotation merging optimization as a sequence of transformations discovered by \sys. Note that \sys can also update the final circuit to have the same CNOT topology as the original circuit by applying a few additional transformations, which are omitted in the figure.}
\vspace{-6mm}
\end{figure}

As shown in~\Cref{tab:nam_total_gate}, \sys outperforms existing rule-based circuit optimizers on almost all benchmark circuits.
Rule-based optimizers rely on a fixed set of manually designed rules and schedule them in a predetermined, typically greedy, manner.
For instance, Nam~\cite{nam2018automated} is an optimizer that has been specifically fine-tuned for the Nam gate set and is among the best-performing rule-based optimizers on that gate set.
Nam applies gate-set-specific heuristics and rules, such as rotation merging and floating $R_z$ gates, and also uses 1- and 2-qubit gate cancellation as subroutines to simplify the circuits.
To apply these subroutines, \citet{nam2018automated} uses two fixed schedules.
We report the results of the {\em heavy} schedule, which applies more aggressive optimization and achieves better results.
Other rule-based optimizers (i.e., VOQC~\cite{VOQC}, Qiskit~\cite{qiskit}, \tket~\cite{tket}) operate similarly, where VOQC~\cite{VOQC} also implements specific optimizations for the Nam gate set, which are combined into a pass \texttt{optimize\_nam}.
In contrast, \sys is equipped with transformation rules automatically generated by Quartz~\cite{pldi2022-quartz}.
%Moreover, \sys uses a learning-based approach to applying these transformations, thus enhancing the flexibility of rule application.

\sys surpasses Quartz~\cite{pldi2022-quartz} for {\em all} benchmark circuits. 
Both tools employ the same set of transformations, and therefore the difference in their performance arises from their respective search algorithms.
Quartz uses cost-based backtracking search to schedule transformations, while \sys leverages reinforcement learning to detect optimization opportunities and determine the best combination of transformations to achieve them.
\ZJ{I don't like the following sentence --- it seems subjective and incorrect since Quartz is actually guided by cost during the search.
In essence, the cost-based backtracking search can be compared to searching for something valuable in an unknown area without guidance, while the reinforcement learning approach resembles exploring with the aid of a compass and a map.} \ZL{I think we can remove this sentence.}

The reduction in gate count achieved by \sys translates to an average reduction of 25.7\% in circuit depth, which outperforms the best depth reduction achieved by existing optimizers (i.e., 19.6\%). We report the circuit depth results in the supplementary material due to the page limit.

\Cref{tab:nam_cnot} shows CNOT count reduction on the same gate set.
\sys reduces CNOT count by 33.9\% on average, while the reduction achieves by existing optimizers is at most 20.6\%.
Reducing CNOT count is a desirable advantage since it significantly increases the fidelity of executing a circuit on modern NISQ devices with sparse inter-qubit connections.
Even for relatively small circuits, \sys achieves better performance than Quartz~\cite{pldi2022-quartz}, as shown in \Cref{tab:nam_total_gate} and \Cref{tab:nam_cnot}.
This is because exhaustively exploring the search space is prohibitively expensive even for very small circuits, and Quartz's back-tracking search discards circuits whose cost is larger than a threshold.
This approach can miss optimizations that require applying cost-increasing transformations.
In contrast, \sys's learning-based approach provides a more efficient and scalable solution, which allows \sys to better explore the search space and discover more optimization opportunities.

% One larger circuits, the performance gap between Quartz~\cite{pldi2022-quartz} and \sys is widened.
% \sys outperforms Quartz on large circuits where discovering {\tt CNOT} optimizations using randomized search becomes inefficient.

% We follow the evaluation metrics used by prior work~\cite{nam2018automated, VOQC, pldi2022-quartz} and first compare \sys and existing circuit optimizers on the total gate count. 
% The results are shown in \Cref{tab:result}.
% To rule out the effects of rotation merging, we report the Quartz and \sys results both with and without rotation merging.
%%%

% indicating that \sys can automatically perform most transformations used in these optimizers and can learn to discover new optimization opportunities from its own exploration on these circuits.

%%%
As shown in \Cref{tab:nam_total_gate} and \Cref{tab:nam_cnot}, the performance gap between \sys with and without rotation merging is relatively small (i.e., 3.2\% in both cases); as a comparison, disabling rotation merging decreases Quartz's performance by 15.4\% for total gate count and 9.7\% for CNOT count.
Rotation merging is an efficient global optimization technique that merges $R_z$ gates with identical phase polynomial expressions.
Though it can be achieved by multiple applications of local $R_z$ transformation rules, it is challenging to rebuild rotation merging with those rules since applying them does not provide immediate rewards until we finally fuse the rotation gates.
%in Quartz since it provides no guidance on when and where to apply these rules, so it must try all possible schedules to find the best one.
%In contrast, \sys knows the best schedule.
Though it takes some time to learn, once the best schedule is learned, \sys can apply it every where, which is much more efficient.
\sys can learn to perform optimizations similar to rotation merging through its own exploration.
\Cref{fig:rotation_merging} shows such an example, where \sys merges two pairs of $R_z$ gates with identical phase polynomial using a sequence of local transformations. 
%It is also possible for Quartz to find the schedule in this example, however, it could take a large number of trials due to the inefficiency of its cost-based backtracking search.
% To evaluate this, we run \sys on the original circuits without applying rotation merging and examine whether the circuits returned by \sys can be further optimized using rotation merging. 
% Parenthesis in \Cref{tab:result} shows how each circuit can be optimized with rotation merging.
%%%
% Among the 15 circuits with less than 400 gates, 13 of them do not include rotation merging opportunities. Larger circuits still contain rotation merging opportunities, which is because the fine-tuning procedure on these circuits does not converge within the 24h timeout.
%%%

\commentout{
\JP{
First we measure the performance of \sys in terms of total gate count. The results are obtained by our pre-training and fine-tuning way. The pre-training takes 6 circuits mentioned before in which lasts about 2.5 hours for 115 iterations, and gets on-par results with previous works on these small circuits. Then we fine-tune each circuit in the benchmark suite to optimize them further. We can see after this two-step approach \sys provides much better results than other works in Table \ref{tab:nam}.
}
}

% \begin{figure}
%     \centering
%     \subfloat[Pre-training.]{
%     \includegraphics[scale=0.32]{figs/pretrain.pdf}
%     \label{fig:eval_pretrain}
%     }
%     \\
%     %\subfloat[Hierarchical Gate-Transformation Selection.]{
%     %\includegraphics[scale=0.26]{figs/ppo_design_abl.pdf}
%     %\label{fig:ppo_ablation}
%     %}
%     \subfloat[Choice of $K$ in GNN.]{
%     \includegraphics[scale=0.32]{figs/khop.pdf}
%     \label{fig:khop}
%     }
%     \caption{Ablation studies. \Cref{fig:eval_pretrain} compares \sys's performance using the pre-trained model and a randomly initialized model.
%     %\Cref{fig:ppo_ablation} evaluates \sys's hierarchical gate-transformation design by comparing \sys against a strawman PPO approach (introduced in \Cref{subsec:ablation}) on the {\tt barenco\_tof\_10} circuit. 
%     \Cref{fig:khop} compares \sys's optimization performance on the {\tt barenco\_tof\_10} circuit using different GNN architectures. Numbers in parenthesis are the total amount of collected experiences and the gate count of the best discovered circuit.}
% \end{figure}

\subsection{Comparison on the IBM Gate Set\label{subsec:ibm_exp}}

\begin{table*}[t]
\centering
\caption{Comparing \sys and existing circuit optimizers on reducing total gate count and CNOT count of the benchmark circuits for the IBM gate set. The best result for each circuit is in bold.}
\label{tab:ibm}
\vspace{-2mm}
\setlength{\tabcolsep}{1.1mm}{\fontsize{7}{8}\selectfont
% \resizebox{\textwidth}{!}
% {%
\begin{tabular}{l||r|r|r|r|r|r|r|r||r|r}
\hline
\multirow{2}{7em}{\bf Circuit} & \multicolumn{2}{c|}{\bf Orig.} & \multicolumn{2}{c|}{\bf Qiskit} & \multicolumn{2}{c|}{\bf Tket} & \multicolumn{2}{c||}{\bf Quartz} & \multicolumn{2}{c}{\bf \sys}\\
\cline{2-11}
& {\bf \# gates} & {\bf \# CXs} & {\bf \# gates} & {\bf \# CXs} &{\bf \# gates} & {\bf \# CXs} &{\bf \# gates} & {\bf \# CXs} &{\bf \# gates} & {\bf \# CXs} \\
% {\bf Circuit} & 
% {\bf Orig.} & 
% {\bf Nam} & 
% {\bf VOQC} & 
% {\bf Qiskit} & 
% {\bf Tket} & 
% {\bf \begin{tabular}{@{}c@{}}Quartz \\ w/ R.M.\end{tabular}} &
% {\bf \begin{tabular}{@{}c@{}}Quartz \\ w/o R.M.\end{tabular}} &
% {\bf \begin{tabular}{@{}c@{}}\sys \\ w/ R.M.\end{tabular}} &
% {\bf \begin{tabular}{@{}c@{}}\sys \\ w/o R.M.\end{tabular}} \\ 
\hline
\texttt{tof\_3} & 53 & 18 & 51 & 18 & 51 & 18 & 42 & \textbf{14} & \textbf{39} & \textbf{14}\\
\texttt{barenco\_tof\_3} & 65 & 24 & 61 & 24 & 61 & 24 & 52 & 20 & \textbf{42} & \textbf{16}\\
\texttt{mod5\_4} & 71 & 28 & 69 & 28 & 69 & 28 & 62 & 27 & \textbf{50} & \textbf{22}\\
\texttt{tof\_4} & 88 & 30 & 84 & 30 & 84 & 30 & 67 & \textbf{22} & \textbf{62} & \textbf{22}\\
\texttt{tof\_5} & 123 & 42 & 117 & 42 & 117 & 42 & 92 & \textbf{30} & \textbf{85} & \textbf{30}\\
\texttt{barenco\_tof\_4} & 125 & 48 & 117 & 48 & 117 & 48 & 99 & 40 & \textbf{75} & \textbf{30}\\
\texttt{mod\_mult\_55} & 140 & 48 & 135 & 48 & 145 & 48 & 114 & 41 & \textbf{100} & \textbf{39}\\
\texttt{vbe\_adder\_3} & 165 & 70 & 186 & 62 & 159 & 62 & 124 & 50 & \textbf{82} & \textbf{36}\\
\texttt{barenco\_tof\_5} & 185 & 72 & 173 & 72 & 173 & 72 & 146 & 60 & \textbf{108} & \textbf{44}\\
\texttt{csla\_mux\_3} & 200 & 80 & 235 & 71 & 188 & 71 & 181 & 72 & \textbf{164} & \textbf{65}\\
\texttt{rc\_adder\_6} & 229 & 93 & 250 & 81 & 277 & 81 & 189 & 75 & \textbf{164} & \textbf{71}\\
\texttt{gf2\^{}4\_mult} & 246 & 99 & 232 & 99 & 232 & 99 & 229 & 99 & \textbf{192} & \textbf{84}\\
\texttt{hwb6} & 287 & 116 & 281 & 115 & 282 & 111 & 272 & 114 & \textbf{208} & \textbf{85}\\
\texttt{mod\_red\_21} & 294 & 105 & 279 & 105 & 317 & 104 & 292 & 105 & \textbf{204} & \textbf{77}\\
\texttt{tof\_10} & 298 & 102 & 282 & 102 & 282 & 102 & 217 & \textbf{70} & \textbf{203} & \textbf{70}\\
\texttt{gf2\^{}5\_mult} & 374 & 154 & 353 & 154 & 353 & 154 & 372 & 154 & \textbf{295} & \textbf{135}\\
\texttt{csum\_mux\_9} & 459 & 168 & 453 & 168 & 474 & 168 & 459 & 168 & \textbf{318} & \textbf{140}\\
\texttt{barenco\_tof\_10} & 485 & 192 & 453 & 192 & 453 & 192 & 482 & 192 & \textbf{278} & \textbf{116}\\
\texttt{ham15-low} & 485 & 236 & 456 & 236 & 455 & 225 & 480 & 234 & \textbf{366} & \textbf{188}\\
\texttt{qcla\_com\_7} & 486 & 186 & 511 & 174 & 486 & 174 & 470 & 174 & \textbf{304} & \textbf{119}\\
\texttt{gf2\^{}6\_mult} & 528 & 221 & 496 & 221 & 496 & 221 & 528 & 221 & \textbf{422} & \textbf{198}\\
\texttt{qcla\_adder\_10} & 587 & 233 & 636 & 213 & 556 & 205 & 586 & 232 & \textbf{412} & \textbf{164}\\
\texttt{gf2\^{}7\_mult} & 708 & 300 & 665 & 300 & 665 & 300 & 708 & 300 & \textbf{573} & \textbf{273}\\
\texttt{gf2\^{}8\_mult} & 928 & 405 & 864 & 405 & 861 & 402 & 923 & 403 & \textbf{821} & \textbf{392}\\
\texttt{qcla\_mod\_7} & 982 & 382 & 980 & 366 & 933 & 366 & 978 & 382 & \textbf{746} & \textbf{308}\\
\texttt{adder\_8} & 1004 & 409 & 1010 & 385 & 1169 & 383 & 997 & 405 & \textbf{685} & \textbf{283}\\
\texttt{vqe\_8} & 199 & 14 & 91 & \textbf{14} & 93 & \textbf{14} & 92 & \textbf{14} & \textbf{89} & \textbf{14}\\
\texttt{qgan\_8} & 256 & 28 & 114 & 28 & 98 & 28 & 100 & 26 & \textbf{81} & \textbf{18}\\
\texttt{qaoa\_8} & 284 & 32 & 211 & \textbf{32} & 159 & \textbf{32} & \textbf{157} & \textbf{32} & 163 & \textbf{32}\\
\texttt{ae\_8} & 502 & 56 & 350 & \textbf{56} & \textbf{283} & \textbf{56} & 373 & \textbf{56} & 305 & \textbf{56}\\
\texttt{qpeexact\_8} & 555 & 64 & 373 & 64 & \textbf{286} & \textbf{55} & 371 & 64 & 326 & 63\\
\texttt{qpeinexact\_8} & 571 & 65 & 381 & 65 & \textbf{312} & \textbf{56} & 381 & 65 & 332 & 65\\
\texttt{qft\_8} & 578 & 68 & 392 & 68 & \textbf{262} & \textbf{56} & 326 & 67 & 298 & 67\\
\texttt{qftentangled\_8} & 647 & 75 & 415 & 75 & \textbf{295} & \textbf{61} & 489 & 75 & 332 & 75\\
\texttt{portfoliovqe\_8} & 708 & 84 & 288 & 84 & 232 & 84 & 359 & 84 & \textbf{202} & \textbf{54}\\
\texttt{portfolioqaoa\_8} & 1352 & 168 & 975 & 168 & \textbf{712} & 168 & 1207 & 168 & 806 & \textbf{161}\\
\hline
{\bf \begin{tabular}{@{}c@{}}Geo. Mean\\Reduction\end{tabular}}  & - & - & 14.4\% & 1.8\% & 20.1\% & 4.1\% & 19.9\% & 7,7\% & \textbf{36.6\%} & \textbf{21.3\%} \\
\hline
\end{tabular}
}
\end{table*}

% \begin{wrapfigure}{r}{0.48\textwidth}
%     \vspace{-5mm}
%     % \centering
%     \includegraphics[scale=0.48]{figs/ibm_fidelity.pdf}
%     \vspace{-6mm}
%     \caption{Circuit fidelity comparison for the IBM gate set. The numbers in parentheses in the legend indicate the average relative fidelity improvement (geometric mean) achieved by the optimizers.}
%     \label{fig:ibm_fidelity}
% \vspace{-2mm}
% \end{wrapfigure}
We also conduct a comparison between \sys and existing optimizers on the IBM gate set using four metrics: total gate count, CNOT count, circuit depth, and fidelty.
%, including rule-based (Qiskit~\cite{qiskit} and \tket~\cite{tket}) and search-based (Quartz~\cite{pldi2022-quartz}) methods, on metrics including total gate count, CNOT count, depth, and fidelity under the IBM gate set.
%We exclude VOQC~\cite{VOQC} from the comparison since it currently does not support the IBM gate set, although it does support the older IBM gate set (${CX, U1, U2, U3}$). 
To optimize a circuit, both \sys and Quartz~\cite{pldi2022-quartz} use a set of 6881 transformation rules generated by Quartz.
For parametric quantum gates, Quartz only discovers symbolic transformations applicable to these gates with arbitrary parameter values, and misses transformations that are only valid for specific parameter values.
Some of these transformations missed by Quartz are important to optimize circuits on the IBM gate set.
%%%
To address this limitation, we also include three single-qubit transformations missed by Quartz: (1) $Rz(\pi) = SX\ Rz(\pi)\ SX$, (2) $SX\ Rz(\pi/2)\ SX = Rz(\pi/2)\ SX\ Rz(\pi/2)$, and (3) $SX\ Rz(3\pi/2)\ SX = Rz(3\pi/2)\ SX\ Rz(3\pi/2)$. 
We verify their correctness by directly computing the concrete matrix representation of the two circuits in each transformation.
These three transformations are used by both \sys and Quartz to optimize circuits.
Since rotation merging is not native to the IBM gate set, it is not applied during the evaluation. 

\begin{wrapfigure}{r}{0.48\textwidth}
    \vspace{-2mm}
    % \centering
    \includegraphics[scale=0.48]{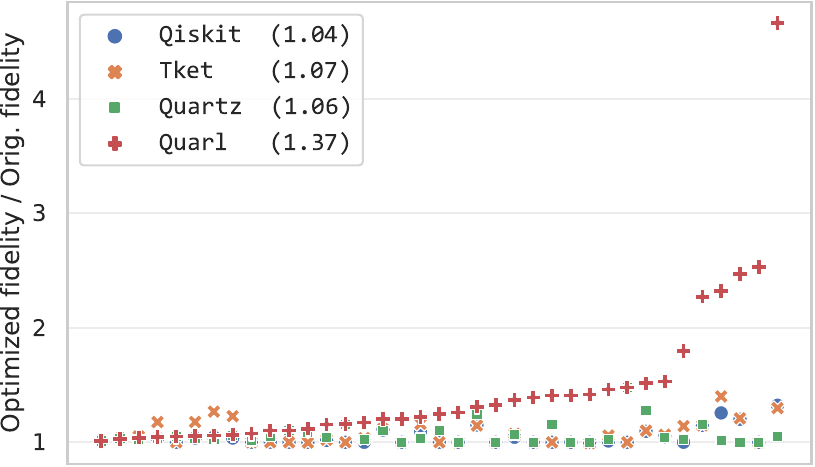}
    \vspace{-6mm}
    \caption{Circuit fidelity comparison for the IBM gate set. The numbers in parentheses in the legend indicate the average relative fidelity improvement (geometric mean) achieved by the optimizers.}
    \label{fig:ibm_fidelity}
\vspace{-2mm}
\end{wrapfigure}

\commentout{
\ZL{
Though transformation rules from Quartz are proven to be complete, the rules involving rotation gates only take symbolic parameters and constant parameters are not supported.
For example, the rule $Rz(\theta_0 + \theta_1) = Rz(\theta_0) + Rz(\theta_1)$ is generated, but rule $Z = Rz(\pi/2) + Rz(\pi/2)$ is not.
However, some of the rules involving constant parameters are critical in the IBM gate set.
Thus, we manually implemented 3 local transformation rules that use constant parameters and use them in the evaluation of both \sys and Quartz (Quartz's pattern matching engine supports such rules).
The rules we added are: 1) $Rz(\pi) = SX\ Rz(\pi)\ SX$; 2) $SX\ Rz(\pi/2)\ SX = Rz(\pi/2)\ SX\ Rz(\pi/2)$ and 3) $SX\ Rz(3\pi/2)\ SX = Rz(3\pi/2)\ SX\ Rz(3\pi/2)$.
Note that \sys's contribution is the rule scheduling algorithm, which is orthogonal to rule generation. 
We believe with higher quality rules, \sys can achieve better performance.
}
Since rotation merging is not native to the IBM gate set, it is not applied during the evaluation. 
% To evaluate fidelity, we gathered the optimized logical circuits from all optimizers and subjected them to physical mapping and routing using Qiskit~\cite{qiskit} transpiler, with the optimization level set to 0, which only performs mapping and routing passes.
}

% \begin{wrapfigure}{r}{0.48\textwidth}
%     % \vspace{-3mm}
%     % \centering
%     \includegraphics[scale=0.48]{figs/ibm_fidelity.pdf}
%     \vspace{-6mm}
%     \caption{Circuit fidelity comparison for the IBM gate set. The numbers in parentheses in the legend indicate the average relative fidelity improvement (geometric mean) achieved by the optimizers.}
%     \label{fig:ibm_fidelity}
% \vspace{-2mm}
% \end{wrapfigure}

As shown in~\Cref{tab:ibm}, \sys greatly outperform existing optimizers for the IBM gate set for both total gate count and CNOT count.
Specifically, \sys reduces total gate count by 36.6\% on average, while existing optimizers reduce total gate count by at most 20.1\%. 
The \sys-optimized circuits also reduce CNOT count by 21.3\%, while existing approaches can achieve up to 7.7\% CNOT count reduction.
The performance gap between \sys and Quartz~\cite{pldi2022-quartz} is widened on the IBM gate set.
We hypothesize this is because of Quartz's reliance on rotation merging as a preprocessing pass, which is not applicable on the IBM gate set.
%One reason is that rotation merging is not performed while Quartz~\cite{pldi2022-quartz} relies on rotation merging to achieve better performance.
\sys achieves 22.1\% circuit depth reduction on average, whereas existing optimizers reduce the circuit depth by at most 13.5\%. We report the circuit depth results in the supplementary material.

\Cref{fig:ibm_fidelity} shows the fidelity improvement achieved by different optimizers.
\sys improves circuit fidelity by up to 4.66$\times$ (on \texttt{adder\_8}) and 1.37 times on average, while Qiskit, Tket and Quartz improve circuit fidelity by 1.04$\times$, 1.07$\times$, 1.06$\times$ on average, respectively.
\sys achieves the best fidelity improvements for most circuits.
This is largely because CNOT involves a much higher error rate than other gates, and \sys performs the best on CNOT reduction.

\subsection{Ablation Studies\label{subsec:ablation}}

\begin{figure}[htbp]
  \centering
  \begin{minipage}[t]{0.32\textwidth}
    \centering
    \includegraphics[width=\linewidth]{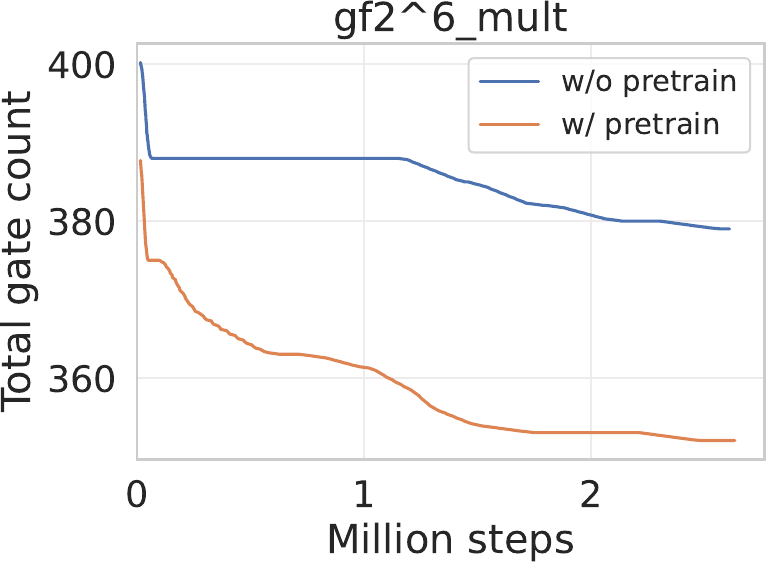}
    % \caption*{(a) \texttt{gf2\^{}6\_mult}}
  \end{minipage}
  \hfill
  \begin{minipage}[t]{0.32\textwidth}
    \centering
    \includegraphics[width=\linewidth]{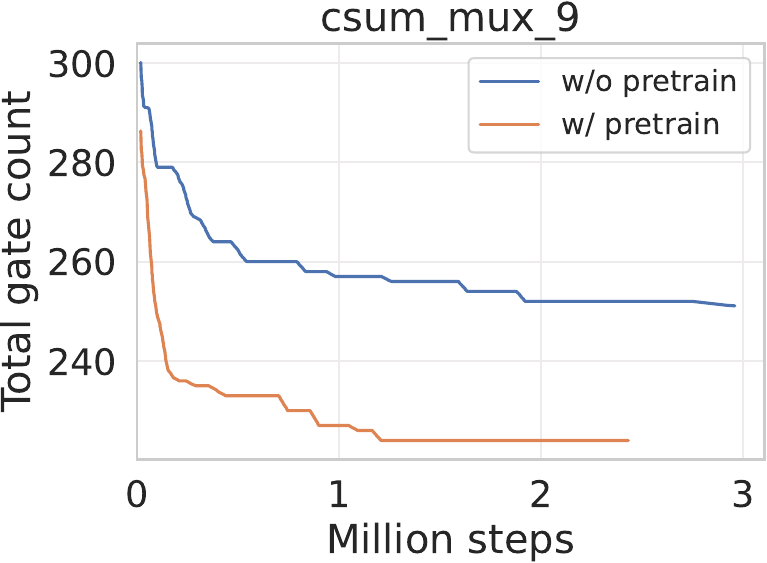}
    % \caption*{(b) \texttt{csum\_mux\_9}}
  \end{minipage}
  \hfill
  \begin{minipage}[t]{0.32\textwidth}
    \centering
    \includegraphics[width=\linewidth]{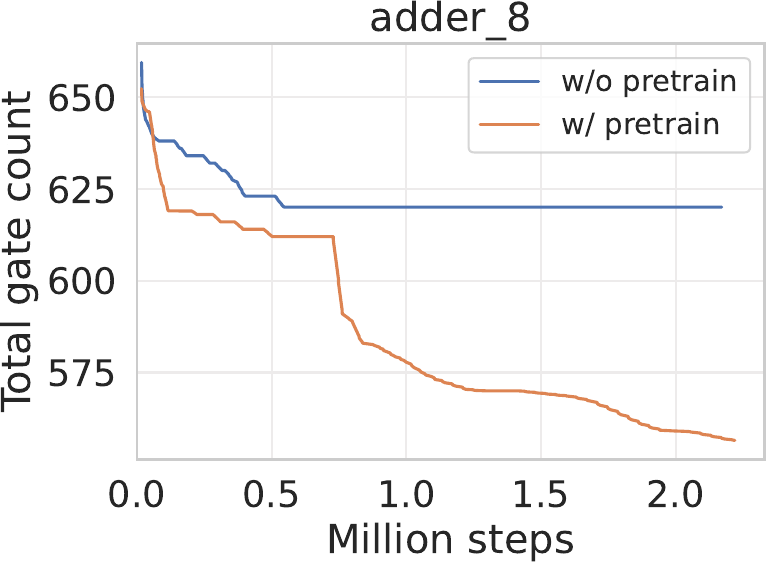}
    % \caption*{(c) \texttt{adder\_8}}
  \end{minipage}
\vspace{-2mm}
  \caption{Comparison on \sys's optimization with and without pretraining.}
  \label{fig:pretrain}
\vspace{-5mm}
\end{figure}

\begin{figure}[htbp]
  \centering
  \begin{minipage}[t]{0.32\textwidth}
    \centering
    \includegraphics[width=\linewidth]{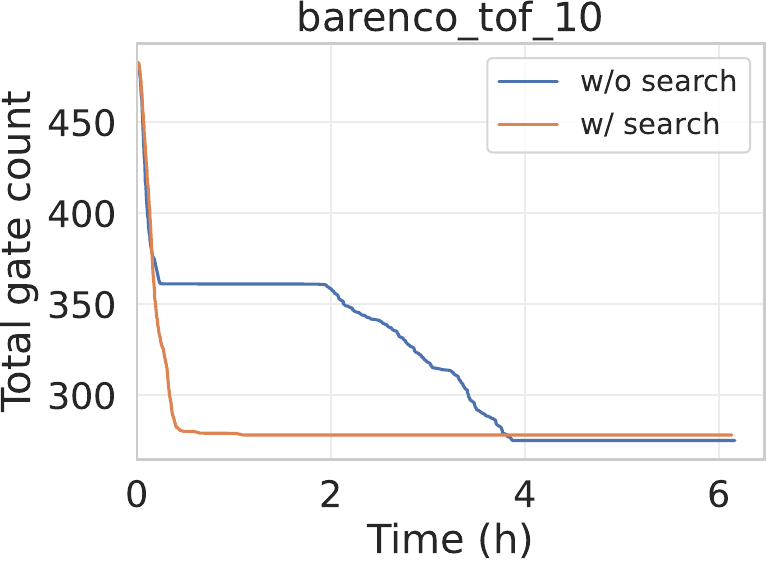}
    % \caption*{(a)}
  \end{minipage}
  \hfill
  \begin{minipage}[t]{0.32\textwidth}
    \centering
    \includegraphics[width=\linewidth]{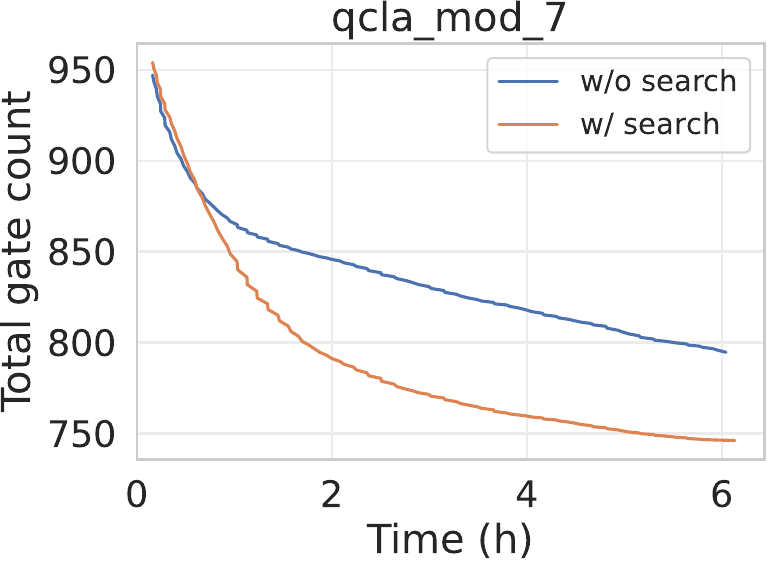}
    % \caption*{(b)}
  \end{minipage}
  \hfill
  \begin{minipage}[t]{0.32\textwidth}
    \centering
    \includegraphics[width=\linewidth]{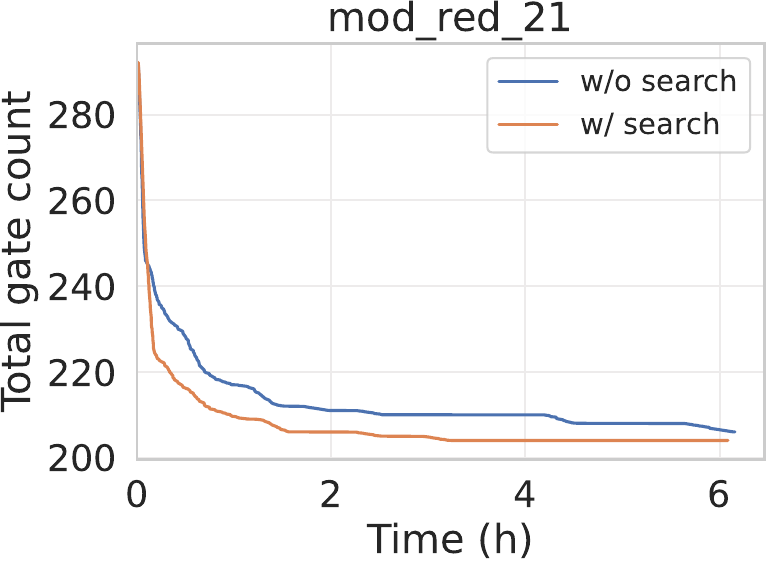}
    % \caption*{(c)}
  \end{minipage}
\vspace{-2mm}
  \caption{Comparison on \sys's optimization with and without policy-guided search.}
  \label{fig:search}
\end{figure}

\begin{figure}[htbp]
\label{fig:gnn_ab}
  \centering
  \begin{minipage}[t]{0.32\textwidth}
    \centering
    \includegraphics[width=\linewidth]{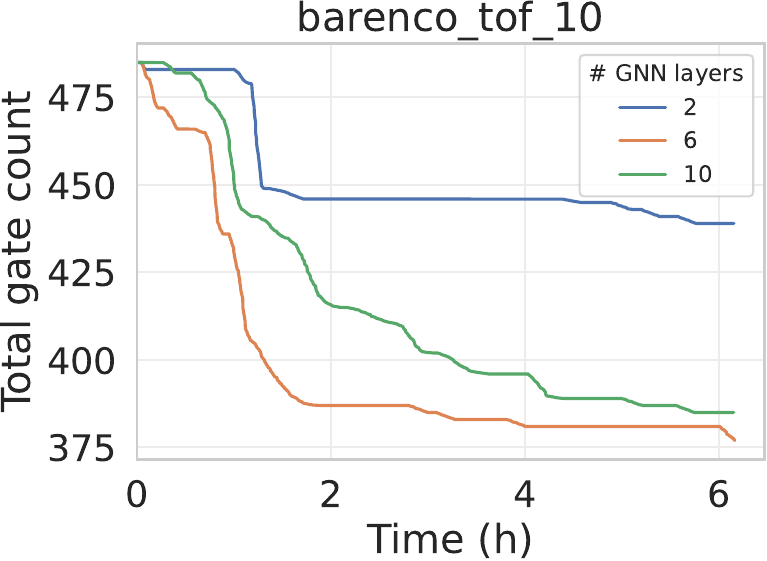}
    % \caption*{(a)}
  \end{minipage}
  \hfill
  \begin{minipage}[t]{0.32\textwidth}
    \centering
    \includegraphics[width=\linewidth]{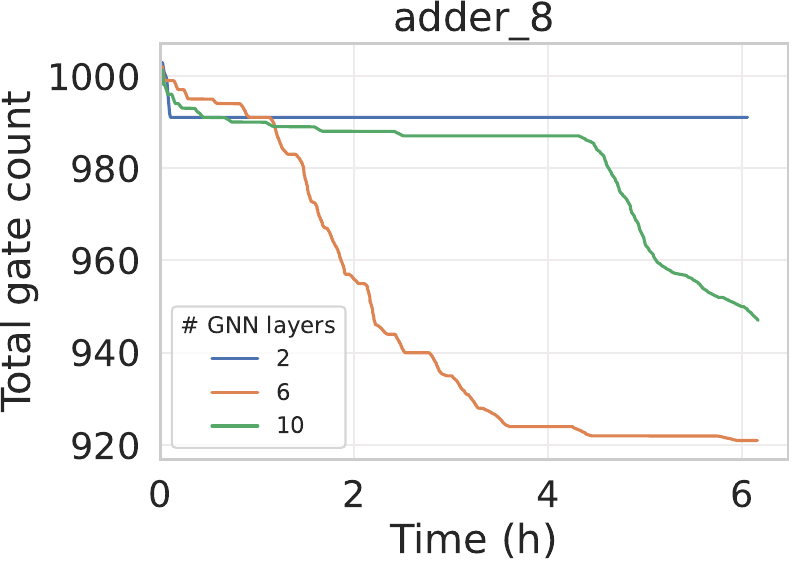}
    % \caption*{(b)}
  \end{minipage}
  \hfill
  \begin{minipage}[t]{0.32\textwidth}
    \centering
    \includegraphics[width=\linewidth]{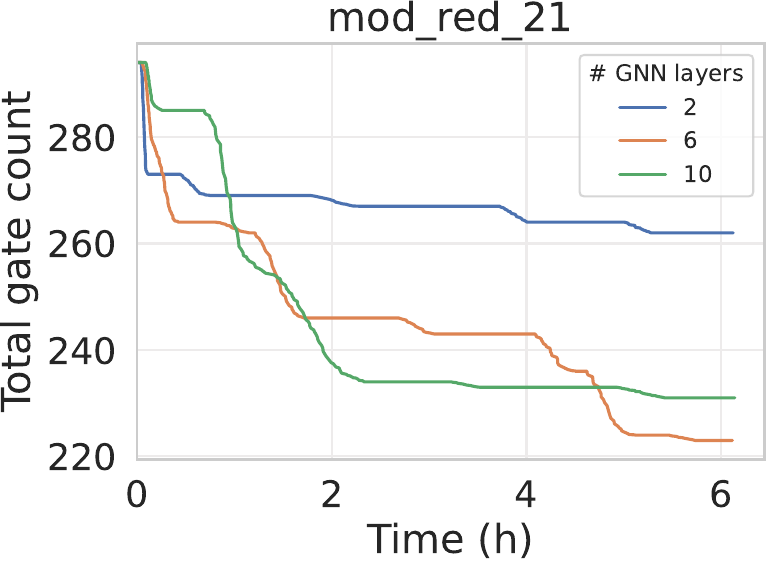}
    % \caption*{(c)}
  \end{minipage}
\vspace{-2mm}
  \caption{Comparison on \sys's optimization with different GNN layers.}
  \label{fig:gnn_layers}
\end{figure}

% \begin{figure}[htbp]
%   \centering
%   \begin{minipage}[t]{0.32\textwidth}
%     \centering
%     \includegraphics[width=\linewidth]{1.pdf}
%   \end{minipage}
%   \hfill
%   \begin{minipage}[t]{0.32\textwidth}
%     \centering
%     \includegraphics[width=\linewidth]{2.pdf}
%   \end{minipage}
%   \hfill
%   \begin{minipage}[t]{0.32\textwidth}
%     \centering
%     \includegraphics[width=\linewidth]{figs/adder_8_predecessor.pdf}
%   \end{minipage}
%   \caption{Comparison on \sys's number of hops in influenced gates.}
%   \label{fig:three-figures}
% \end{figure}

% We perform several ablation studies to evaluate how different components and choice of hyperparameters in \sys contribute to its final performance.

\paragraph{Pre-training.}
We evaluate whether \sys's pre-trained neural architecture can generalize to unseen circuits and how pre-training affects \sys's performance. 
We start a group of fine-tuning processes (with search) using the pre-trained model and another group of fine-tuning processes (with search) using randomly initialized parameters, and compare their performance.
%by starting a group of fine-tuning (with search) from the pre-trained model and another group of fine-tuning (with search) from randomly initialize model parameters.
The results are shown in~\Cref{fig:pretrain}.
%Circuits \texttt{gf2\^{}6\_mult}, \texttt{csum\_mux\_9}, and \texttt{adder\_8} are used in this experiment.
%\texttt{gf2\^{}6\_mult} has similar structure with \texttt{gf2\^{}4\_mult} in the pre-training set while \texttt{csum\_mux\_9} and \texttt{adder\_8} are unseen circuits.
%Also, these 3 circuits are larger than circuits used in pre-training (e.g. \texttt{adder\_8} is 4 times larger than \texttt{gf2\^{}4\_mult}, the largest circuit used in pre-training).
Note that these three circuits are much larger than the circuits used in pre-training.
Pre-training allows \sys to quickly identify optimizations for all three circuits and eventually achieve better results, which shows that \sys's pre-trained model generalizes well to larger circuits and boosts optimization performance on previously unseen circuits.

\commentout{
\paragraph{Hierarchical gate-transformation selection.}
\ZL{Do we still want to include this part?}
\sys uses a hierarchical approach to optimizing quantum circuits by jointly training a gate- and a transformation-selecting policy. A key  feature re of our design is using a single actor-critic architecture to train both policies (see \Cref{subsec:approach,sec:design,sec:training}).
To evaluate this feature, we tried to implement a straightforward PPO approach using two actors and two critics to train the gate- and transformation-selecting policies,
where the networks use a global state representation. However, our attempts resulted in very poor optimization results, and the agent could hardly find any useful optimizations.
%\oded{I think maybe now this part should move somewhere else, but I'm not sure where. If we don't get to it, it's also fine to leave it here for the submission.}
}

\paragraph{Policy-guided search.}
To evaluate \sys's policy-guided search, we start two groups of experiments from the same pre-trained model checkpoint: the first group performs fine-tuning with policy-guided search while the second only conducts fine-tuning.
As shown in \Cref{fig:search}, policy-guided search allows \sys to discover better solution faster.
Compared to fine-tuning, \sys's policy-guided search allows greedier exploitation (i.e., storing only best cost circuits in the initial buffer and using argmax instead of sampling during transfer selection) and bolder exploration (i.e. using soft mask to force exploration).
% The reason is that, during fine-tuning, we have to keep the training stable.
% \ZL{Modify after we have the results}

\commentout{
\sys uses a hierarchical approach to optimizing quantum circuits by jointly training a gate- and a transformation-selecting policy. A key insight behind our design is to use a single actor-critic architecture to train both policies (see \Cref{subsec:approach}).
To evaluate this hierarchical design, we compare \sys against a strawman PPO approach that uses two actors to train the gate- and transformation-selecting policies. The critic networks of these actors take as input a representation of the global circuit and estimate its value.
\Cref{fig:ppo_ablation} compares these two approaches on the {\tt barenco\_tof\_10} circuit.
Directly applying PPO to the circuit optimization problem cannot discover highly optimized solutions, since the action space is very large and highly depends on the circuit, and it's difficult to learn a unified representation for the entire circuit. 
\sys's hierarchical approach addresses these challenges by decomposing the action space into sub-spaces and using a single actor-critic architecture to train both policies.
\ZJ{I don't like the current version. Maybe we should remove this paragraph...}
}

\paragraph{The number of GNN layers.}
An important hyperparameter in \sys is the number of GNN layers in the gate representation generator (i.e., $K$ in \Cref{subsec:design:gnn}). 
Using more GNN layers allows \sys to encode a larger subcircuit for each gate to guide the gate and transformation selection at the cost of introducing deeper network and more trainable parameters, which in turn requires more time and resource to train.
\Cref{fig:gnn_layers} shows \sys's performance with different values of $K$.
We use randomly initialized model parameters to rule out the effect of pre-training, and train \sys's neural architecture with different values of $K$ for 6 hours.
% For each $K$, We report the median of three runs to reduce variability.
%%%
%\Cref{fig:khop} shows the result.
% We normalize different runs along the x-axis using the number of collected experiences (i.e., interactions with the environment).
Using a small value of $K$ (e.g., $K=2$) results in suboptimal performance, since \sys's gate and transformation selectors can only view a small neighborhood of each gate to make decision.
%%%
A larger $K$ enables representations to better guide the gate and transformation selection.
%at the cost of increased training time and resource. 
However, a very large value of $K$ (e.g., $K=10$) is also suboptimal as it introduces more parameters which consumes more data to converge and increases the time cost of \sys's representation generator to encode subcircuit into representation.

% limited training time. $K=4,6,8$ achieves on-par performance. 
\oded{it seems $K=10$ is worse even if we don't consider training time, just number of experiences. I think we should say something about it here. Is it just bad luck? Does $K=10$ require more experiences to train because it has more parameters? Does it overfit and not generalize well?}
\oded{Also, I think we should add to the legend of \Cref{fig:khop} the total number of collected experiences. Instead of just listing the final number of gates, we can list both $(x,y)$ value of the final point, something like $(22M, 403)$ for $K = 2$ and so on.}
% \ZL{Modify after we have the results}

%We compare our approach to a strawman solution where PPO is directly applied.
%GNNs are also used in the strawman approach, and we managed to combine the gate representations into a global representation for the whole circuit.
%The gate selector selects node based on the normalized dot-product similarity score of the gate representations and the global representation transformed by a single-layer perceptron, while the transformation selector makes decision solely base on the global representation.
%To make it a fair comparison, we also mask out invalid transformations for the transformation selector.
%Both the gate and transformation selectors are trained as the actor.
%The critic is a network that takes in the global circuit representation and estimates its value.
%As shown in ...

% \paragraph{Influence of the number of hops included in influence gates}

\commentout{
\begin{figure}
    \centering
    \includegraphics[scale=0.36]{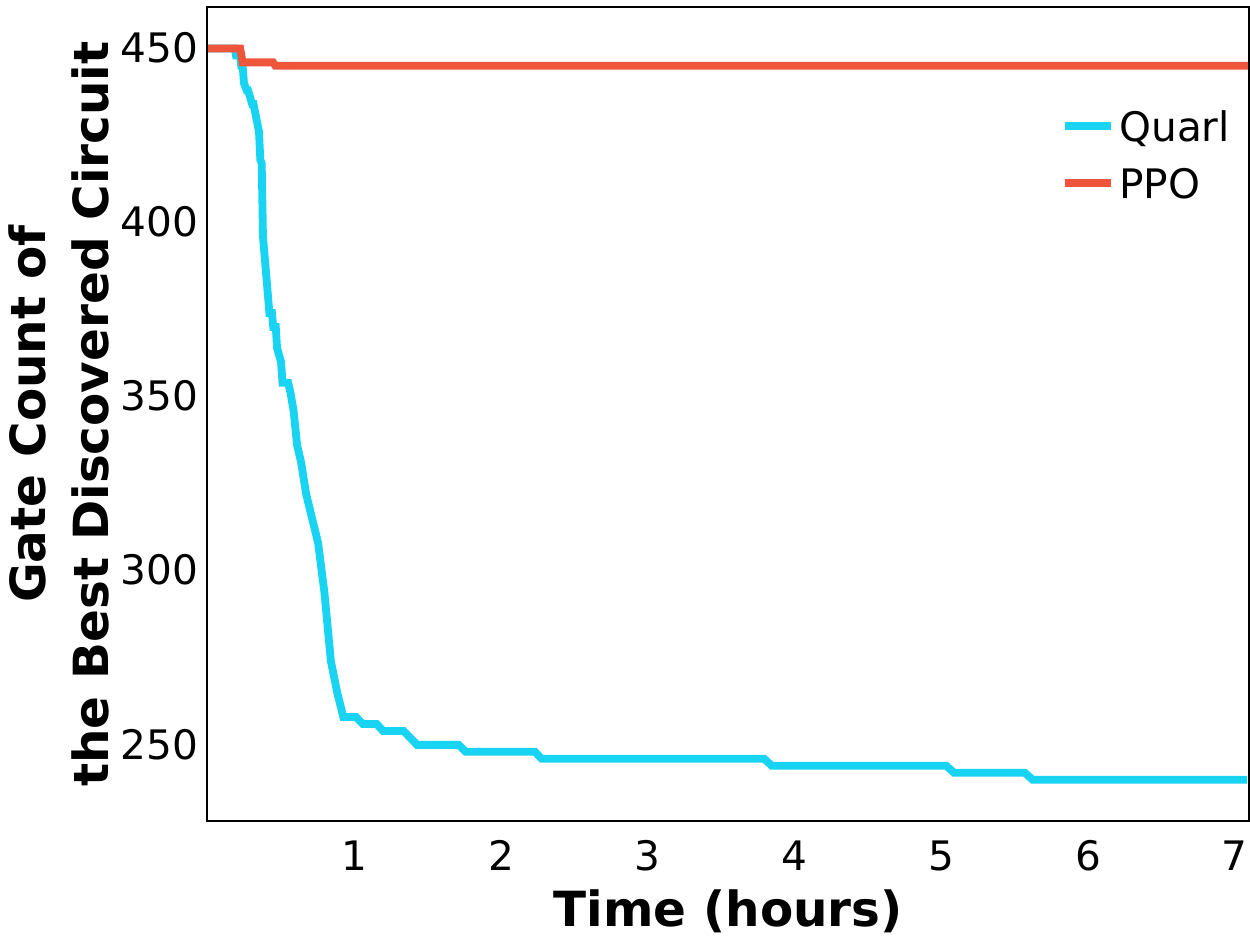}
    \caption{Comparison between our approach and a strawman approach.}
    \label{fig:ppo_abl}
\end{figure}
}

\commentout{
\begin{figure}
    \centering
    \includegraphics[scale=0.36]{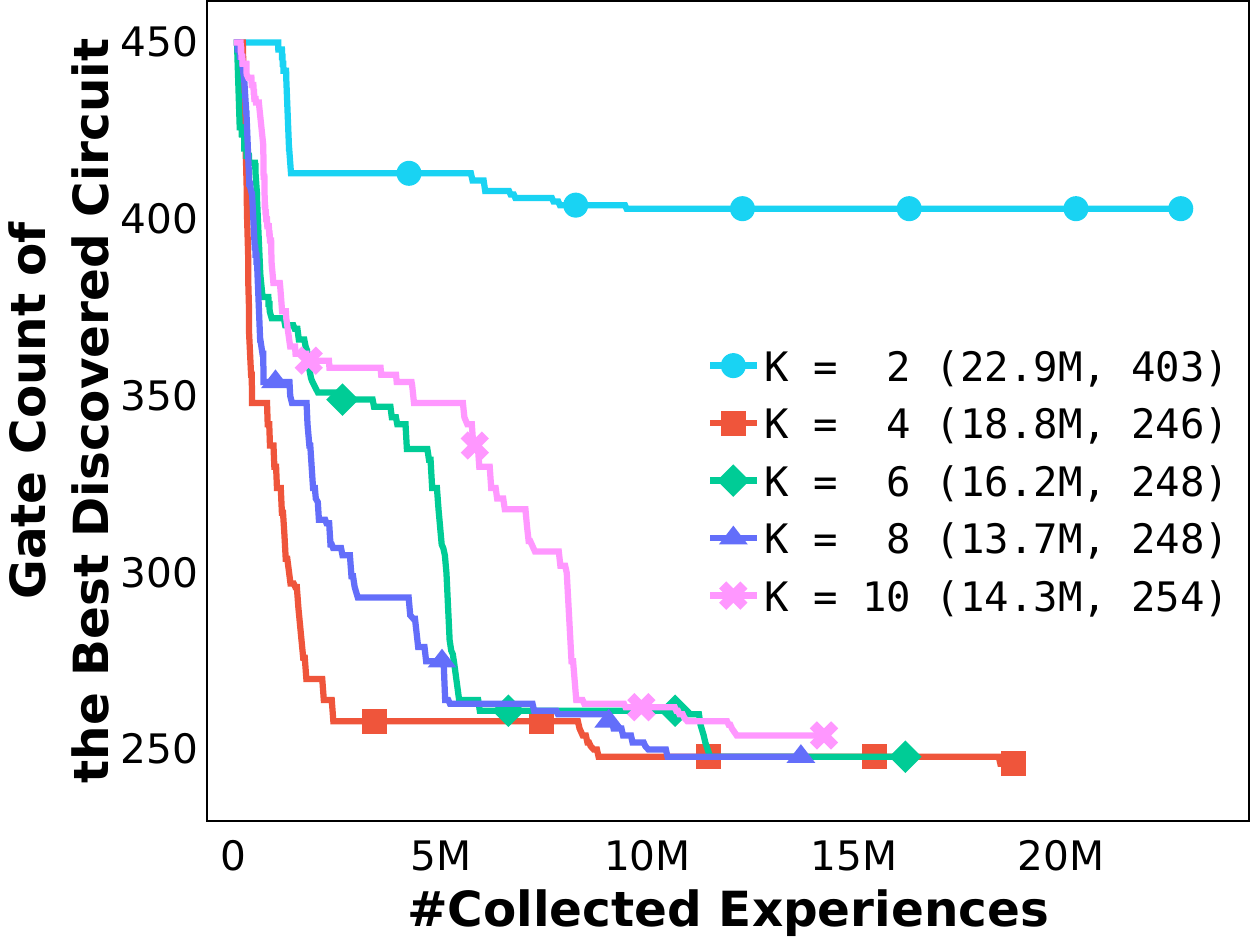}
    \caption{Evaluation of different $K$ values in \sys's graph neural network. Number in parenthesis shows the gate count of the best discovered circuit in each run.}
    \label{fig:khop}
\end{figure}
}

\subsection{Scalability Analysis\label{subsec:scale_evaluation}}

This section investigates the relationship between circuit size and time to collect training data, and evaluates how circuit partitioning improves \sys's scalability.
We evaluate \sys using four circuits with similar structures, namely \texttt{adder\_8}, \texttt{adder\_16}, \texttt{adder\_32}, and \texttt{adder\_64}, which have 900, 1437, 3037, and 6237 gates, respectively. 
\Cref{fig:scale_a} shows the time to collect training data for these circuits, which aligns well with our analysis in \Cref{subsec:scale} and demonstrates a linear relation.
%The time cost of \sys taking a single step during data collection is measured on these circuits, and the results, shown in~\Cref{fig:scale_a}, align with our analysis in~\Cref{subsec:scale} and demonstrate a linear relationship. 
Additionally, \Cref{fig:scale_b} compares the performance of \sys on the original and partitioned circuits, where we partition the input circuit into subcircuits with at most 512 gates based on a topological order.
The experiments were conducted on the same computation platform with a search time budget of 2 hours to ensure a fair comparison and rule out the influence of optimization saturation.
%with circuits partitioned naively based on their topological order into 512-gate partitions. 
The results indicate that circuit partitioning generally increases \sys's performance in limited search time, but it may lead to performance degradation due to missed optimization opportunities across partitions.

\begin{figure}
  % \centering
  \subfloat[Relation between gate count and single-step time cost in data collection.] {
    \label{fig:scale_a}
    \includegraphics[width=0.45\textwidth]{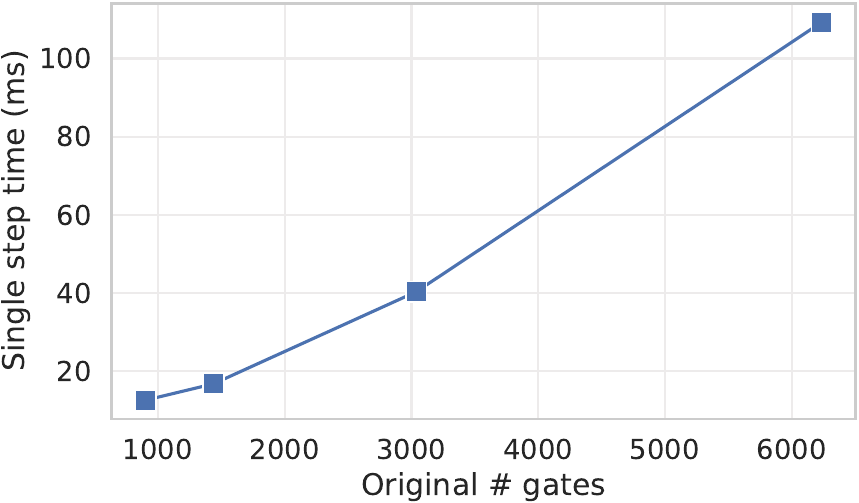}
  }
  \hspace{0.03\textwidth}
  \subfloat[Comparison between \sys's performance with and without circuit partitioning.] {
    \label{fig:scale_b}
    \includegraphics[width=0.45\textwidth]{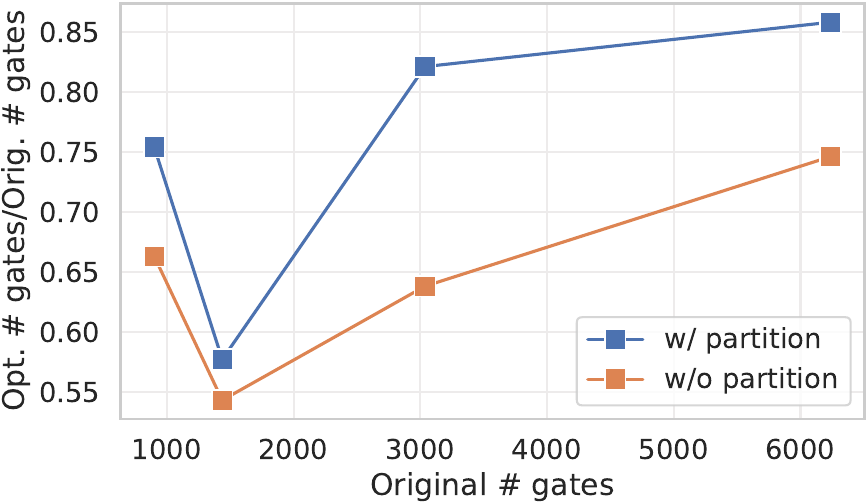}
  }
\vspace{-2mm}
  \caption{Scalability analysis.}
  \label{fig:scale}
\end{figure}

\section{Related Work}
\label{sec:related}
\paragraph{Quantum circuit optimization.} 
% A number of optimizing compilers for quantum circuits have been recently introduced: \qiskit~\cite{qiskit}, \tket~\cite{tket}, \quilc~\cite{quilc}, \voqc~\cite{VOQC}. These circuit optimizers all rely on a rule-based strategy that greedily applies circuit transformations manually designed by quantum experts whenever possible.
% Quartz~\cite{pldi2022-quartz} uses a search-based approach that automatically generates a comprehensive set of circuit transformations and uses a backtracking search algorithm to apply these transformations.
% QUESO~\cite{xu2022synthesizing} uses a path-sum-based approach to synthesizing non-local circuit transformations and beam search to apply the discovered transformations.
% %%%
% \sys was motivated by an important observation that existing rule-based or search-based approaches cannot effectively discover circuit optimizations due to the large search space of functionally equivalent circuits and the necessity of applying cost-decreasing transformations.
% %%%
% \sys addresses this challenge with a novel neural architecture and RL-training procedure. \sys automatically learns to identify circuit optimization opportunities and apply the right transformations by interacting with the environment.

In recent years, several optimizing compilers for quantum circuits have been introduced, such as \qiskit~\cite{qiskit}, \tket~\cite{tket}, \quilc~\cite{quilc}, and \voqc~\cite{VOQC}. These optimizers rely on a rule-based strategy that applies manually designed circuit transformations whenever possible. Quartz~\cite{pldi2022-quartz} uses a search-based approach that generates a comprehensive set of circuit transformations and applies them using a backtracking search algorithm. QUESO~\cite{xu2022synthesizing}, on the other hand, uses a path-sum-based approach to synthesize non-local circuit transformations and apply them using beam search. 
However, the large search space of functionally equivalent circuits and the need for cost-iecreasing transformations pose challenges for existing approaches to effectively discover circuit optimizations. 
\sys addresses this challenge with a novel neural architecture and RL-training procedure, enabling it to automatically identify optimization opportunities and apply the right transformations through interaction with the environment.
%%%
Different from the aforementioned compilers that apply transformations, PyZX~\cite{pyzx} represents circuits as ZX-diagrams and use ZX-calculus~\cite{hadzihasanovic18, jeandel18} to simplify ZX-diagrams, which are eventually converted back to the circuit representation.
%%%
In contrast, \sys is designed to learn to perform transformations. Applying \sys's techniques to optimize circuits in ZX-calculus is a promising avenue for future research.

\paragraph{Reinforcement learning for quantum computing.}
Reinforcement learning has been applied to optimize quantum circuits. For instance, \citet{ostaszewski2021reinforcement} used Double Deep Q-Learning to optimize Variational Quantum Eigensolvers (VQEs) \cite{peruzzo2014variational}. Similarly, \citet{fosel2021quantum} utilized a convolutional neural network and PPO \cite{ppo} to learn circuit transformations. However, this approach is limited to a small set of transformations with up to two qubits. Notably, \sys differs from these methods in its ability to support a comprehensive set of thousands of transformations and optimize arbitrary quantum circuits.

Recently, reinforcement learning has also been applied to the qubit routing problem, which involves inserting {\tt SWAP} gates into the circuits to enable their execution on near-term intermediate-scale quantum (NISQ) devices~\cite{sabre}. Existing examples include QRoute~\cite{sinha2022qubit} and a DQN-based method proposed by~\citet{pozzi2020using}, both of which aim to minimize the final depth of circuits after qubit routing using RL.
It is worth noting that \sys's techniques are orthogonal to RL-based qubit routing techniques, and combining them may lead to further improvements in circuit optimization on NISQ devices.

\section{Conclusion\label{sec:conclusion}}

In this paper, we present \sys, a learning-based quantum circuit optimizer.
%We formulate the quantum circuit optimization problem into a reinforcement learning problem.
%However, applying reinforcement learning to quantum circuit optimization raises two main challenges: the large and varying action space, and the non-uniform state space.
%Based on the locality hypothesis that optimization decisions can be mostly guided by local reasoning, 
%\sys uses a hierarchical approach to combining action space decomposition and graph neural network to address these issues.
In \sys's hierarchical approach, circuits are optimized with a sequence of transformation applications where in each application a gate-selecting policy and a transformation-selecting policy are used to choose a gate and a transformation to apply on the selected gate, respectively.
The two policies are trained jointly, with a single actor-critic architecture.
To apply PPO in the training of this architecture, we propose hierarchical advantage estimation (HAE) as an novel advantage estimator.
Experiment results show that \sys significantly outperforming existing circuit optimizers.

\section*{Acknowledgement}
We thank Shinjae, Yoo and Mingkuan Xu for their helpful feedback. This work is supported by NSF awards CNS-2147909, CNS-2211882, and CNS-2239351, and research awards from Amazon, Google, and Meta.
This research used resources of the National Energy Research Scientific Computing Center (NERSC), a U.S. Department of Energy Office of Science User Facility located at Lawrence Berkeley National Laboratory, operated under Contract No. DE-AC02-05CH11231 using NERSC award DDR-ERCAP0023403.

% To tackle the challenge of huge action space, we propose a hierarchical action space that uses a gate-transformation to uniquely identify an application of transformation.
% Another challenge is how to effectively represent the states (i.e. circuits) that has varying structure and size.
% Based on our hypothesis of the locality of quantum circuit transformations, \sys deals with this challenge by concentrating on fix-sized sub-circuits whose representation is generated with graph neural network.
% Applying reinforcement learning techniques to quantum circuit optimization raises two main challenges: the large and varying action space, and the non-uniform state space.
% \sys addresses these issues with a novel neural architecture and RL-training procedure. Our neural architecture decomposes the action space into two parts, and leverages graph neural networks in its state representation, both of which are guided by the intuition that optimization decisions can be mostly guided by local reasoning while allowing some global reasoning.
% Based on these considerations, we propose a hierarchical approach to optimize quantum circuits by applying transformations, where in each step of the trajectory, a gate is first selected with the gate-selecting policy and 

%\section{Conclusion\label{sec:conclusion}}
%This paper presents \sys, a learning-based quantum circuit optimizer.

% \bibliographystyle{plainnat}
\bibliography{references}

\newpage
\appendix
\section{Proof of Theorem~\ref{thm:influenced_gates}}
\label{prf:1}

\begin{proof}
Let $NG(C, g, x)$ refer to the new gates introduced by $x$ in $C'$.
For gates $g'\in C' \wedge g'\notin\Call{IG}{d, C, g, x}$, there are three cases: 
(1) $g'$ is a $k$-hop predecessor of $NG(C, g, x)$ where $k > d$; 
(2) $g'$ is a successor of $NG(C, g, x)$; 
(3) $g'$ is neither a predecessor nor a successor of $NG(C, g, x)$.
The availability of any transformation $(G, G')$ where the depth of $G$ is less than or equal to $d$ on gate $g$, only depends on the $d$-hop successors of $g$.
For case (1), (2) and (3), the $d$-hop successors of $g'$ remains unchanged.
For case (1), since $g'$ is a $k$-hop predecessor of $NG(C, g, x)$ where $k > d$, then all gates in $NG(C, g, x)$ are outside $g'$'s $d$-hop successors, so $g'$'s $d$-hop successors remain unchanged.
For case (2), all successors of $g'$ are successors of $NG(C, g, x)$, remaining unchanged.
For case (3), since $g'$ is not a predecessor of $NG(C, g, x)$, none of its successors are in $NG(C, g, x)$. Thus, $g'$'s successors remain unchanged.
\end{proof}

\section{Depth Results on the Nam Gate Set}
\begin{table*}[ht]
\centering
\caption{Comparing \sys and existing circuit optimizers on reducing the depth of the benchmark circuits under the Nam gate set. The best result for each circuit is in bold.}
\label{tab:nam_depth}
\setlength{\tabcolsep}{1.1mm}{\fontsize{8}{9}\selectfont
% \resizebox{\textwidth}{!}
% {%
\begin{tabular}{l||r|r|r|r|r|r|r||r|r}
\hline
\multicolumn{10}{c}{\bf Depth (Nam gate set)}\\
\hline
{\bf Circuit} & 
{\bf Orig.} & 
{\bf Nam} & 
{\bf VOQC} & 
{\bf Qiskit} & 
{\bf Tket} & 
{\bf \begin{tabular}{@{}c@{}}Quartz \\ w/ R.M.\end{tabular}} &
{\bf \begin{tabular}{@{}c@{}}Quartz \\ w/o R.M.\end{tabular}} &
{\bf \begin{tabular}{@{}c@{}}\sys \\ w/ R.M.\end{tabular}} &
{\bf \begin{tabular}{@{}c@{}}\sys \\ w/o R.M.\end{tabular}} \\ 
\hline
\texttt{tof\_3} & 33 & 31 & 31 & 33 & 31 & 31 & 32 & 29 & \textbf{25}\\
\texttt{barenco\_tof\_3} & 45 & 38 & 39 & 45 & 42 & 37 & 39 & \textbf{32} & 34\\
\texttt{mod5\_4} & 52 & 46 & 46 & 52 & 51 & 29 & 31 & \textbf{18} & \textbf{18}\\
\texttt{tof\_4} & 54 & 46 & 46 & 54 & 50 & 46 & 45 & \textbf{39} & 40\\
\texttt{barenco\_tof\_4} & 87 & 68 & 70 & 87 & 80 & 64 & 73 & \textbf{56} & \textbf{56}\\
\texttt{tof\_5} & 75 & 61 & 61 & 75 & 69 & 61 & 60 & 48 & \textbf{47}\\
\texttt{mod\_mult\_55} & 55 & 47 & 50 & 54 & 51 & 45 & 53 & \textbf{41} & 48\\
\texttt{vbe\_adder\_3} & 90 & 58 & 60 & 96 & 81 & 60 & 67 & 43 & \textbf{42}\\
\texttt{barenco\_tof\_5} & 129 & 95 & 98 & 129 & 118 & 90 & 114 & \textbf{79} & 87\\
\texttt{csla\_mux\_3} & 72 & 63 & 68 & 84 & 64 & 62 & 63 & 64 & \textbf{61}\\
\texttt{rc\_adder\_6} & 111 & 83 & 95 & 111 & 97 & 90 & 95 & \textbf{74} & 88\\
\texttt{gf2\^{}4\_mult} & 112 & 102 & 102 & 110 & 108 & 105 & 103 & 87 & \textbf{81}\\
\texttt{tof\_10} & 180 & 136 & 136 & 180 & 164 & 136 & 131 & 122 & \textbf{121}\\
\texttt{mod\_red\_21} & 176 & \textbf{129} & 135 & 173 & 150 & 146 & 176 & 131 & 135\\
\texttt{hwb6} & 178 & - & 153 & 177 & 157 & 147 & 173 & 132 & \textbf{130}\\
\texttt{gf2\^{}5\_mult} & 146 & 132 & 131 & 142 & 141 & 135 & 150 & \textbf{114} & 133\\
\texttt{csum\_mux\_9} & 64 & 47 & \textbf{46} & 64 & 61 & 55 & 64 & 53 & 70\\
\texttt{barenco\_tof\_10} & 339 & 230 & 238 & 339 & 308 & 228 & 339 & \textbf{215} & 249\\
\texttt{qcla\_com\_7} & 94 & 70 & 67 & 99 & 88 & 66 & 91 & 76 & \textbf{63}\\
\texttt{ham15-low} & 285 & - & 245 & 283 & 258 & 250 & 285 & 221 & \textbf{209}\\
\texttt{gf2\^{}6\_mult} & 184 & 166 & 166 & 180 & 178 & 165 & 182 & \textbf{153} & 186\\
\texttt{qcla\_adder\_10} & 84 & \textbf{63} & 67 & 89 & 78 & 65 & 84 & 78 & 75\\
\texttt{gf2\^7{}\_mult} & 220 & 198 & 198 & 215 & 213 & \textbf{197} & 218 & 232 & 245\\
\texttt{gf2\^8{}\_mult} & 262 & 236 & 235 & 255 & 253 & \textbf{235} & 262 & 295 & 315\\
\texttt{qcla\_mod\_7} & 229 & \textbf{184} & 199 & 231 & 211 & 195 & 229 & 191 & 205\\
\texttt{adder\_8} & 259 & \textbf{193} & 196 & 264 & 243 & 199 & 255 & 242 & 217\\
\hline
{\bf \begin{tabular}{@{}c@{}}Geo. Mean\\Reduction\end{tabular}} & - & 19.4\% & 17.3\% & -0.7\% & 7.4\% & 19.6\% & 9.1\% & \textbf{25.7\%} & 23.4\% \\
\hline
\end{tabular}
}
\end{table*}

\section{Depth Results on the IBM Gate Set}

\begin{table*}[ht]
\centering
\caption{Comparing \sys and existing circuit optimizers on reducing the depth of the benchmark circuits under the IBM gate set. The best result for each circuit is in bold.}
\label{tab:ibm_depth}
\setlength{\tabcolsep}{1.1mm}{\fontsize{8}{9}\selectfont
% \resizebox{\textwidth}{!}
% {%
\begin{tabular}{l||r|r|r|r||r}
\hline
\multicolumn{6}{c}{\bf Depth (IBM gate set)}\\
\hline
{\bf Circuit} & 
{\bf Orig.} & 
{\bf Qiskit} & 
{\bf Tket} & 
{\bf Quartz} & 
{\bf \sys}  \\
\hline
\texttt{tof\_3} & 37 & 37 & 37 & \textbf{32} & \textbf{32}\\
\texttt{barenco\_tof\_3} & 50 & 50 & 50 & 40 & \textbf{39}\\
\texttt{mod5\_4} & 57 & 57 & 57 & 52 & \textbf{47}\\
\texttt{tof\_4} & 60 & 60 & 60 & \textbf{47} & \textbf{47}\\
\texttt{tof\_5} & 83 & 83 & 83 & 64 & \textbf{61}\\
\texttt{barenco\_tof\_4} & 96 & 96 & 96 & 80 & \textbf{67}\\
\texttt{mod\_mult\_55} & 60 & 58 & 61 & 55 & \textbf{52}\\
\texttt{vbe\_adder\_3} & 98 & 99 & 93 & 73 & \textbf{48}\\
\texttt{barenco\_tof\_5} & 142 & 142 & 142 & 119 & \textbf{94}\\
\texttt{csla\_mux\_3} & 80 & 83 & 75 & 84 & \textbf{65}\\
\texttt{rc\_adder\_6} & 124 & 119 & 125 & 109 & \textbf{92}\\
\texttt{gf2\^{}4\_mult} & 117 & 115 & 114 & 118 & \textbf{95}\\
\texttt{hwb6} & 186 & 183 & 189 & 188 & \textbf{131}\\
\texttt{mod\_red\_21} & 183 & 178 & 201 & 183 & \textbf{128}\\
\texttt{tof\_10} & 198 & 198 & 198 & 149 & \textbf{140}\\
\texttt{gf2\^{}5\_mult} & 148 & 144 & 143 & 149 & \textbf{138}\\
\texttt{csum\_mux\_9} & 68 & \textbf{68} & 77 & \textbf{68} & 84\\
\texttt{barenco\_tof\_10} & 372 & 372 & 372 & 376 & \textbf{249}\\
\texttt{ham15-low} & 310 & 298 & 288 & 300 & \textbf{235}\\
\texttt{qcla\_com\_7} & 98 & 100 & 97 & 95 & \textbf{78}\\
\texttt{gf2\^{}6\_mult} & 189 & 185 & \textbf{184} & 189 & 197\\
\texttt{qcla\_adder\_10} & 89 & 92 & 84 & 89 & \textbf{77}\\
\texttt{gf2\^{}7\_mult} & 225 & 220 & \textbf{219} & 225 & 248\\
\texttt{gf2\^{}8\_mult} & 267 & 260 & \textbf{259} & 268 & 312\\
\texttt{qcla\_mod\_7} & 241 & 241 & \textbf{235} & 238 & 250\\
\texttt{adder\_8} & 278 & 274 & 318 & 275 & \textbf{208}\\
\texttt{vqe\_8} & 69 & \textbf{20} & 27 & 24 & 28\\
\texttt{qgan\_8} & 60 & 31 & \textbf{30} & 45 & 34\\
\texttt{qaoa\_8} & 56 & 45 & \textbf{42} & 48 & 51\\
\texttt{ae\_8} & 193 & 143 & \textbf{136} & 175 & 142\\
\texttt{qpeexact\_8} & 199 & 140 & \textbf{132} & 160 & 149\\
\texttt{qpeinexact\_8} & 212 & \textbf{144} & 146 & 165 & 156\\
\texttt{qft\_8} & 170 & 115 & \textbf{108} & 131 & 120\\
\texttt{qftentangled\_8} & 184 & \textbf{121} & 153 & 162 & 129\\
\texttt{portfoliovqe\_8} & 151 & 67 & \textbf{61} & 86 & 100\\
\texttt{portfolioqaoa\_8} & 300 & 232 & \textbf{216} & 288 & 305\\
\hline
{\bf \begin{tabular}{@{}c@{}}Geo. Mean\\Reduction\end{tabular}} & - & 13.5\% & 12.9\% & 13.4\% & \textbf{22.1}\% \\
\hline
\end{tabular}
}
\end{table*}

% \section{Rotation Merging as Circuit Transformations\label{sec:rotation_merging}}
% Rotation merging is a circuit optimization commonly used in existing circuit optimizers. It merges $R_z$ gates with identical phase polynomial expressions.
% %%%
% The $R_z$ gates being merged may be arbitrarily far apart, making it difficult to represent rotation merging as combinations of local circuit transformations.
% Existing optimizers rely on a manual implementation of rotation merging.
% However, for some circuits, \sys can automatically discover optimizations similar to rotation merging from its own exploration. \Cref{fig:rotation_merging_transfers} shows a sequence of transformations discovered by \sys that merges the $T$ and $T^\dag$ gates with identical phase, resulting in an optimization process similar to rotation merging.

\end{document}